\documentclass[%
 reprint,
 superscriptaddress,
 amsmath,amssymb,
 aps,
]{revtex4-2}
\usepackage[T1]{fontenc}
\usepackage{graphicx}
\usepackage{dcolumn}
\usepackage{bm}
\usepackage[colorlinks=true, allcolors=blue]{hyperref}
\usepackage{braket}
\usepackage[table]{xcolor}
\usepackage{xfrac}
\usepackage{comment}
\usepackage[normalem]{ulem}
\usepackage{mathtools}
\usepackage{layouts}
\usepackage{quantikz}
\usepackage{tabularx}
\usepackage{multirow}
\usepackage{esint}
\usepackage{enumitem}
\usepackage{soul}

\usepackage{amsthm}
\usepackage[capitalise]{cleveref}

\newcommand{\approptoinn}[2]{\mathrel{\vcenter{
  \offinterlineskip\halign{\hfil$##$\cr
    #1\propto\cr\noalign{\kern2pt}#1\sim\cr\noalign{\kern-2pt}}}}}

\newcommand{\sectiontitle}[1]{\textbf{#1:}}
\begin{document}
\title{Lowering Connectivity Requirements For Bivariate Bicycle Codes Using Morphing Circuits}
\author{Mackenzie H. Shaw}
\affiliation{QuTech, Delft University of Technology, Lorentzweg 1, 2628 CJ Delft, The Netherlands}
\affiliation{Delft Institute of Applied Mathematics, Delft University of Technology, Mekelweg 4, 2628 CD Delft, The Netherlands}
\author{Barbara M. Terhal}
\affiliation{QuTech, Delft University of Technology, Lorentzweg 1, 2628 CJ Delft, The Netherlands}
\affiliation{Delft Institute of Applied Mathematics, Delft University of Technology, Mekelweg 4, 2628 CD Delft, The Netherlands}
\newtheorem{theorem}{Theorem}[section]
\newtheorem{corollary}[theorem]{Corollary}
\newtheorem{lemma}[theorem]{Lemma}
\newtheorem{proposition}[theorem]{Proposition}
\newtheorem{criterion}{Criterion}
\newtheorem{appendixcriterion}{Criterion}[section]
\theoremstyle{definition}
\newtheorem{assumption}{Assumption}
\theoremstyle{remark}
\newtheorem*{remark}{Remark}
\theoremstyle{definition}
\newtheorem{definition}[theorem]{Definition}
\Crefname{criterion}{Criterion}{Criteria}
\crefname{criterion}{Crit.}{Crit.}
\begin{abstract}
In Ref.~\cite{Bravyi24}, Bravyi \textit{et al.}~found examples of Bivariate Bicycle (BB) codes with similar logical performance to the surface code but with an improved encoding rate.
In this work, we generalize a novel parity-check circuit design principle called \textit{morphing circuits} and apply it to BB codes. We define a new family of BB codes
whose parity check circuits require a qubit connectivity of degree five instead of six while maintaining their numerical performance. 
Logical input/output to an ancillary surface code is also possible in a biplanar layout.
Finally, we develop a general framework for designing morphing circuits and present a sufficient condition for its applicability to two-block group algebra codes.
\end{abstract}
\maketitle

\sectiontitle{Introduction}
Quantum error correction (QEC) is crucial for achieving fault-tolerant universal quantum computation. One of the most widely-studied QEC codes is the surface code~\cite{Kitaev03,bravyi1998,Cleland22}, whose strengths include its planar qubit connectivity and good performance at experimentally achievable error rates. However, since only one logical qubit is encoded in each surface code patch, the qubit overhead becomes extremely high at the low error rates required for practical quantum algorithms.

One alternative approach is to use low-density parity check (LDPC) codes which encode more than one logical qubit per code block, at the expense of no longer having a purely planar connectivity. 
Recently, Bravyi \textit{et al.}~\cite{Bravyi24} introduced a set of LDPC codes called bivariate bicycle (BB) codes --- a subset of the more general two-block group algebra (2BGA) codes~\cite{Kovalev13,Lin23}--- that, for the first time, match the logical performance of surface codes even at relatively high physical error rates. As was noted in Ref.~\cite{Bravyi24}, physically implementing these BB codes using, for example, superconducting qubits, presents an additional experimental challenge: each qubit performs a CNOT gate with six other qubits in a biplanar layout.

In this work, we simplify the experimental requirements by constructing a set of closely related BB codes whose physical implementation only requires each qubit to interact with five other qubits in a biplanar layout. The new codes are designed using a recently proposed parity check circuit design philosophy that we refer to as \textit{morphing}. Originally called ``middle-out'' circuits, these morphing circuits have been applied to both surface codes and color codes to reduce the connectivity requirements in those codes \cite{McEwen23,Gidney23}. We choose the name morphing circuits as the procedure is also related to the concept of \textit{morphing} quantum codes from Ref.~\cite{Vasmer22}.

Our contribution is to generalise the idea of morphing circuits to generate parity check circuits for general codes and apply this methodology to the BB codes in Ref.~\cite{Bravyi24}. The procedure takes as input a known code $C$ and outputs a pair of new ``end-cycle'' codes $\tilde{C}_{1}$ and $\tilde{C}_{2}$ along with a pair of parity check circuits. Each parity check circuit measures all of the stabilizer generators of one end-cycle code $\tilde{C}_{i}$ while simultaneously transforming into the other code $\tilde{C}_{i'}$, see Fig.~\ref{fig:overview}. Moreover, mid-way through the parity check circuit the joint state of the data and ancilla qubits is encoded in the original, known code $C$. Despite this, the end-cycle codes may bear little resemblance to the codes from which they are derived --- indeed, in our case, the weight of the stabilizers of the end-cycle codes is nine instead of six.

To demonstrate the practicality of our new codes, we investigate their performance against uniform circuit-level depolarising noise using the BP-OSD decoder \cite{Panteleev21,Roffe20}. We find that the new codes perform at least as well as those in Ref.~\cite{Bravyi24}, and therefore provide the same overhead savings versus the surface code as those in Ref.~\cite{Bravyi24}. Moreover, we demonstrate that the input and output (I/O) of arbitrary logical qubits from the new codes to the surface code is possible using morphing circuits with a biplanar graph layout. To the best of our knowledge, our I/O construction is the first time morphing circuits have been used to perform a lattice surgery operation.

\sectiontitle{Mid- and End-cycle Codes} We begin by introducing some terminology (following Ref.~\cite{McEwen23}) for an arbitrary parity check circuit that measures the stabilizer generators of the code. Each round of parity checks begins and ends with a measurement and reset of a set of ancilla qubits. During this time, we say that the remaining data qubits are encoded in the \textit{end-cycle} code. Next, a circuit of Clifford gates is performed. At each step during this circuit, the data and ancilla qubits are also encoded in some QEC code. In particular, we define the \textit{mid-cycle} code as the code that arises precisely mid-way through the circuit~\footnote{Strictly speaking of course the ``mid-way'' point of a parity check circuit is only well-defined if the circuit has an even depth --- which is satisfied for all the new circuits constructed in this Letter.}. The stabilizer generators of this mid-cycle code can be determined using the Gottesman-Knill theorem~\cite{Aaronson04} and originate from two sources: the stabilizers of the end-cycle code, and the reset of the ancilla qubits at the start of the QEC cycle. This latter set of stabilizers ensures that the mid-cycle code has the same number of logical qubits as the end-cycle code despite being encoded across both the data and ancilla qubits.


\sectiontitle{QEC through Morphing Circuits}
The standard approach to designing a parity check circuit is to assume that the end-cycle code corresponds to some specified, ``known'' code. In contrast, to design a morphing parity check circuit we instead assume that the known code corresponds to the \textit{mid-cycle} code of the circuit, while the end-cycle code is yet-to-be-determined.

More precisely, the morphing construction is defined by a series of \textit{contraction circuits} $F_{i}$. Each contraction circuit $F_{i}$ must be a Clifford circuit that takes a subset $S_{i}\subseteq S$ of the stabilizer generators of the known, mid-cycle code $C$ and contracts each generator onto a single qubit. We call the stabilizers in $S_{i}$ \textit{contracting} stabilizers. The circuit $F_{i}$ should \textit{not} use any additional qubits since the mid-cycle code $C$ is already encoded across all the data and ancilla qubits. Subsequently, each qubit that hosts a contracted stabilizer is measured in the $X$- or $Z$-basis to reveal the eigenvalue of the contracted stabilizer. We label this set of measurements $M_{i}$. At this point, the remaining non-measured qubits are encoded in the end-cycle code $\tilde{C}_{i}$, which is a new code determined by propagating the stabilizers of the known code $C$ through the circuit $M_{i}\circ F_{i}$. Finally, we restore the mid-cycle code $C$ by first resetting all the measured qubits (labeled $R_{i}$) and then applying the inverse circuit $F_{i}^{\dagger}$.

A set of contracting circuits $F_{i}$ defines a valid morphing (parity check) protocol if every stabilizer generator is contained in at least one of the contracting sets $S_{i}$, i.e. $\bigcup_{i}S_{i}=S$. When this is the case, we can use the morphing protocol to implement a parity check schedule of the \textit{end-cycle} code $\tilde{C}_{i}$, that simultaneously transforms the code into $\tilde{C}_{i+1}$ after each QEC round. The end-cycle codes $\tilde{C}_{i}$ are new and their parameters are $[[\tilde{n}_{i},k,\tilde{d}_{i}]]$ when the known code $C$ has parameters $[[n,k,d]]$. It is guaranteed that $\tilde{n}_{i}<n$ and, typically, the distance $\tilde{d}_{i} \leq d$; 
in \cite{supp} we give a simple lower-bound on $\tilde{d}_{i}$ given the circuits $F_{i}$. 

For the codes considered in this Letter, we only need two contracting circuits $F_{1}$ and $F_{2}$, with each contracting subset $S_{i}$ containing half of the generators in $S$. In this case, the parity check schedule from $\tilde{C}_{1}\rightarrow \tilde{C}_{2}$ measures all of the stabilizers of $\tilde{C}_{1}$, and vice versa.
In \cref{fig:overview} we summarise the practical operation of such a morphing protocol.
The parity check circuit for $\tilde{C}_{1}$ is given by $M^{\vphantom{\dag}}_{2}\circ F^{\vphantom{\dag}}_{2}\circ F_{1}^{\dag}\circ R_{1}^{\vphantom{\dag}}$, after which we are in the code $\tilde{C}_{2}$. Then, the parity check circuit for $\tilde{C}_{2}$ is $M^{\vphantom{\dag}}_{1}\circ F^{\vphantom{\dag}}_{1}\circ F_{2}^{\dag}\circ R_{2}^{\vphantom{\dag}}$, which returns us back to $\tilde{C}_{1}$.

\begin{figure}
\includegraphics[width=\linewidth]{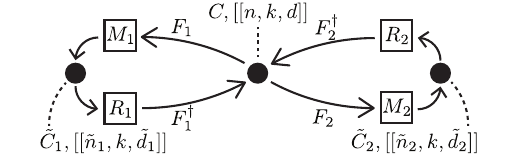}
\caption{Graphical summary of the operation of the morphing protocol, including the known mid-cycle code $C$, the end-cycle codes $\tilde{C}_{i}$, the contraction circuits $F_{i}$, and the measurement and reset rounds $M_{i}$, $R_{i}$.}
\label{fig:overview}
\end{figure}

\sectiontitle{Weight-6 Abelian 2BGA Codes}
We consider weight-6 Abelian 2BGA codes which includes the BB codes studied in Ref.~\cite{Bravyi24} --- see \cite{supp} for a generalisation to all 2BGA codes. Each code is defined by an Abelian group $G$ and two sets of group elements $A=\{a_{1},a_{2},a_{3}\}$ and $B=\{b_{1},b_{2},b_{3}\}$.
The code is defined on $n=2|G|$ physical qubits labeled $q(L,g)$ or $q(R,g)$ for $g\in G$, with $L$ and $R$ standing for ``left'' and ``right'' qubits respectively. We write $X(P,Q)$ for the product of $X$ operators on the left qubits with labels in the subset $P\subseteq G$ and on right qubits $Q\subseteq G$, and similarly for $Z(P,Q)$ for a $Z$ operator. The stabilizer generators of the code are then given by $s(X,g)=X(Ag,Bg)$ and $s(Z,g)=Z(B^{-1}g,A^{-1}g)$ for $g\in G$, both of which have weight $w=|A|+|B|=6$. Here, for any subset $H\subseteq G$, we interpret inverse and multiplication element-wise, i.e.~$Hg=\{hg\mid h\in H\}$ and $H^{-1}=\{h^{-1}\mid h\in H\}$. Multiplying $A$ or $B$ by a group element leaves the code invariant~\cite{Lin23}; so without loss of generality we can assume that $a_{1}=b_{1}=1$.

In Ref.~\cite{Bravyi24}, the authors find a number of examples of weight-6 Abelian 2BGA codes that achieve comparable circuit-level performance to the surface code, listed in~\cref{tab:code_parameters}. Each code has $G=\mathbb{Z}_{\ell}\times\mathbb{Z}_{m}$ for positive integers $\ell,m$. The syndrome extraction schedule for this family of codes is highly optimized and requires seven rounds of CNOTs in total during which each qubit interacts with six other qubits. The Tanner graph of the codes --- the bipartite graph with qubits and checks as nodes and an edge between a qubit and a check if the check acts on the qubit --- is biplanar, meaning the edges can be split into two subsets each of which forms a planar graph.

\begin{table*}
    \caption{Table of BB codes from \cite{Bravyi24}, with the possible choices of homomorphism $f_x, f_y$ or $f_{xy}$ [\cref{eq:homomorphisms}] that satisfy Criterion \ref{crit:homomorphism}, including the code parameters and circuit-level distance of the corresponding standard parity check schedule. When at least one homomorphism exists that satisfies Crit.~\ref{crit:homomorphism}, we list the code parameters of the end-cycle codes and the circuit-level distance upper-bound of the morphing circuits, which are the same for all choices of $f$ and for both end-cycle codes. Note here that $n$ refers to the number of qubits in the BB code from \cite{Bravyi24}, which includes \textit{both} data \textit{and} ancilla qubits when ran as a morphing circuit, but \textit{only} includes data qubits when ran using the circuits from~\cite{Bravyi24}.}
    \label{tab:code_parameters}
    \renewcommand{\arraystretch}{1.35}
    \setlength{\tabcolsep}{6pt}
    \centering
\begin{tabular}{|c|c|c|c|c|c|c|c|}
        \hline
        \multicolumn{4}{|c|}{Code Definition} & \multicolumn{2}{c|}{BB Code $C$ \cite{Bravyi24} }& \multicolumn{2}{c|}{End-cycle code $\tilde{C}_{i}$} \\
        \hline
        $\ell$, $m$ & $A$ & $B$ & $f$ & $[[n,k,d]]$ & $d_{\text{circ}}$ & $[[\tilde{n},k,\tilde{d}]]$ & $\tilde{d}_{\text{circ}}$ \\\hline
        6, 6 & $\{x^3,y,y^2\}$ & $\{y^3,x,x^2\}$ &  $f_x,f_y,f_{xy}$ &$[[72,12,6]]$& $\leq 6$ & $[[36,12,3]]$ & $\leq 3$ \\ \hline 
        9, 6 & $\{x^3,y,y^2\}$ & $\{y^3,x,x^2\}$ & $f_y$ & $[[108,8,10]]$& $\leq 8$& $[[54,8,8]]$ & $\leq 7$ \\  \hline
        12, 6 & $\{x^3,y,y^2\}$ & $\{y^3,x,x^2\}$ & $f_x,f_y,f_{xy}$ & $[[144,12,12]]$ & $\leq 10$ & $[[72,12,6]]$ & $\leq 6$ \\ \hline
        12, 12 & $\{x^3,y^7,y^2\}$& $\{y^3,x,x^2\}$ & $f_x,f_y, f_{xy}$ & $[[288,12,18]]$& $\leq 18$ & $[[144,12,12]]$ & $\leq 12$\\ \hline
    \end{tabular}
\end{table*}

\sectiontitle{Applying the Morphing Protocol}
We now show how to construct a pair of contraction circuits $F_1$ and $F_2$ that measure all the stabilizers of a given weight-6 Abelian 2BGA code, whenever the code satisfies the following:
\begin{criterion}\label{crit:homomorphism}
    There exists a group homomorphism $f:G\rightarrow\mathbb{Z}_{2}$ with the property that $f(a_{1})\neq f(a_{2})=f(a_{3})$ and $f(b_{1})\neq f(b_{2})=f(b_{3})$.
\end{criterion}
Under the assumption that $a_{1}=b_{1}=1$, Criterion~\ref{crit:homomorphism} becomes simply $f(a_{2})=f(a_{3})=f(b_{2})=f(b_{3})=u$, where we write $\mathbb{Z}_{2}=\{1,u\}$ with $u^{2}=1$. Since any group homomorphism obeys $f(xy)=f(x)f(y)$, $f$ can be uniquely specified by how it acts upon the generators of $G$. As such, when $G=\mathbb{Z}_{\ell}\times\mathbb{Z}_{m}$, there are at most three possible choices of homomorphism that could satisfy Crit.~\ref{crit:homomorphism}, given by
\begin{subequations}\label{eq:homomorphisms}
    \begin{align}
    f_x(x)&=u,& f_x(y)&=1,& \text{if }\ell&\equiv0\,\mathrm{mod}\,2, \label{eq:x_homomorphism}\\
    f_y(x)&=1,& f_y(y)&=u,& \text{if }m&\equiv0\,\mathrm{mod}\,2, \label{eq:y_homomorphism}\\
    f_{xy}(x)&=u,& f_{xy}(y)&=u, & \text{if }\ell\equiv m&\equiv0\,\mathrm{mod}\,2.\label{eq:xy_homomorphism}
    \end{align}
\end{subequations}
Note that the conditions on $\ell$ and $m$ are necessary to ensure that the function $f$ is a group homomorphism of $G$.

The BB codes from Ref.~\cite{Bravyi24} are listed in \cref{tab:code_parameters}, along with the possible choices of homomorphism in \cref{eq:homomorphisms} that satisfy Crit.~\ref{crit:homomorphism}~\footnote{Note that different homomorphisms applied to the same mid-cycle BB code can indeed lead to distinct end-cycle codes that are unrelated via the mappings of Ref.~\cite{Lin23}.}. When a code satisfies Crit.~\ref{crit:homomorphism}, we define the two cosets $K=\ker{f}=\{g\mid f(g)=1\}$ and $K^{c}=G\setminus K=\{g\mid f(g)=u\}$. Moreover, we call qubits and stabilizers ``even'' if they're labeled by an element $g\in K$\, and ``odd'' if $g\in K^{c}$.

\begin{table}
    \caption{Definition of the contracting circuits $F_{i}$ and measurements $M_{i}$, assuming that Crit.~\ref{crit:homomorphism} is satisfied. CNOT$(q_{1},q_{2})$ indicates a CNOT gate with control $q_{1}$ and target $q_{2}$ and $M_{P}$ represents a measurement in the $P$-basis. The circuit $F_{1}$ (respectively, $F_{2}$) is defined by applying the gates below for each $g\in K$ ($g\in K^{c}$) and $h\in K^{c}$ ($h\in K$).}\label{tab:contraction_circuits}
    \renewcommand{\arraystretch}{1.35}
    \setlength{\tabcolsep}{6pt}
    \begin{tabular}{|c|c|}
        \hline
        \multirow{2}{*}{Round 1}& CNOT$\big(q(L,g),q(R,b_{3}g)\big),$\\
        & CNOT$\big(q(R,a_{3}^{-1}h),q(L,h)\big)$\\\hline
        \multirow{2}{*}{Round 2}& CNOT$\big(q(L,g),q(R,b_{2}g)\big),$\\
        & CNOT$\big(q(R,a_{2}^{-1}h),q(L,h)\big)$\\
        \hline
        \multirow{2}{*}{Round 3}& CNOT$\big(q(R,g),q(L,g)\big),$\\
        &CNOT$\big(q(L,h),q(R,h)\big)$\\\hline
        Round 4& $M_{X}\big(q(R,g)\big)$, $M_{Z}\big(q(R,h)\big)$\\
        \hline
    \end{tabular}
\end{table}

\begin{figure}
\centering
\includegraphics[width=\linewidth]{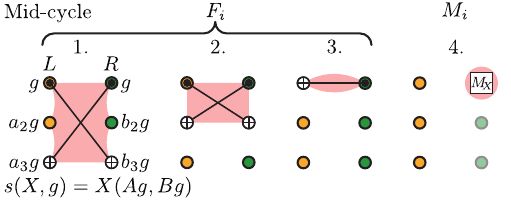}
\caption{Visual representation of a contracting $X$-stabilizer, beginning with the mid-cycle support of $s(X,g)$, then, the three steps of the contraction circuit $F_{i}$, and the measurement step $M_{i}$, as given in \cref{tab:contraction_circuits}. The support of the stabilizer before each step is shown in red.}
\label{fig:contracting_X}
\end{figure}

The morphing circuit is defined through the contracting circuits $F_{i}$, measurements $M_{i}$, and resets $R_{i}$ listed in \cref{tab:contraction_circuits}. In particular, $F_{i}$ consists of three rounds of CNOTs, such that the total CNOT depth of the parity check circuit $F_{2}^{\vphantom{-1}}\circ F_{1}^{-1}$ is six. The contracting stabilizers in the set $S_{1}$ are all even $X$-stabilizers $s(X,g)$ ($g\in K$) and the odd $Z$-stabilizers $s(Z,g)$ ($g\in K^{c}$), while the contracting stabilizers in $S_{2}$ are the odd $X$- and even $Z$-stabilizers. Both measurements $M_{1}$ and $M_{2}$ measure all of the right qubits of $C$, in the $X$- or $Z$-basis depending on whether the qubit is even or odd. One can see in \cref{fig:contracting_X} that the contracting stabilizers are indeed contracted and measured under the circuit $M_{i}\circ F_{i}$, we show this formally in \cite{supp}.


\sectiontitle{Connectivity}
One advantage of the morphing protocol is that the connectivity graph of the circuits --- with vertices for each qubit and edges between qubits that participate in a CNOT --- has degree 5, one fewer than the degree of the circuits in Ref.~\cite{Bravyi24}. Indeed, by considering \cref{tab:contraction_circuits} for both circuits $F_{1}$ and $F_{2}$, we see that the connectivity graph is bipartite between the left and right qubits, with edges $\{q(L,g),q(R,a_{i}^{-1}b_{j}^{\vphantom{-1}}g)\}$ for all $g\in G$ and $(i,j)\in\{(1,1),(1,2),(1,3),(2,1),(3,1)\}$. Moreover, we explicitly prove that the connectivity graph is biplanar in \cite{supp}. For each of the codes listed in \cref{tab:code_parameters}, the standard circuit can be implemented in the ``toric$^+$'' layout using the four short-range connections of the toric code plus two long-range edges. Meanwhile, the morphing protocol requires the three short-range connections of the hex-grid rotated toric code~\cite{McEwen23} plus two long-range edges, see \cite{supp} for more details.

\sectiontitle{The End-cycle Codes}
The stabilizers that are \textit{not} contracting are called \textit{expanding} stabilizers, and these form a set of stabilizer generators for the end-cycle code $\tilde{C}_{i}$. Each end-cycle code has support only on the left qubits of $C$, and has stabilizer generators of the form $X(ABg,0)$ and $Z(A^{-1}B^{-1}h,0)$, where $g\in K$ and $h\in K^{c}$ for the code $\tilde{C}_{1}$ (and vice versa for $\tilde{C}_{2}$) --- see Fig.~4 of \cite{supp}. 
Here we have inherited the multiplication operation from the group algebra representation $\mathbb{Z}_{2}[G]$ of the subsets $A$ and $B$; explicitly, the product $AB$ is simply the set $\{a_{i}b_{j}\}$ unless
some $a_{i}b_{j}=a_{i'}b_{j'}$, in which case these two elements are removed from the set. Thus, if each of the products $a_{i}b_{j}$ is unique, each end-cycle stabilizer generator has weight nine. Moreover, the stabilizer groups of the codes $\tilde{C}_{1}$ and $\tilde{C}_{2}$ are identical up to a shift of the qubits by any element $s\in K^{c}$. Finally, each end-cycle code can be rewritten as a 2BGA code by identifying the end-cycle-left and end-cycle-right qubits as the even and odd left qubits, as shown in \cite{supp}.

We summarise the parameters of the end-cycle codes alongside their corresponding mid-cycle BB codes in \cref{tab:code_parameters}. The distance of the end-cycle codes was calculated using a linear binary integer program~\cite{Gurobi}, following the methods of Ref.~\cite{Landahl11}. Perhaps coincidentally, the end-cycle codes often have the same parameters as a different BB code. For example, the end-cycle code derived from the $[[288,12,18]]$ BB code has parameters $[[144,12,12]]$, presenting a second ``gross'' code that could be targeted by future experiments.

\begin{figure}
    \centering
    \includegraphics[width = \linewidth]{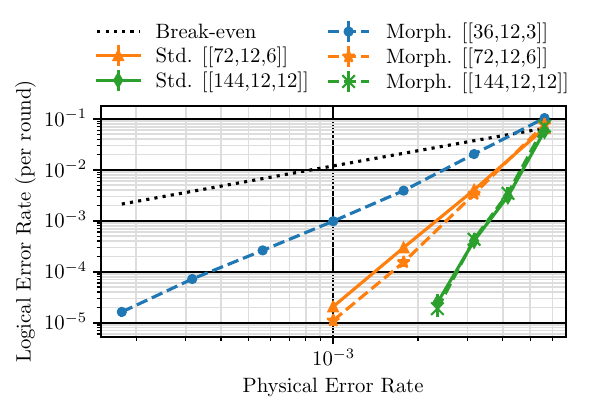}\hfill
    \caption{Numerical logical performance of codes from Ref.~\cite{Bravyi24} under standard parity check circuits (Std.) and the new codes designed from morphing circuits (Morph.) with respect to a uniform circuit-level depolarizing noise model and decoded using BP-OSD; $[[n,k,d]]$ here refers to the \textit{end-cycle} code parameters. The break-even line represents the logical error rate of 12 bare physical qubits.
    }\label{fig:numerics}
\end{figure}

\sectiontitle{Circuit-level Performance}
We have estimated an upper-bound of the circuit-level distance $d_{\rm circ}$ of the morphing parity check protocol using the BP-OSD decoder~\cite{Panteleev21,Roffe20}, following the methods of Ref.~\cite{Bravyi24}.
In \cref{fig:numerics} we numerically simulated the performance of each $k=12$ parity check circuit in \cref{tab:code_parameters} under a uniform circuit-level depolarizing noise model using the BP-OSD decoder, see \cite{supp} for the parameters used and data for more codes. When the code parameters of the BB code and the end-cycle code match, we see that they perform similarly under this circuit-level noise model, demonstrating that the morphing protocol is a viable alternative to the circuits in Ref.~\cite{Bravyi24}.

\sectiontitle{Logical Operations}
In Ref.~\cite{Bravyi24}, Bravyi \text{et al.} show how to perform the logical input and output (I/O) between an arbitrary logical qubit of the BB code and a surface code ancilla system. This involves operations within the BB code using only already-existing connections, as well as the addition of a separate linking code~\cite{Cohen22} that preserves the biplanarity of the code. In \cite{supp} we show that, under some loose assumptions about the structure of the logical operators, the I/O of arbitrary logical qubits is also possible in a biplanar layout with the morphing protocol. We show how to perform shift automorphisms within the BB code without using additional connections~\footnote{We leave investigations into the $ZX$-duality to future work, since even for the codes of Ref.~\cite{Bravyi24}, more work must be done to optimize the gate sequence before it can be implemented in practice.}. Moreover, we explain how to perform general lattice-surgery-like operations within the framework of the morphing protocol and apply this to the logical I/O between the BB and linking code. 


\sectiontitle{Modifications to $F_{i}$}
The morphing protocol defined in \cref{tab:contraction_circuits} is far from unique. For example, one can reverse the direction of some of the CNOTs in~\cref{tab:contraction_circuits} and still obtain a valid contraction circuit for the BB code. In \cite{supp} we show that a reversal of the CNOTs in Round 3 can be used to swap the data and ancilla qubits in each round of QEC without compromising the parameters of the end-cycle codes. This could be used experimentally to mitigate the effects of leakage~\cite{McEwen23}.

A less trivial modification involves reversing the CNOTs in Round 2. In \cite{supp} we show that the resulting end-cycle codes $\tilde{C}_{i}$ are \textit{not} equivalent to those derived from \cref{tab:contraction_circuits}, and we identify three circuits that have a larger distance and circuit-level distance upper-bound than the corresponding protocols presented in~\cref{tab:code_parameters}. However, their numerical performance against circuit-level noise does not improve, as shown in~\cite{supp}. We leave further investigation to future work.


\sectiontitle{Outlook}
In this work we have developed a general framework to design morphing protocols for arbitrary codes and applied this framework to the BB codes presented in Ref.~\cite{Bravyi24}. Similarly to surface codes and colour codes~\cite{McEwen23,Gidney23}, these new parity check circuits reduce the degree of the connectivity graph and allow for the swapping of data and ancilla qubits between each round. Moreover, these advantages are achieved without sacrificing the biplanarity of the connectivity graph, numerical performance, or logical capabilities of the original BB codes. An exciting area of future research is therefore to apply the morphing construction to more codes; for example, BB codes that do \textit{not} satisfy Criterion~\ref{crit:homomorphism}, as well as other classes of codes such as hypergraph product codes and higher-dimensional topological codes.

\sectiontitle{Acknowledgements}
This work is supported by QuTech NWO funding 2020-2024 – Part I “Fundamental Research”, project number 601.QT.001-1, financed by the Dutch Research Council (NWO). B.M.T. thanks the OpenSuperQPlus100 project (no. 101113946) of the EU Flagship on Quantum Technology (HORIZON- CL4-2022-QUANTUM-01-SGA) for support.
We acknowledge the use of the DelftBlue supercomputer for running the decoding.
We thank Marc Serra Peralta for the helpful conversations and assistance in running the decoding.

\appendix
\section{Numerical Simulations}\label{sec:numerics}

Here, we provide the details of how we obtained the numerical results in this work. All numerical code is available \href{https://github.com/Mac-Shaw/morphing_qec_circuits}{on GitHub}~\cite{Morphing_GitHub}.

The end-cycle code distance was obtained by using the numerical optimization software Gurobi~\cite{Gurobi}, which allows one to approximately solve linear optimization problems over the integers. Specifically, we follow the methods of Ref.~\cite{Landahl11} to write the problem of finding the distance of a code as an integral linear optimization, converting ${\rm mod} 2$ arithmetic to finding even integers. That is, for $\mathbf{x}\in \mathbb{Z}^{n}$, $\mathbf{y}\in\mathbb{Z}^{n_{Z}}$ and $z\in\mathbb{Z}$, where $n$ is the number of qubits and $n_{Z}$ is the number of $Z$-stabilizer generators, and for $\overline{Z}$ the $\mathbb{Z}_{2}^{n}$ row vector representing the support of a logical $Z$-operator, we minimize $\sum_{i}x_{i}$ subject to the constraints
\begin{subequations}\label{eq:distance_linear_program}
    \begin{align}
        H_{Z}\mathbf{x}+2\mathbf{y}&=\mathbf{0},\\
        \overline{Z}\mathbf{x}+2z&=1,\\
        0\leq x_{i}&\leq 1.
    \end{align}
\end{subequations}
The minimization in \cref{eq:distance_linear_program} finds the shortest logical $\overline{X}$ operator that anti-commutes with the particular logical $\overline{Z}$ operator chosen. 

Repeating this procedure for all $k$ logical $\overline{Z}$ operators ensures that the $X$-distance $d_{X}$ is found. From Ref.~\cite{Bravyi24}, this immediately gives the full distance since $d_{X}=d_{Z}=d$.

The circuit-level distance upper-bound was found using BP-OSD following the approach of Ref.~\cite{Bravyi24}. Specifically, we define the circuit-level distance as the minimum number of faults in the circuit required to give a logical error. Since our codes are CSS we have $d_{\text{circ}}=\min(d_{X,\text{circ}},d_{Z,\text{circ}})$ and therefore we only need to check for $X$ and $Z$ logical errors independently. Taking $X$-errors, we consider a two-QEC-round $Z$-basis memory experiment in which the initial state is the all $\ket{0}^{\otimes n}$ state, before performing the circuit from $\tilde{C}_{1}\rightarrow\tilde{C}_{2}$, and followed by a measurement of all the data qubits in the $Z$-basis. We construct all our circuits using Stim~\cite{Gidney21}. We construct a noise model consisting only of single-qubit $X$-errors occurring at every space-time location in the circuit with probability $p$. Note that we do not need to explicitly add two-qubit $X$-errors after CNOT gates since these are always logically equivalent to a single-qubit $X$-error before the same CNOT. Our noise model gives rise to a detector error model with detectors corresponding to the $Z$-stabilizer measurements. However, we also include a chosen logical $\overline{Z}$ operator --- which would usually be included as an observable to check the outcome of decoding --- instead as a detector. Then, we configure the BP-OSD decoder~\cite{Roffe20} using the detector error model. We used a physical error rate of $p=0.0001$, the product\_sum BP method, the osd\_cs OSD method, an OSD order of 100, and maximum BP iterations of 100. To calculate the circuit-level distance upper bound, we supply the BP-OSD decoder with a syndrome consisting of all-zeros for both rounds of standard detectors, and one for the logical $\overline{Z}$ operator. BP-OSD will then attempt to find the lowest-weight physical $X$-error that commutes with the stabilizers but anti-commutes with the $\overline{Z}$ operator. The weight of the output of the BP-OSD decoder is therefore an upper bound on the circuit-level distance. We repeat this for every logical $\overline{Z}$ representative of each code to obtain an upper-bound on $d_{Z\text{circ}}$, and repeat this for $d_{X,\text{circ}}$, to obtain the upper-bounds listed in Table I of Ref.~\cite{Shaw24Lowering}.

\begin{figure*}
    \centering
    \includegraphics[width = \linewidth]{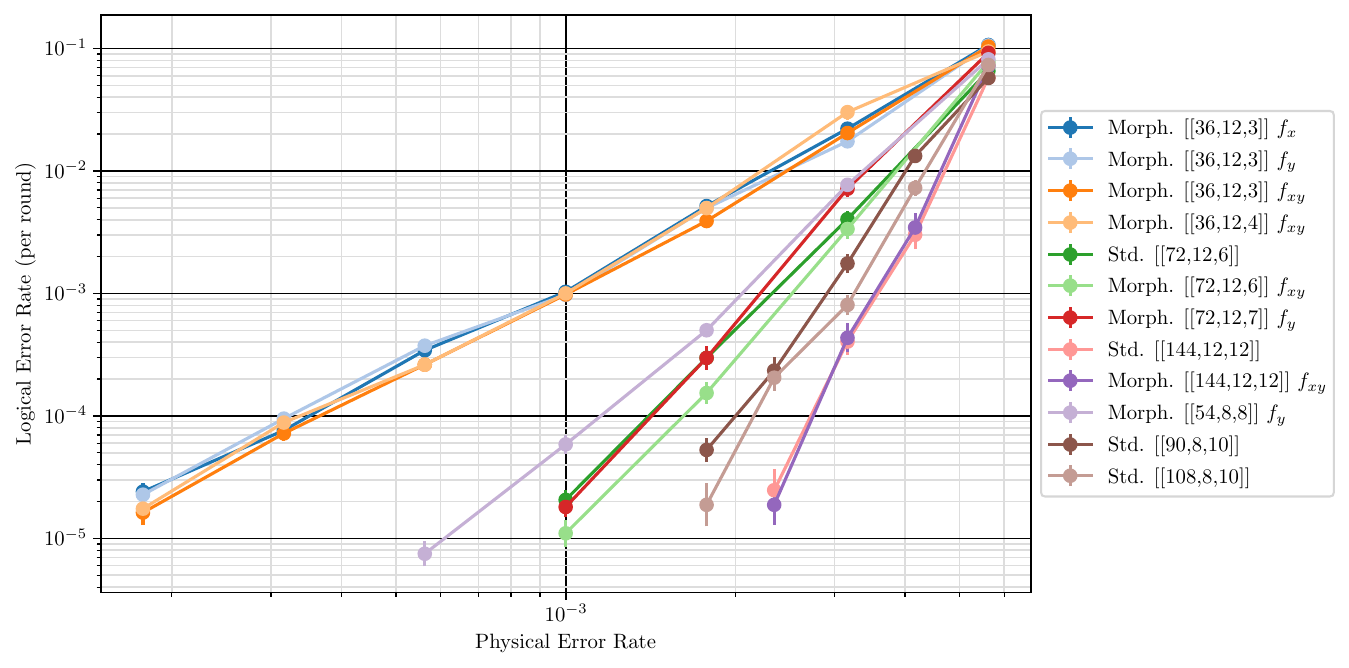}
    \caption{Numerical logical performance of codes from Ref.~\cite{Bravyi24} under standard parity check circuits (Std.) and the new codes designed from morphing circuits (Morph.) with respect to a uniform circuit-level depolarizing noise model and decoded using BP-OSD.}\label{fig:full_numerical_results}
\end{figure*}

To obtain the numerical results in \cref{fig:full_numerical_results}, we again use Stim~\cite{Gidney21}. Our noise model is a uniform circuit-level depolarising noise model. Each measurement provides the incorrect result with probability $p$, each reset prepares an orthogonal state with probability $p$, and after each CNOT gate we apply a two-qubit depolarising noise channel of strength $p$, i.e.~each non-identity element of the two-qubit Pauli group has a probability $p/15$ of occurring. We begin each experiment by perfectly preparing an encoded Bell state between the data qubits in the code and a set of $k$ error-free reference qubits. Then, we run noisy QEC circuits for $\tilde{d}$ rounds, before measuring the logical $\overline{X}_{i}X_{i}$ and $\overline{Z}_{i}Z_{i}$ operators between the BB code and the reference qubits. We again use BP-OSD to decode the syndrome --- with the $X$ and $Z$ syndromes together --- this time configured with the minimum\_sum BP method, the osd\_cs OSD method, an OSD order of 20 and a maximum BP iterations of 10,000. If the decoder fails to predict the measurement outcome of any of the Bell measurements in a given shot, we record that shot as a failure. We ran simulations for all three homomorphisms for the $[[36,12,3]]$ code and seeing that the error rates were very similar, limited ourselves to the $f_{xy}$ homomorphism where possible for the larger code instances. Our full results are presented in \cref{fig:full_numerical_results}.


\section{A General Description of the Morphing Protocol}\label{sec:general_description}

Here we describe the morphing construction in more depth than in the main text and discuss a few additional points, namely its relation to the previous work on morphing quantum codes and how to simply bound the distance of the end-cycle codes.

To recap, given a known, mid-cycle code $C$ with parameters $[[n,k,d]]$ with a set of stabilizer generators $S$, the goal is to design a parity check circuit that only involves Clifford gates and Pauli basis measurement/reset. The idea is to contract a subset $S_{1}\subseteq S$ of the stabilizer generators each of which is contracted onto a single qubit by the contraction circuit $F_1$. Then, those qubits are measured ($M_{1}$) and reset ($R_{1}$) before $F_1^{\dagger}$ is applied. We then repeat this process for a set of contraction circuits $F_{i}$, measurements $M_{i}$ and resets $R_{i}$ for some number of rounds $i=1,\dots,I$ such that every stabilizer generator is contracted in at least one of the contraction rounds, i.e.~$\bigcup_{i=1}^{I}S_{i}=S$.

During each round of measurements $M_{i}$, we say that the remaining qubits are encoded in the \textit{end-cycle} code $\tilde{C}_{i}$, and we label the parameters of this code $[[\tilde{n}_{i},k,\tilde{d}_{i}]]$. We then perform error-correction on the code $\tilde{C}_{i}$ by performing the circuit $M_{i+1}^{\vphantom{-1}}\circ F_{i+1}^{\vphantom{\dagger}}\circ F_{i}^{\dagger}\circ R_{i}^{\vphantom{-1}}$ (with addition modulo $I$), which simultaneously transforms the code into the next end-cycle code $\tilde{C}_{i+1}$. The codes $\tilde{C}_{i}$ may be closely related to each other --- for example, being equivalent up to a permutation of qubits --- but this is not a necessity. Note that not all the stabilizers of $\tilde{C}_{i}$ are necessarily measured when going to $\tilde{C}_{i+1}$; only when $I=2$ can we say that all stabilizers are measured each round.

We also stated in the main text that the contraction circuit should not use any ancilla qubits, since the mid-cycle code is already encoded across all data and ancilla qubits. However, this does not have to be the case, and indeed it can sometimes be useful to add a small number of ancilla qubits to design the contraction circuit. We can interpret the additional ancilla qubits as still being encoded in the mid-cycle code, but that support single-qubit stabilizers. Indeed, both the hex-grid surface code~\cite{McEwen23} and the morphing triangular color code~\cite{Gidney23} use $O(\sqrt{n})$ boundary ancilla qubits. We will also make use of ancilla qubits later in \cref{sec:logical_operations}.

Let us briefly consider a few extreme cases. One can view each contraction circuit $F_{i}$ as a partial decoding circuit of the code $C$, where a subset $S_i$ of the stabilizer generators are decoded. In the case where $I=1$, the contraction circuit $F_{1}$ becomes a full decoding circuit, and therefore the end-cycle code simply consists of $k$ unencoded qubits with distance $\tilde{d}=1$. 
It is therefore only sensible to consider morphing protocols with more than one contraction circuit, $I\geq 2$.

On the other extreme, if $I$ is equal to the number of stabilizer generators of the code, each $F_{i}$ contracts just one generator of the code at the time. This has the advantage of always being possible. However, since the number of contraction circuits is extremely large, each stabilizer generator is measured very infrequently and hence the QEC cycle is very long and inefficient.

It is therefore favorable to minimize the number of contraction circuits $I$ while keeping the (circuit) distance as high as possible. Indeed, in this work, we only consider circuits with $I=2$. It may be possible that a larger $I> 2$ leads to a larger end-cycle code distance $\tilde{d}_{i}$, so that the increased distance counteracts the negative effects of the larger $I$ in the parity check performance. We leave investigating this possibility to future work.

\subsection{Morphing Quantum Codes}

We have noted that the transformation from code $C$ to $\tilde{C}_i$ is related to the idea of morphing quantum codes~\cite{Vasmer22}, which was introduced as a way of generating new quantum codes with fault-tolerant logical gates. In this construction, one begins with a known code $C$ and selects a subset of qubits $R$. The stabilizers with support on $R$ define a subgroup $\mathcal{S}(R) \leq \mathcal{S}$ of the stabilizer group $\mathcal{S}$. To generate the new, morphed quantum code, one applies the $O(|R|)$-depth decoding circuit of the code $\mathcal{S}(R)$ so that its logical qubits get decoded, and its stabilizer generators are contracted to single-qubit $Z$ stabilizers. This morphing procedure leads to a new code $C_{\backslash R}$ (called the ``child'' code) from the known code $C$ (called the ``parent'' code).

Morphing in Ref.~\cite{Vasmer22} is distinct from the morphing construction here in that we don't select a subset of qubits (in a region), but only a subgroup of stabilizers $\mathcal{S}_{i}=\langle S_{i}\rangle$ that are typically (but not necessarily) spread out through the code lattice. The circuit $F_{i}$ can then be interpreted as a decoding circuit for the quantum code defined by $\mathcal{S}_{i}$. Indeed, each of the morphing constructions in Ref.~\cite{Vasmer22} does define a valid morphing parity-check circuit, but there is no a priori reason why this circuit would be useful in QEC since its aim is to transform to a different code to enact fault-tolerant logical gates.

With regards to fault-tolerant logical gates, one can make the following simple observation. If the code $C$ has some transversal logical gate $\overline{U}$, then the end-cycle codes $\tilde{C}_i$ directly inherit it, as one can just apply the gate $\overline{U}$ while one is in the mid-cycle code $C$. On the other hand, note that fault-tolerant gate constructions which come about via stabilizer measurements themselves (e.g. lattice surgery) need to be adapted to the morphing parity check circuits, see details in \cref{subsec:logical_I/O}.

\subsection{Distance Bound}
\label{sec:distance_lowerbound}

Compared to the code $C$ with parameters $[[n,k,d]]$, each end-cycle code has the same number of logical qubits but fewer physical qubits, $\tilde{n}_{i}<n$. Moreover, the distance $\tilde{d}_{i}$ may be larger or smaller than $d$. However, we can derive a simple lower-bound for the distance $\tilde{d}_{i}$ based on the structure of the contraction circuit $F_{i}$:
\begin{proposition}\label{prop:distance_lower_bound}
    The distance $\tilde{d}_{i}$ of the end-cycle code $\tilde{C}_{i}$ obeys the lower bound
    \begin{equation}
        \tilde{d}_i\geq d/c_{i},
    \end{equation}
    where $c_{i}$ is the maximum weight of the operator $F_{i}^{\dagger}PF_{i}^{\vphantom{\dagger}}$, for any single-qubit Pauli operator $P$ not supported on the qubits involved in the measurement $M_{i}$. If the end-cycle code $\tilde{C}_i$ is CSS, we can additionally let $P=X,Z$.
\end{proposition}
\begin{proof}
    Consider a minimum-weight logical Pauli operator $\tilde{P}$ with weight $\tilde{d}$ in the end-cycle code (i.e.~with support only on the data qubits). Now consider the ``pre-measurement''/``post-reset'' code, which is the code on both the data and ancilla qubits consisting of the stabilizer generators of $\tilde{C}_{i}$ on the data qubits, and single-qubit $X$ or $Z$ stabilizers on the ancilla qubits. This pre-measurement/post-reset code is the code that is encoded immediately before (or after) the measurement (reset) of the ancilla qubits. $\tilde{P}$ is also a valid logical operator of the pre-measurement/post-reset code. Then, $F_{i}^{\dagger}\tilde{P}F_i^{\vphantom{\dagger}}$ with Clifford circuit $F_i$ is a valid logical operator of the mid-cycle code $C$ that we define to have weight $w$. Moreover, the weight $w$ is upper-bounded by $w \leq c_{i}\tilde{d}_i$. Meanwhile, since the distance of the mid-cycle code $C$ is $d$, we have $d\leq w$. This proves the proposition.
\end{proof}

In the most general case, one can use light-cone arguments to bound $c_{i}\leq 2^{t}$, where $t$ is the depth of the contraction circuit $F_{i}$. Therefore, if the contraction circuit is constant depth across some family of LDPC codes, we have $\tilde{d}_{i}=\Omega(d)$. For a specific contraction circuit, a tighter lower bound can readily be achieved by evaluating $c_{i}$. For example, for the circuits presented in the main text~\cite{Shaw24Lowering}, $c_{i}=3$ so our lower-bound is $\tilde{d}_{i}\geq d/3$. Note that all our end-cycle codes in Table I in the main text~\cite{Shaw24Lowering} outperform this lower bound, achieving at least $\tilde{d}_{i}\geq d/2$. It is an open question whether our distance bound can be improved.


\section{A Morphing Protocol for Two-Block Group Algebra Codes}
\label{sec:general_2BGA}
Having described the general properties of the morphing protocol in \cref{sec:general_description}, we now provide a concrete construction that applies to arbitrary, not-necessarily-Abelian 2BGA codes \cite{Lin23}. Moreover, in this appendix, we will provide a different approach to constructing the morphing circuits than given in the main text. First, this more systematic approach will give a better sense of the general approach that we used to construct the circuits themselves. For those simply wishing to find a justification of the validity of the circuits presented in the main text, we refer the reader instead to \cref{cor:validity} in \cref{sec:logical_operators}. Second, it will allow us to consider 2BGA codes that do not necessarily satisfy the homomorphism criterion Crit.~1 from the main text~\cite{Shaw24Lowering}. We will later show how satisfying Crit.~1 is sufficient to guarantee the existence of a morphing protocol with $I=2$.

2BGA codes are defined by an arbitrary group $G$ and two sets of group elements $A=\{a_{i}\}$, $B=\{b_{i}\}$. We use the same notation as in the main text, with qubits labeled $q(L,g)$ and $q(R,g)$ for left and right qubits respectively, and $X(P,Q)$ and $Z(P,Q)$ for an $X$- or $Z$-operator with support on the left qubits $P\subseteq G$ and the right qubits $Q\subseteq G$. Then, the stabilizer generators, of weight $w=|A|+|B|$, are given by $s(X,g)=X(Ag,gB)$ and $s(Z,g)=Z(gB^{-1},A^{-1}g)$ for $g\in G$. Here, the inverse notation is interpreted element-wise, i.e.~$A^{-1}=\{a_{i}^{-1}\}$.

It is possible without loss of generality to assume that $a_{1}=b_{1}=1$, however in these appendices we will not make this assumption. This is largely for convenience --- for example, the labeling of the $a_{1}$ and $b_{1}$ elements differ depending on the homomorphism chosen, even for the same code. Therefore making the assumption $a_{1}=b_{1}=1$ can be inconvenient because different relabellings need to be applied for each application of the homomorphism.

Before continuing, it is worth briefly understanding the structure of the stabilizers of a 2BGA code, summarized in the ``Overlapping stabilizers'' columns of \cref{fig:general_contracting_stabilizers}. Consider an $X$-stabilizer $s(X,g)$, with support on the left qubits $q(L,a_{j}g)$ and the right qubits $q(R,gb_{j})$. The stabilizer $s(X,g)$ will overlap with another $X$-stabilizer $s(X,g')$ on a left qubit $q(L,a_{j}g)$ iff $g=a_{j}^{-1}a_{k}^{\vphantom{-1}}g'$ for some $k$, and on the right qubit $q(R,gb_{j})$ iff $g=g'b_{k}^{\vphantom{-1}}b_{j}^{-1}$ for some $k$. Meanwhile, the stabilizer $s(X,g)$ will overlap with a $Z$-stabilizer $s(Z,g')$ on a left qubit $q(L,a_{j}g)$ iff $g=a_{j}^{-1}g'b_{k}^{-1}$, in which case they necessarily also overlap on the right qubit $q(R,gb_{k})$. A single $X$-stabilizer can therefore overlap with no more than $|A|(|A|-1)+|B|(|B|-1)$ other $X$-stabilizers, and no more than $|A|\cdot|B|$ $Z$-stabilizers. Similar equations can be derived for a $Z$ stabilizer $s(Z,g)$.

\begin{figure*}
    \centering
    \includegraphics[width=\linewidth]{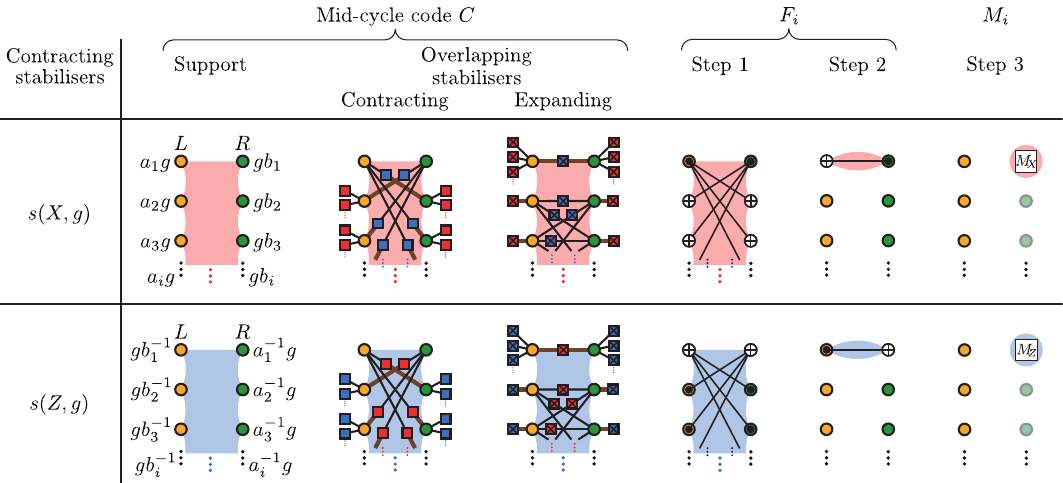}
    \caption{Contraction circuits and measurements for general 2BGA codes. On the left, the mid-cycle support of an arbitrary contracting $X$-stabilizer $s(X,g)$ and $Z$-stabilizer $s(Z,g)$. In the two columns labeled ``Overlapping stabilizers'' we show all the other stabilizers that have overlapping support with the contracting stabilizer. The overlapping stabilizers are shown in the style of a Tanner graph by red and blue boxes for $X$- and $Z$-stabilizers respectively and edges connecting them to the qubits they have support on. Note that it is possible that some of these overlapping stabilizers coincide if the 2BGA code has a special structure. For convenience, we format edges that represent group multiplication by either $a_{1}$ or $b_{1}$ to be brown and bold. The overlapping stabilizers in the expanding column \textit{must} be expanding due to Crit.~\ref{crit:X_criterion}-\ref{crit:Z_criterion}. The overlapping stabilizers in the contracting column do not necessarily need to be contracting, but are under Crit.~\ref{crit:homomorphism_appendix}. On the right, the contraction circuit $F_{i}$ and measurement $M_{i}$, along with the support of the contracting stabilizer at each step before the gates/measurements are applied. When executed just for a single stabilizer, steps 1 and 2 correspond to the local contraction circuits $F(X,g)$ and $F(Z,g)$. Simultaneously executing multiple local contraction circuits for multiple contracting stabilizers gives the global contraction circuit $F_{i}$.}
    \label{fig:general_contracting_stabilizers}
\end{figure*}

We now describe the morphing construction for 2BGA codes, shown in \cref{fig:general_contracting_stabilizers}. We proceed by first describing a local contraction circuit $F(P,g)$ for a \textit{single} stabilizer $s(X,g)$ or $s(Z,g)$, and then considering the conditions under which these circuits can be executed simultaneously in a single, global contraction circuit $F_{i}$. We split the local circuit into two steps, each of which contains a set of commuting CNOT gates. We call these ``steps'' instead of ``rounds'' as in Table II of the main text~\cite{Shaw24Lowering} because each step may contain gates that are not simultaneously executable.

The circuit $F(X,g)$ proceeds as follows for $j=2,\dots,|A|$ and $k=2,\dots,|B|$:
\begin{itemize}
    \item \textit{Step 1:} 
     CNOT$\big(q(L,a_{1}g),q(R,gb_{k})\big)$, and 
    CNOT$\big(q(R,gb_{1}),q(L,a_{j}g)\big)$. This step has depth $\max(|A|,|B|)-1$.
    \item \textit{Step 2:} CNOT$\big(q(R,gb_{1}),q(L,a_{1}g)\big)$.
\end{itemize}
After completing $F(X,g)$ the stabilizer $s(X,g)$ has been contracted to the qubit $q(R,gb_{1})$ which can then be subsequently measured and reset. Meanwhile, for $F(Z,g)$ we have:
\begin{itemize}
    \item \textit{Step 1:} CNOT$\big(q(R,a_{j}^{-1}g),q(L,gb_{1}^{-1})\big)$, and CNOT$\big(q(L,gb_{k}^{-1}),q(R,a_{1}^{-1}g)\big)$,
    \item \textit{Step 2:} CNOT$\big(q(L,gb_{1}^{-1}),q(R,a_{1}^{-1}g)\big)$.
\end{itemize}
After $F(Z,g)$, the stabilizer $s(Z,g)$ has been contracted to the qubit $R(a_{1}^{-1}g)$.

Now, we want to construct a global contraction circuit $F_{i}$ that applies the local circuits $F(X,g)$ for $g\in G_{X,i}$ and $F(Z,g)$ for $g\in G_{Z,i}$ in parallel, for some subsets $G_{X,i},G_{Z,i}\subset G$. Throughout the following section we will only provide \textit{sufficient} conditions for the subsets $G_{X,i}$, $G_{Z,i}$ to define a valid global contraction circuit, and other possible combinations of stabilizer generators could be possible.

In order for $F_{i}$ to be a valid contraction circuit, we require that any contracting stabilizers should \textit{not} be inadvertently ``expanded'' by adjacent local contraction circuits. Take, for example, a contracting $X$-stabilizer $s(X,g)$ and its local contraction circuit $F(X,g)$. In Step 1, CNOT gates are performed with controls on the qubits $q(L,a_{1}g)$ and $q(R,gb_{1})$. These CNOT gates affect, for example, any other $X$-stabilizers $s(X,g')$ with support on these qubits, i.e. the support of this stabilizer could spread to other qubits originally in $s(X,g)$. However, this is not a problem if none of these overlapping $X$-stabilizers are contracting. From our earlier characterization of the $X$-stabilizers, this gives a sufficient criterion for the contracting circuits $F(X,g)$ not to interfere with each other:
\begin{appendixcriterion}[$X$-criterion]\label{crit:X_criterion}
    $$g\in G_{X,i}\Rightarrow a_{j}^{-1}a_{1}^{\vphantom{-1}}g\notin G_{X,i}\text{ and }gb_{1}^{\vphantom{-1}}b_{j}^{-1}\notin G_{X,i},$$
    for all $i=1,\dots,I$ and for all $j\neq 1$.
\end{appendixcriterion}
This criterion is represented visually in \cref{fig:general_contracting_stabilizers}. Any overlapping $X$-stabilizers that cannot be contracting due to Crit.~\ref{crit:X_criterion} are marked with a cross (standing for eXpanding) and placed in the ``Expanding'' column. The remaining $X$-stabilizers may or may not be contracting.

The CNOT in step 2 is guaranteed not to interfere in the contraction of any other $X$-stabilizers, because at this point the support of each contracting $X$-stabilizer is simply $X(a_{1}g,gb_{1})$, which is non-overlapping for different values of $g$. Therefore Crit.~\ref{crit:X_criterion} is sufficient to ensure that the application of $F(X,g)$ for $g \in G_{X,i}$ indeed simultaneously contracts these $X$-stabilizers $s(X,g)$ for $g \in G_{X,i}$.

It remains to check whether the circuit $F(X,g)$ interferes with the contraction of a $Z$-stabilizer $s(Z,g')$. If the $Z$-stabilizer overlaps with the target of a CNOT in $F(X,g)$ but \textit{not} the control, then its support will spread and that $Z$-stabilizer should not be contracting. This can be readily determined by inspection of \cref{fig:general_contracting_stabilizers}.
We find that this is the case for the $Z$-stabilizers $s(Z,a_{j}gb_{k})$ for $j,k\neq 1$ that overlap with $s(X,g)$ on the pair of qubits $q(L,a_{j}g)$ and $q(R,gb_{k})$. However, we have an additional condition that arises from step 2. The only possible overlapping contracted $Z$-stabilizer is $s(Z,a_{1}gb_{1})$, in which case both the contracting $X$- and $Z$-stabilizers have support on the qubits $q(L,a_{1}g)$ and $q(R,gb_{1})$. However, the direction of the CNOTs in step 2 of $F(X,g)$ and $F(Z,g)$ are in the opposite directions, and therefore we cannot simultaneously execute the circuits~\footnote{Of course, it would be fairly simple to reverse the direction of the CNOT in Step 2 of, say, $F(Z,g)$ to avoid this problem. However in our constructions, this modification doesn't provide any benefits, so for simplicity we do not consider it here.}. Adding this, we arrive at the condition
\begin{appendixcriterion}[$XZ$-criterion]\label{crit:XZ_criterion}
    $$g\in G_{X,i}\Rightarrow a_{j}g b_{k}\notin G_{Z,i},$$
    for all $i=1,\dots,I$ and for all $j,k$ where either $j=k=1$ or $j,k\neq1$.
\end{appendixcriterion}

Identical arguments for contracting $Z$-stabilizers with respect to overlapping stabilizers recover the $XZ$-criterion and also give rise to a third criterion:
\begin{appendixcriterion}[$Z$-criterion]\label{crit:Z_criterion}
    $$g\in G_{Z,i}\Rightarrow a_{j}^{\vphantom{-1}}a_{1}^{-1}g\notin G_{Z,i}\text{ and }gb_{1}^{-1}b_{j}^{\vphantom{-1}}\notin G_{Z,i},$$
    for all $i=1,\dots,I$ and for all $j\neq 1$.
\end{appendixcriterion}

As long as the subsets $g\in G_{X,i}$ and $F(Z,g)$, $g\in G_{Z,i}$ satisfy these three Criteria~\ref{crit:X_criterion}-\ref{crit:Z_criterion}, we can ensure that the local contraction circuits $F(X,g)$ do not interfere with each other and therefore can be simultaneously executed in a single global contraction circuit $F_{i}$. 

The problem of designing the morphing circuits has now been reduced to finding contracting subsets $G_{X,i}$ and $G_{Z,i}$ that satisfy Criteria~\ref{crit:X_criterion}-\ref{crit:Z_criterion}. At this point, there is no restriction on $I$; indeed, we can recover one of the extreme cases from \cref{sec:general_description} if we set $I$ to be the size of the set of stabilizer generators $|S_{i}|$ and contract only one stabilizer at a time. However, one way to ensure the existence of a morphing protocol with $I=2$ is via the following homomorphism criterion:
\begin{appendixcriterion}[Homomorphism Criterion]\label{crit:homomorphism_appendix}
    There exists a group homomorphism $f:G\rightarrow\mathbb{Z}_{2}$ such that $f(a_{1})\neq f(a_{j})$ and $f(b_{1})\neq f(b_{j})$ for all $j\neq 1$.
\end{appendixcriterion}
Note that $f$ cannot be the trivial homomorphism $f(g)=1$. If such a homomorphism exists, then we can construct a pair of global contraction circuits $F_1$ and $F_2$ as follows.  Recall that the definition of a group homomorphism is that it ``respects'' group multiplication, i.e.~$f(g_{1}g_{2})=f(g_{1})f(g_{2})$. We write the \textit{kernel} of $f$ as $K=\{g\in G\mid f(g)=1\}$, and the complement of $K$ as $K^{c}=\{g\in G\mid f(g)=u\}$. Since $\mathbb{Z}_{2}$ is Abelian, note that one can also left- or right-multiply a coset by an element of $G$ to obtain another coset, for example, $gK=Kg=\{g'\in G\mid f(g')=f(g)\}$. Note moreover that because $u=u^{-1}$, we have $f(g^{-1})=f(g)^{-1}=f(g)$  and therefore $gK=g^{-1}K$ for all $g \in G$.

With this notation, we define the contracting subsets as
\begin{subequations}\label{eq:contracting_subsets}
\begin{align}
    G_{X,1}&=a_{1}K,&G_{Z,1}&=b_{1}K^{c},\\
    G_{X,2}&=a_{1}K^{c},&G_{Z,2}&=b_{1}K.
\end{align}
\end{subequations}
We can show that these subsets satisfy Criterion~\ref{crit:X_criterion} using
\begin{subequations}
\begin{align}
    g\in G_{X,i}&\Rightarrow f(g)=f(a_{1})u^{i-1}\\
    &\Rightarrow f(a_{j}^{-1}a_{1}g)\neq f(a_{1})u^{i-1}\\
    &\Rightarrow a_{j}^{-1}a_{1}^{\vphantom{-1}}g\notin G_{X,i}.
\end{align}
\end{subequations}
Similar arguments show that $gb_{1}^{\vphantom{-1}}b_{j}^{-1}\notin G_{X,i}$, and it is likewise straightforward to show that these subsets satisfy Crit.~\ref{crit:Z_criterion}. To show that Crit.~\ref{crit:XZ_criterion} is satisfied, note that if $j=k=1$ then
\begin{multline}
    g\in G_{X,i}\Rightarrow f(a_{1}gb_{1})=f(b_{1})u^{i-1}\neq f(b_{1})u^{i},
\end{multline}
and therefore $a_{1}gb_{1}\notin G_{Z,i}$.
Meanwhile if $j,k\neq 1$, then $f(a_{j})=u\cdot f(a_{1})$ and $f(b_{k})=u\cdot f(b_{1})$, so
\begin{multline}
    g\in G_{X,i}\Rightarrow \notag \\
    f(a_{j}gb_{k})=f(a_{1})f(a_{j})f(b_{k})u^{i-1}=f(b_{1})u^{i-1},
\end{multline}
and again $a_{j}gb_{k}\notin G_{Z,i}$.

\begin{table}
    \caption{Definition of the contracting circuits $F_{i}$ and measurements $M_{i}$ for a 2BGA code that satisfies Crit.~\ref{crit:homomorphism_appendix}. The circuit $F_{1}$ is defined by applying the gates below for each $g\in K$, $h\in K^{c}$, $j=2,\dots,|A|$ and $k=2,\dots,|B|$, while for $F_{2}$ one chooses $g\in K^{c}$ and $h\in K$ (and the same ranges for $j$ and $k$). The circuits consist of three steps --- named as such because the CNOTs in each step commute but are not necessarily simultaneously executable.}\label{tab:general_contraction_circuits}
    \renewcommand{\arraystretch}{1.35}
    \setlength{\tabcolsep}{6pt}
    \begin{tabular}{|c|c|}
    \hline
        \multirow{2}{*}{Step 1}& CNOT$\big(q(L,g),q(R,a_{1}^{-1}gb_{k}^{\vphantom{-1}})\big),$\\
        & CNOT$\big(q(R,a_{j}^{-1}hb_{1}^{\vphantom{-1}}),q(L,h)\big)$\\\hline
        \multirow{2}{*}{Step 2}& CNOT$\big(q(R,a_{1}^{-1}gb_{1}^{\vphantom{-1}}),q(L,g)\big),$\\
        &CNOT$\big(q(L,h),q(R,a_{1}^{-1}hb_{1}^{\vphantom{-1}})\big)$\\\hline
        \multirow{2}{*}{Step 3}& $M_{X}\big(q(R,a_{1}^{-1}gb_{1}^{\vphantom{-1}})\big)$,\\
        &$M_{Z}\big(q(R,a_{1}^{-1}hb_{1}^{\vphantom{-1}})\big)$\\\hline
    \end{tabular}
\end{table}

These calculations prove that any 2BGA code with a homomorphism $f$ that satisfies Crit.~\ref{crit:homomorphism_appendix} --- which is given as Crit.~1 in the main text~\cite{Shaw24Lowering} for weight-6 codes --- admits a morphing parity check protocol with $I=2$ and the global contraction circuits shown in \cref{tab:general_contraction_circuits}. This is convenient since Crit.~\ref{crit:homomorphism_appendix} can be very quickly checked, usually by inspection. Indeed, any group homomorphism $f$ is uniquely specified by how it transforms the generators $\{g_{j}\}$ of the group $G$. For each generator of odd order, i.e.~satisfying $g_{j}^{n}=1$ for some odd integer $n$, we must have $f(g_{j})=1$ for $f$ to be a homomorphism. For each generator that is not odd, we can choose to map it to either $u$ or $1$ under the homomorphism. This gives a relatively limited number of possible homomorphisms $f$ to check; for groups with two generators such as $\mathbb{Z}_{\ell}\times\mathbb{Z}_{m}$ that define the BB codes, there are at most three such choices (once the trivial map has been excluded), as listed in Eq.~(1) of the main text~\cite{Shaw24Lowering}.

It is straightforward to verify Crit.~\ref{crit:homomorphism_appendix} applies for each homomorphism listed in Table 1 of the main text~\cite{Shaw24Lowering}. To demonstrate the generality of this Crit.~\ref{crit:homomorphism_appendix}, we also found two non-Abelian 2BGA codes that satisfy Crit.~\ref{crit:homomorphism_appendix} from Table I of Ref.~\cite{Lin23}. These are:
\begin{itemize}
    \item the [[96,8,12]] code defined by $G=\langle x,y| x^{16},y^3,x^{-1} yxy\rangle$ with $A=\{x,1,y,x^{14}\}$ and $B=\{x^{11},1,x^{2},x^{4}y\}$, under the $f_{x}$ homomorphism, and
    \item the [[96,12,10]] code defined by $G=\langle x,y| x^{8},y^{6},x^{-1}yxy\rangle$ with $A=\{1,x,x^{2}y^{3},x^{3}y^{2}\}$ and $B=\{x,1,x^{6}y^{4},x^{3}y^{5}\}$, under the $f_{xy}$ homomorphism.
\end{itemize}
We leave it to future research to investigate these codes in more detail.

The connectivity graph of the morphing protocol can be found by listing the CNOT gates in steps 1 and 2 for all $g\in G$. This is listed in \cref{tab:general_contraction_circuits}, although the results are the same no matter the choices of $G_{X,i}$ and $G_{Z,i}$. In particular, the CNOTs that need to be performed are CNOT$\big(q(L,g),a(R,a_{j}^{-1}gb_{k}^{\vphantom{-1}})\big)$ where $j=1$ or $k=1$. Therefore, the connectivity graph is \textit{bipartite} between the left and right qubits, and has degree at most $|A|+|B|-1=w-1$, where $w=|A|+|B|$ is the weight of each stabilizer. This is one fewer than the degree of the connectivity graph of a standard parity-check circuit for a code with weight $w$.

\begin{figure*}
    \centering
    \includegraphics[width = \linewidth]{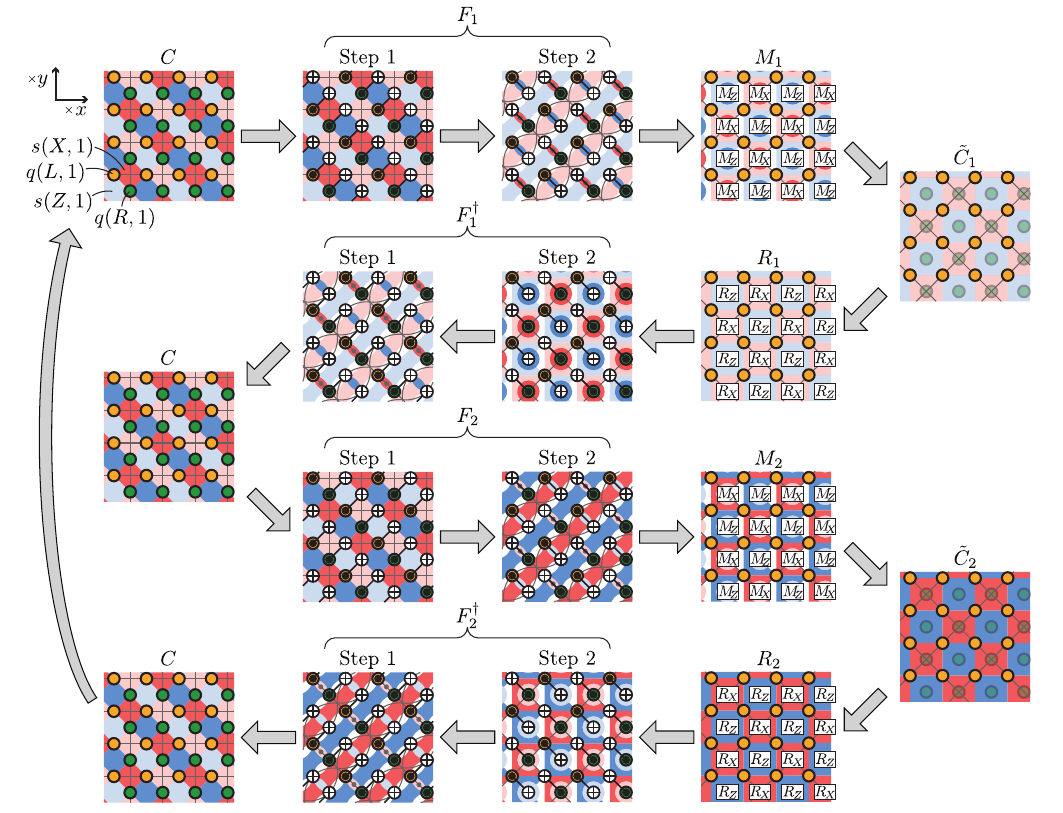}
    \caption{The morphing protocol applied to the $d=4$ toric code, described as an (Abelian) 2BGA by $G=\mathbb{Z}_{d}\times\mathbb{Z}_{d}$, $A=\{1,x\}$ and $B=\{1,y\}$. Qubits are coloured yeLlow (Left) or gReen (Right), with a group label determined by their horizontal and vertical position in the periodic lattice. stabilizers are coloured red ($X$) or blue ($Z$). Moreover, the stabilizers are opaque if they are contracted in the first contraction circuit $F_{1}$ --- corresponding to even $X$-stabilizers and odd $Z$-stabilizers --- and partially transparent if they are expanding in $F_{1}$. The grey grid denotes the toric code lattice representation of the code, with $X$-stabilizers on vertices, $Z$-stabilizers on faces, and qubits on edges. In practice, one starts by preparing an end-cycle codestate $\tilde{C}_{1}$ or $\tilde{C}_{2}$ on the right, and performs error correction by following the grey arrows. At each time step, the stabilizers and toric code lattice represent the code immediately before the execution of the gates/measurements/resets displayed in that step.}\label{fig:surface_code_example}
\end{figure*}

We provide an illustrative example of our morphing protocol applied to the $d=4$ toric code in \cref{fig:surface_code_example}. In this case, the morphing construction reduces to the ``hex-grid'' surface code protocol from Ref.~\cite{McEwen23} but with periodic boundary conditions. The toric code is an Abelian 2BGA code with $G=\mathbb{Z}_{d}\times\mathbb{Z}_{d}$, $A=\{1,x\}$ and $B=\{1,y\}$, and therefore satisfies Crit.~\ref{crit:homomorphism_appendix} with the $f_{xy}$ homomorphism defined by $f_{xy}(x)=f_{xy}(y)=u$. One can therefore run the morphing protocol with the contracting stabilizers corresponding to the cosets $G_{X,1}=K=G_{Z,2}$ and $G_{X,2}=K^{c}=G_{Z,1}$. Note that here, step 1 only takes one round of CNOT gates since the weight of the stabilizers is only four. In this example, the end-cycle codes $\tilde{C}_{i}$ are both $d=4$ \textit{rotated} toric codes encoded on the left-qubits of $C$. However, the two end-cycle codes are not the same: the stabilizers of $\tilde{C}_{1}$ can be mapped to those of $\tilde{C}_{2}$ by multiplying each of their group labels by $x$ (or, indeed, any element in $K^c$). The required connectivity of each qubit is only $w-1=3$ and therefore this can be implemented on a hexagonal grid layout.

Finally, it is worth commenting that each of the definitions and criteria above is just one choice of morphing parity check protocol that works for some 2BGA codes. However, there are many other possible contraction circuits that could work; we leave investigating these possibilities to future research.

\section{Weight-6 Bivariate Bicycle Codes}

In this appendix, we provide more details and calculations pertaining to the morphing parity check protocol in Table 1 of the main text~\cite{Shaw24Lowering} for BB codes satisfying Criterion~1. That is, compared to \cref{sec:general_2BGA}, we will now assume that the group is Abelian and bicyclic, i.e.~$G=\mathbb{Z}_{\ell}\times\mathbb{Z}_{m}$.

We begin by providing a direct proof that the contracting circuit $F_{i}$ does indeed contract the stabilizers in $S_{i}$. Then, we provide a full description of the end-cycle codes $\tilde{C}_{i}$, including their stabilizers as well as their logical operators, and how to rewrite them explicitly as Abelian 2BGA codes.

\subsection{Transformation of Operators Under the Contraction Circuits}\label{sec:logical_operators}

We begin by deriving a few helpful formulae for how Pauli operators are transformed by the contraction circuits $F_{i}$. Just as in the main text, we will use the notation $X(P,Q)$ to denote an $X$-operator with support on some subset $P\subseteq G$ of left qubits and on $Q\subseteq G$ right qubits. Note that the set of subsets of $G$ is isomorphic to the group algebra $\mathbb{Z}_{2}[G]$. Elements of the group algebra can therefore be multiplied by a group element (with action defined by group multiplication), added to each other (with action defined by the group algebra over $\mathbb{Z}_{2}$), or intersected with each other (as sets). Moreover, for a subset $H=\{g_{1},g_{2},\dots\}\subseteq G$ we will write $H^{-1}\equiv \{g_{1}^{-1},g_{2}^{-1},\dots\}$; although note that in the canonical matrix representation of the group algebra, this operation corresponds to taking the transpose of the matrix $H$ and not the inverse.

\begin{table}
    \caption{Definition of the contracting circuits $F_{i}$ and measurements $M_{i}$, assuming that Crit. 1 of the main text~\cite{Shaw24Lowering} is satisfied. The circuit $F_{1}$ is defined by applying the gates below for each $g\in K$ and $h\in K^{c}$, while for $F_{2}$ one chooses $g\in K^{c}$ and $h\in K$. This is identical to Table II of the main text~\cite{Shaw24Lowering} but without the assumption that $a_{1}=b_{1}=1$.}\label{tab:contraction_circuits_appendix}
    \renewcommand{\arraystretch}{1.35}
    \setlength{\tabcolsep}{6pt}
    \begin{tabular}{|c|c|}
    \hline
        \multirow{2}{*}{Round 1}& CNOT$\big(q(L,g),q(R,a_{1}^{-1}b_{3}^{\vphantom{-1}}g)\big),$\\
        & CNOT$\big(q(R,a_{3}^{-1}b_{1}^{\vphantom{-1}}h),q(L,h)\big)$\\\hline
        \multirow{2}{*}{Round 2}& CNOT$\big(q(L,g),q(R,a_{1}^{-1}b_{2}^{\vphantom{-1}}g)\big),$\\
        & CNOT$\big(q(R,a_{2}^{-1}b_{1}^{\vphantom{-1}}h),q(L,h)\big)$\\
        \hline
        \multirow{2}{*}{Round 3}& CNOT$\big(q(R,a_{1}^{-1}b_{1}^{\vphantom{-1}}g),q(L,g)\big),$\\
        &CNOT$\big(q(L,h),q(R,a_{1}^{-1}b_{1}^{\vphantom{-1}}h)\big)$\\\hline
        \multirow{2}{*}{Round 4}& $M_{X}\big(q(R,a_{1}^{-1}b_{1}^{\vphantom{-1}}g)\big)$,\\
        &$M_{Z}\big(q(R,a_{1}^{-1}b_{1}^{\vphantom{-1}}h)\big)$\\\hline
    \end{tabular}
\end{table}

We will write $A=\{a_{1},a_{2},a_{3}\}\equiv a_{1}+a_{2}+a_{3}$ and $B=\{b_{1},b_{2},b_{3}\}\equiv b_{1}+b_{2}+b_{3}$, so that the mid-cycle stabilizers are $X(Ag,Bg)$ and $Z(B^{-1}g,A^{-1}g)$ for $g\in G$. Again, in the appendix we do not assume $a_{1}=b_{1}=1$; the required contraction circuits in this case are given in \cref{tab:contraction_circuits_appendix}. Throughout this appendix we will consider operators that are in the centralizer of the stabilizer group; that is, they are possibly trivial logical operators. It was shown in Ref.~\cite{Bravyi24} that any mid-cycle (possibly trivial) logical operator $X(P,Q)$ must satisfy
\begin{subequations}\label{eq:mid-cycle_logical_condition}
\begin{equation}
    BP+AQ=0,\label{eq:Xcond}
\end{equation}
and any mid-cycle logical operator $Z(P,Q)$ must satisfy
\begin{equation}
    A^{-1}P+B^{-1}Q=0.
    \label{eq:Zcond}
\end{equation}
\end{subequations}

The end-cycle codes $\tilde{C}_{1}$ and $\tilde{C}_{2}$ are defined implicitly by the circuits in \cref{tab:contraction_circuits_appendix}. Since the measurements $M_{1}$ and $M_{2}$ occur on all of the right qubits of the mid-cycle code, the end-cycle codes are supported only on the left qubits. Moreover, the contraction circuit $F_{1}$ and measurements $M_{1}$ are defined in \cref{tab:contraction_circuits_appendix} by setting $g\in K$ and $h\in K^{c}$, while $F_{2}$ and $M_{2}$ are defined by $g\in K^{c}$ and $h\in K$. Therefore, in this section every expression for $\tilde{C}_{1}$ can be turned into an expression for $\tilde{C}_{2}$ simply by interchanging the cosets $K\leftrightarrow K^{c}$. Equivalently, the stabilizer group of $\tilde{C}_{1}$ can be turned into that of $\tilde{C}_{2}$ by multiplying each qubit label by any element $r\in K^{c}$, since this multiplication interchanges the cosets.

We remind the reader of a few properties of the cosets $K=\{g\in G\mid f(g)=1\}$ and $K^{c}=\{g\in G\mid f(g)=u\}$ that we will use throughout the section.
\begin{itemize}
    \item For any $g\in G$, we have $gK=\{h\in G\mid f(h)=f(g)\}$.
    \item For any $k\in K$, the cosets $xK$ and $kxK$ are equal; this is because $f(kx)=f(k)f(x)=f(x)$.
    \item For any $g\in G$, the cosets $gK$ and $g^{-1}K$ are equal; this is because $f(g^{-1})=f(g)^{-1}=f(g)$ since $f(g)\in\mathbb{Z}_{2}$.
    \item From Crit.~1 of the main text~\cite{Shaw24Lowering}, $a_{1}a_{2}K=a_{1}a_{3}K=b_{1}b_{2}K=b_{1}b_{3}K=K^{c}$, while $a_{2}a_{3}K=b_{2}b_{3}K=K$.
\end{itemize}
Throughout this section we will frequently refer to the cosets $a_{1}b_{1}K$ and $a_{1}b_{1}K^{c}$. Note, however, that there are only two distinct cosets $K$ and $K^{c}$, so we either have $a_{1}b_{1}K=K$ or $a_{1}b_{1}K=K^{c}$. Moreover, in the main text, we have $a_{1}=b_{1}=1$ and these additional factors can be dropped. They are retained here merely for convenience when considering the same BB code under multiple different homomorphisms with different $a_{1}$ and $b_{1}$ elements.

\begin{proposition}\label{prop:mid_to_end_propagation}
    Let $X(P,Q)$ be a (possibly trivial) logical $X$-operator of the mid-cycle code $C$. After the contraction circuit $F_{1}$, the operator has partial support on the to-be-measured right-qubits $Q\cap(a_{1}b_{1}K)$. After the measurements $M_1$, the full transformation by $M_1 \circ F_1$ is
\begin{subequations}\label{eq:mid_to_end_propagation}
    \begin{equation}\label{eq:mid_to_end_propagation_X}
    X(P,Q)\mapsto X\big(P+b_{1}^{-1}A(Q\cap a_{1}b_{1}K),0\big),
    \end{equation}
    which is then a (possibly trivial) logical operator of the code $\tilde{C}_{1}$. Meanwhile, a logical $Z$-operator $Z(P,Q)$, is involved in the right-qubit measurements $Q\cap(a_{1}b_{1}K^{c})$, and the full transformation by $M_1 \circ F_1$ is
    \begin{equation}
    Z(P,Q)\mapsto Z\big(P+a_{1}B^{-1}(Q\cap a_{1}b_{1}K^{c}),0\big).
    \end{equation}
    \end{subequations}
    Under $M_{2}\circ F_{2}$, the mappings are the same but with $K\leftrightarrow K^{c}$.
\end{proposition}
\begin{proof}
    We can consider the action of rounds 1 and 2 in \cref{tab:contraction_circuits_appendix} on a Pauli $X$ operator together since these rounds commute.  Any Pauli $X$ operator on a left qubit $q(L,g)$ with $g\in K$ will propagate to the right qubits $q(R,a_{1}^{-1}b_{3}^{\vphantom{-1}}g)$ and $q(R,a_{1}^{-1}b_{2}^{\vphantom{-1}}g)$, while a Pauli $X$ operator on $q(R,h)$ with $h\in a_{1}b_{1}K$ will spread to $q(L,a_{3}^{\vphantom{-1}}b_{1}^{-1}h)$ and $q(L,a_{2}^{\vphantom{-1}}b_{1}^{-1}h)$. In terms of the operator $X(P,Q)$, rounds 1 and 2 thus transform
    \begin{multline}
        X(P,Q)\mapsto X\Big(P+b_{1}^{-1}(a_{2}+a_{3})(Q\cap a_{1}b_{1} K),\\
        Q+a_{1}^{-1}(b_{2}+b_{3})(P\cap K)\Big).
    \end{multline}
    In round 3, a Pauli $X$ on $q(L,h)$ with $h\in K^{c}$ spreads to $q(R,a_{1}^{-1}b_{1}^{\vphantom{-1}}h)$, and a Pauli $X$ on $q(R,g)$ with $g\in a_{1}b_{1} K$ spreads to $q(L,a_{1}^{\vphantom{-1}}b_{1}^{-1}g)$. Noting that $b_{1}^{-1}a_{2}a_{1}b_{1} K=b_{1}^{-1}a_{3}a_{1}b_{1} K=K^{c}$ and $a_{1}^{-1}b_{2} K=a_{1}^{-1}b_{3} K=a_{1}b_{1}K^{c}$, after round 3 we have
    \begin{multline}
        X(P,Q)\mapsto X\Big(P+b_{1}^{-1}A\big(Q\cap a_{1}b_{1} K\big),\\
        Q+a_{1}^{-1}(b_{2}+b_{3})\big(P\cap K\big)\\
        +a_{1}^{-1}b_{1}^{\vphantom{-1}}\big(P\cap K^{c}\big)\\
        +a_{1}^{-1}(a_{2}+a_{3})\big(Q\cap a_{1}b_{1} K\big)\Big).
    \end{multline}
    We can simplify the support of the operator on the right qubits:
    \begin{subequations}\label{eq:XPQ_R3_RHS}
    \begin{align}
        &Q+a_{1}^{-1}(b_{2}+b_{3})\big(P\cap K\big)+a_{1}^{-1}b_{1}^{\vphantom{-1}}\big(P\cap K^{c}\big)\nonumber\\
        &\qquad+a_{1}^{-1}(a_{2}+a_{3})\big(Q\cap a_{1}b_{1} K\big)\label{eq:XPQ_R3_RHS_1}\\
        &=Q\cap\big(a_{1}b_{1} K\big)+\big(a_{1}^{-1}(BP+AQ)\big)\cap\big(a_{1}b_{1}K^{c}\big)\label{eq:XPQ_R3_RHS_2}\\
        &=Q\cap\big(a_{1}b_{1} K\big),\label{eq:XPQ_R3_RHS_3}
    \end{align}
    \end{subequations}
    where from \cref{eq:XPQ_R3_RHS_2} to \cref{eq:XPQ_R3_RHS_3} we have used \cref{eq:Xcond}.
    
    Finally in round 4 the right qubits $q(R,g)$ are measured, with the measurement being an $X$-measurement if $g\in a_{1}^{\vphantom{-1}}b_{1}^{-1} K=a_1 b_1  K$ and a $Z$-measurement if $g\in a_{1}^{\vphantom{-1}}b_{1}^{-1}K^{c}=a_1 b_1 K^{c}$. Note that \cref{eq:XPQ_R3_RHS} ensures that the logical operator does not anticommute with any of the $Z$-measurements. However, the transformed logical $X$-operator will overlap with the $X$-measurements that take place on the right qubits $q(R,g)$ with $g\in Q\cap(a_{1}b_{1} K)$. This implies that after the measurement the end-cycle support of the operator is simply reduced to
    \begin{equation}
        X\big(P+b_{1}^{-1}A(Q\cap a_{1}b_{1} K),0\big),
    \end{equation}
    as claimed. The transformation of the logical $Z$-operators can be proved similarly. Likewise, the expressions for the transformation to $\tilde{C}_{2}$ can be derived simply by interchanging $ K\leftrightarrow K^{c}$.
\end{proof}

\cref{prop:mid_to_end_propagation} provides a convenient way of proving that the contraction circuit $F_{1}$ indeed contracts every contracting stabilizer in $S_{1}$, as we now show.
\begin{corollary}[Validity of the Contraction Circuits]\label{cor:validity}
    The stabilizers $s(X,g)$ with $g\in a_{1}K$ and $s(Z,g')$ with $g'\in b_{1}K^{c}$ are contracted onto a single qubit by $F_{1}$ and measured by $M_{1}$.
\end{corollary}
\begin{proof}
    The $X$-stabilizer $s(X,g)$ is (by definition) a trivial logical operator of $C$ and has the form $X(Ag,Bg)$, i.e.~$P=Ag$ and $Q=Bg$. Note that if $g\in a_{1}K$, we have $(Bg)\cap(a_{1}b_{1}K)=b_{1}g$. Thus, from \cref{prop:mid_to_end_propagation}, the stabilizer is only involved in the measurement on the right qubit $b_{1}g$. Moreover, substituting into \cref{eq:mid_to_end_propagation_X} gives $X(Ag,Bg)\mapsto X(Ag+Ag,0)=X(0,0)$, as required. The same can be proven for the contracting $Z$-stabilizers.
\end{proof}

Using \cref{prop:mid_to_end_propagation}, we can also easily give the form of the stabilizers of the end-cycle code $\tilde{C}_{1}$:
\begin{corollary}[stabilizers of $\tilde{C}_{1}$]\label{cor:end-cycle_stabilizers}
    The end-cycle code $\tilde{C}_{1}$ is defined by the stabilizer generators $X(b_{1}^{-1}ABg,0)$ for $g\in a_{1}K^{c}$ and $Z(a_{1}A^{-1}B^{-1}g',0)$ for $g'\in b_{1}K$.
\end{corollary}
\begin{proof}
    Here, the $X$-stabilizer $s(X,g)$ still has the form $X(Ag,Bg)$, but now with $g\in a_{1}K^{c}$. We therefore have $(Bg)\cap(a_{1}b_{1}K)=(b_{2}+b_{3})g$. Substituting into \cref{eq:mid_to_end_propagation_X} gives $X(Ag,Bg)\mapsto X(Ag+b_{1}^{-1}A(b_{2}+b_{3})g,0)=X(b_{1}^{-1}ABg,0)$, as claimed. Again the $Z$-stabilizers follow through an analogous computation.
\end{proof}

\begin{figure*}
\centering
    \includegraphics[width = \linewidth]{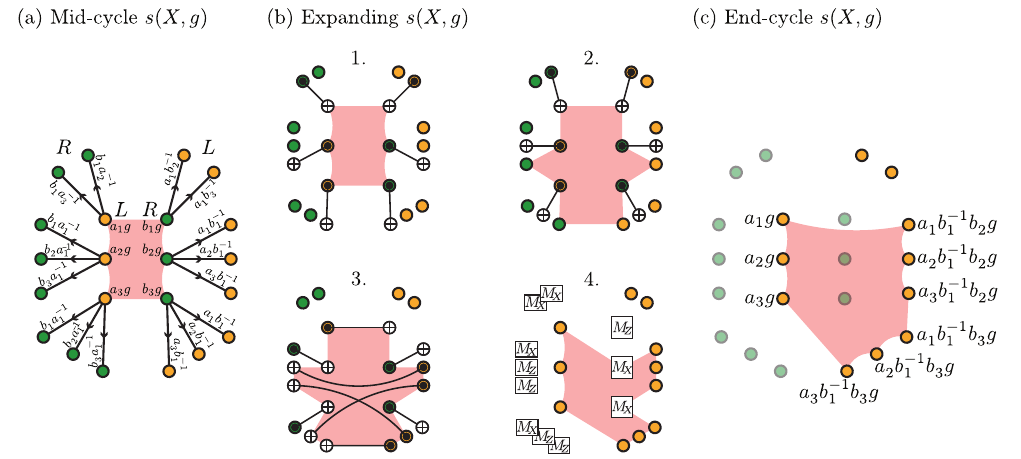}\\
    \includegraphics[width = \linewidth]{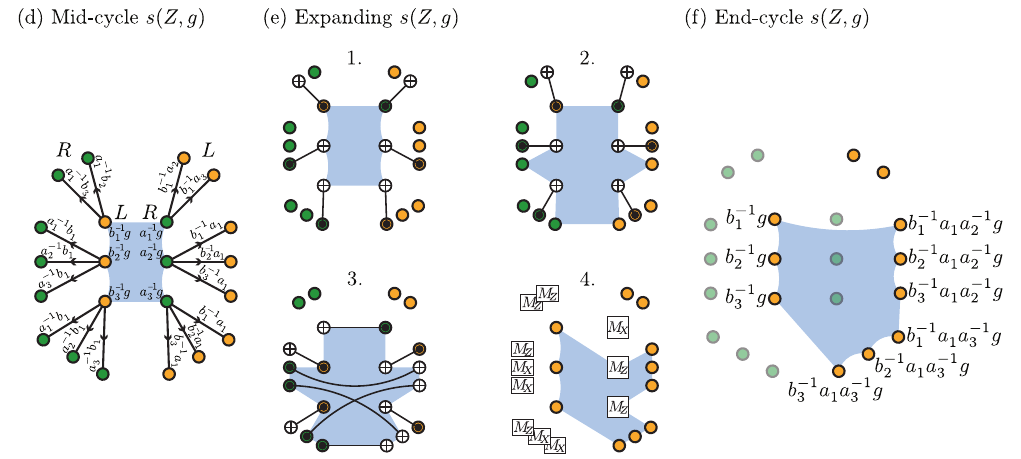}
    \caption{An expanding $X$ and $Z$ stabilizer. (a) The half-cycle support of $s(X,g)$. Also shown are the 16 other qubits that are involved in CNOT gates during the contraction circuit, the group label can be obtained by multiplying the group element in the label of the edge with the group element of the qubit it is connected to. (b) The 4 steps of the contraction circuit $F_{i}$ from the perspective of the expanding $s(X,g)$ stabilizer. The support of the $s(X,g)$ stabilizer is shown in red. (c) The end-cycle support of the stabilizer, representing one of the stabilizer generators of the code $\tilde{C}$. (d--f) The same diagrams for an expanding $Z$ stabilizer $s(Z,g)$.}\label{fig:expanding_stabs}
\end{figure*}

As shown in the main text, when all elements of $AB$ (viewed as a set) are distinct, each stabilizer generator has weight-9, as shown visually in \cref{fig:expanding_stabs}. One can ask if these are the minimum-weight stabilizers of the code. Using numerical optimization --- solved with Gurobi~\cite{Gurobi} similarly to how we found the distance of the end-cycle codes --- we were able to find the minimum-weight stabilizer of each code. For all codes listed in Table I of the main text~\cite{Shaw24Lowering} the computed minimum stabilizer weight stabilizer is 9, \textit{except} the $[[36,12,3]]$ code which has a weight-8 stabilizer.

\subsection{Description of End-cycle Codes as BB Codes}
\label{sec:descripBB}

Here we explain how one can rewrite the stabilizer generators of $\tilde{C}_{1}$ explicitly as those of an Abelian 2BGA code. To do this, we need to define which subsets of qubits to denote as \textit{end-cycle-left} and \textit{end-cycle-right}, for the purpose of writing $\tilde{C}_{1}$ as a 2BGA code. Note that $\tilde{C}_{1}$ itself is defined only on the left qubits of the mid-cycle code $C$.

We define end-cycle-left and -right subsets in terms of the ``even'' and ``odd'' qubits with a slightly generalized definition compared to the main text. Specifically, the left-even qubits are those in the set $q(L,k)$ with $k\in K$, and the left-odd qubits are those with $q(L,k')$, $k'\in K^{c}$. Meanwhile, the right-even qubits are the set $q(R,k)$ with $k\in a_{1}b_{1}K$ and the right-odd qubits are $q(R,k')$ with $k'\in a_{1}b_{1}K^{c}$. With this nomenclature, we can view the right-even qubits as the $X$-ancillas of the end-cycle code $\tilde{C}_{1}$ and the right-odd qubits as the $Z$-ancillas (and vice versa for $\tilde{C}_{2}$). Likewise, we define the left-even and left-odd qubits respectively as the end-cycle-left and end-cycle-right qubits of $\tilde{C}_{1}$ (and vice versa for $\tilde{C}_{2}$).

We also introduce even and odd nomenclature for the stabilizers: the $X$-stabilizer $s(X,g)$ is even if $g\in a_{1}K$ and odd if $g\in a_{1}K^{c}$, while the $Z$-stabilizer $s(Z,g)$ is even if $g\in b_{1}K$ and odd if $g\in b_{1}K^{c}$. With this notation, the even $X$-stabilizers and odd $Z$-stabilizers are contracting in the $F_{1}$ circuit, and the odd $X$- and even $Z$-stabilizers are expanding. Note that all the definitions converge to that in the main text when $a_{1}=b_{1}=1$.

Now, to write $\tilde{C}_{1}$ as a 2BGA code we must first identify the subgroup that labels the qubits and stabilizers. In our case, this is the subgroup $K$, the kernel of $f$. We introduce new end-cycle labels for the qubits and stabilizers that we write as $\tilde{q}(L,k)$, $\tilde{q}(R,k)$, $\tilde{s}(X,k)$ and $\tilde{s}(Z,k)$ for $k\in K$. Fixing a ``relabelling'' (coset representative) element $r\in K^{c}$, we define
\begin{subequations}\label{eq:end-cycle_relabelling}
    \begin{align}
        \tilde{q}(L,k)&=q(L,k),&\tilde{q}(R,k)&=q(L,rk)\\
        \tilde{s}(X,k)&=s(X,a_{1}^{-1}rk),&\tilde{s}(Z,k)&=s(Z,b_{1}k).
    \end{align}
\end{subequations}
Note that these definitions are such that all the expanding stabilizers and mid-cycle-left qubits are given a unique end-cycle label with an element $k\in K$. With these end-cycle labels, the end-cycle stabilizers can be written explicitly as an Abelian 2BGA code, i.e.
\begin{subequations}\label{eq:end-cycle_labeled_stabilizer_support}
    \begin{align}
        \tilde{s}(X,k)&=X(a_{1}^{-1}b_{1}^{-1}ABrk,0)=\tilde{X}(\tilde{A}k,\tilde{B}k)\\
        \tilde{s}(Z,k)&=Z(a_{1}b_{1}A^{-1}B^{-1}k,0)=\tilde{Z}(\tilde{B}^{-1}k,\tilde{A}^{-1}k)
    \end{align}
\end{subequations}
with
\begin{subequations}\label{eq:tilde_A_B}
    \begin{align}
        \tilde{A}&=\big(a_{1}^{-1}(a_{2}+a_{3})+b_{1}^{-1}(b_{2}+b_{3})\big)r,\\
        \tilde{B}&=\big(1+a_{1}^{-1}b_{1}^{-1}(a_{2}+a_{3})(b_{2}+b_{3})\big).
    \end{align}
\end{subequations}

Repeating the exercise to determine the 2BGA representation of $\tilde{C}_{2}$ requires modifying the definitions in~\cref{eq:end-cycle_relabelling} since the roles of even and odd qubits/stabilizers become swapped. With an appropriate choice, one recovers the same equations as in \cref{eq:end-cycle_labeled_stabilizer_support,eq:tilde_A_B}. This is consistent with the observation that $\tilde{C}_{2}$ is related to $\tilde{C}_{1}$ by a shift of any element $r\in K^{c}$, since the 2BGA representation of the code is independent of this shift.

We finish this appendix by stating (without proof) a few bonus propositions that are not used in the rest of the manuscript but may be of interest to some readers. Each of these propositions follows from \cref{prop:mid_to_end_propagation}.

First, we provide a condition analogous to \cref{eq:mid-cycle_logical_condition} for the end-cycle code $\tilde{C}_{1}$.
\begin{proposition}\label{prop:end-cycle_logical_condition}
    In the end-cycle code $\tilde{C}_{1}$, an $X$-operator $X(P,0)$ represents a (possibly trivial) logical operator iff
    \begin{subequations}
        \begin{equation}
            ABP\subseteq a_{1}b_{1}K^{c}.
        \end{equation}
        Likewise, a $Z$-operator $Z(P,0)$ represents a (possibly trivial) logical operator iff
        \begin{equation}
            A^{-1}B^{-1}P\subseteq a_{1}b_{1} K.
        \end{equation}
    \end{subequations}
\end{proposition}

Finally, we present two propositions relating to the transformation of logical operators between the mid-cycle code and the end-cycle codes.

\begin{proposition}
    Any mid-cycle (possibly trivial) logical operator with support only on the left qubits in $C$ is mapped to the same logical operator in both $\tilde{C}_{1}$ and $\tilde{C}_{2}$. The same is true for operators with support only on the right qubits of $C$. 
\end{proposition}

\begin{proposition}
    Any end-cycle (possibly trivial) logical operator with support only on the end-cycle-left qubits in $\tilde{C}_{1}$ (i.e.~the left-even qubits) will be mapped to an operator that has support only on the end-cycle-left qubits in $\tilde{C}_{2}$ (i.e.~the left-odd qubits), and similarly for operators with support on the end-cycle-right qubits.
\end{proposition}

\section{Properties of the Connectivity Graph}\label{sec:connectivity}

In this section, we prove that the connectivity graph of the morphing protocol for any weight-6 BB code that satisfies Crit.~\ref{crit:homomorphism} is biplanar, and discuss how to layout the qubits in a toric$^+$ arrangement. Recall that the connectivity graph $G_{\rm con}=(V,E)$ is the graph consisting of qubits on vertices and edges corresponding to pairs of qubits that participate in the same CNOT during any of the contraction circuits $F_{i}$.
From \cref{tab:contraction_circuits_appendix}, each left qubit $q(L,g)$ participates in a CNOT with the right qubits $q(R,a_{i}^{-1}b_{j}^{\vphantom{-1}}g)$ for $(i,j)\in\{(1,1),(1,2),(1,3),(2,1),(3,1)\}$; while each right qubit $q(R,g')$ participates in a CNOT with the left qubits $q(L,a_{i}^{\vphantom{-1}}b_{j}^{-1}g')$ for the same values of $(i,j)$.
Therefore, the connectivity graph is bipartite (between sets of left and right qubits) and of degree 5.

\subsection{Biplanarity}\label{subsec:biplanarity}

The graph $G_{\rm con}$ is biplanar if one can partition the set of edges $E$ into two subsets $E_{1},E_{2}, E=E_1 \cup E_2$ such that each subgraph $(V,E_{i})$ is planar. We will show this by explicitly providing such a partitioning of the set of edges and demonstrating that each subgraph is planar.

\begin{figure}
    \includegraphics[width = \linewidth]{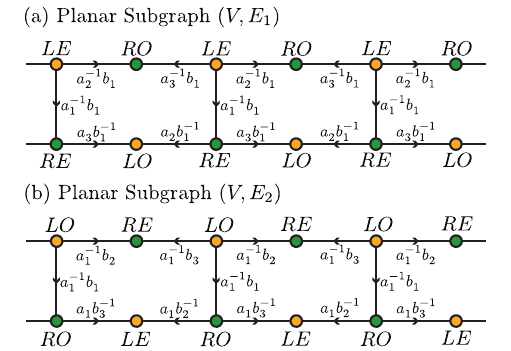}
    \caption{One wheel from each of the planar subgraphs $(V,E_{1})$ and $(V,E_{2})$. Qubits are labeled $LE,LO,RE,RO$ to denote whether they are left or right and even or odd. The edges of the graph are undirected, the arrows are only there to indicate that to go between the qubits you multiply the qubit label by the edge label when traveling in the direction of the arrow. Extending each wheel to the right, each wheel will eventually return to the left-hand side of the figure. The subgraph can therefore be embedded in a planar layout by arranging it as a wheel, with each row of qubits forming a circle and, say, the upper row being inside the bottom row. The full subgraph $(V,E_{i})$ may consist of multiple disjoint wheels.}\label{fig:biplanarity}
\end{figure}

The two subgraphs are shown in \cref{fig:biplanarity} and are constructed as follows, based on \cref{tab:contraction_circuits_appendix}. We use the definitions of ``even'' and ``odd'' qubits introduced in \cref{sec:descripBB}; that is, a left qubit $q(L,g)$ is even if $g\in K$ and odd if $g\in K^{c}$, while a right qubit $q(R,g)$ is even if $g\in a_{1}b_{1}K$ and odd if $g\in a_{1}b_{1}K^{c}$. Moreover, we split the edges into five subsets $E_{(i,j)}$ that each contain the edges connecting $q(L,g)$ to $q(R,a_{i}^{-1}b_{j}^{\vphantom{-1}}g)$ for all $g$. Note that the edge subset $E_{(1,1)}$ contains edges that connect left-even qubits with right-even qubits and left-odd qubits with right-odd qubits, while the other four subsets connect an even and an odd qubit. Therefore, we additionally split $E_{(1,1)}$ into even $E_{(1,1)}^{(e)}$ and odd $E_{(1,1)}^{(o)}$ subsets depending on whether the qubits they connect are even or odd.

Now we define the two planar subgraphs to have edges
\begin{subequations}
    \begin{align}
        E_{1}&=E_{(1,1)}^{(e)}\cup E_{(2,1)}^{\vphantom{(e)}}\cup E_{(3,1)}^{\vphantom{(e)}},\\
        E_{2}&=E_{(1,1)}^{(o)}\cup E_{(1,2)}^{\vphantom{(e)}}\cup E_{(1,3)}^{\vphantom{(e)}}.
    \end{align}
\end{subequations}
By construction, the union of these two edge sets gives the full set of edges. Much like in Ref.~\cite{Bravyi24}, each subgraph consists of a finite number of disjoint \textit{wheels}, as shown in \cref{fig:biplanarity}. Each wheel is planar, and therefore each subgraph is planar, proving the biplanarity of $G_{\rm con}$.

\subsection{Toric$^+$ Layout}\label{subsec:toric+}

Now we turn our attention to the geometric layout of the connectivity graph $G_{\rm con}$. In Ref.~\cite{Bravyi24}, Bravyi \textit{et al.} showed that the standard parity-check protocol for all the codes listed in Table I of the main text~\cite{Shaw24Lowering} can be executed using what we refer to as a toric$^+$ layout. This layout consists of the local edges required to implement the toric code, \textit{plus} some longer-range edges. In particular, since the Tanner graph has degree 6, each qubit must be connected to 4 local and 2 non-local edges.

Here, we wish to find a toric$^{+}$ layout of $G_{\rm con}$ for each of the end-cycle codes presented in Table I of the main text~\cite{Shaw24Lowering}, in which as many edges as possible are local on a torus. We will show that under very similar conditions to Lemma 4 of Ref.~\cite{Bravyi24}, the morphing circuit admits a ``hex-grid rotated'' toric$^+$ layout; that is, consisting of the local edges required to implement the hex-grid rotated toric code (see \cref{fig:surface_code_example} and Ref.~\cite{McEwen23}), plus some longer-range edges. Since the connectivity graph of the morphing protocol has degree 5, each qubit must be connected to 3 local and 2 non-local edges. In other words, the reduction in connectivity compared to the standard protocol involves removing one of the local edges while not increasing the number of non-local edges.

The key tools we will use in our proof come from Theorem 6 in Ref.~\cite{Lin23}, namely that the code is invariant under any transformation
\begin{align}\label{eq:relabelling}
    A&\mapsto \phi(r_{a}A),&B&\mapsto \phi(r_{b}B),
\end{align}
for any relabelling elements $r_{a},r_{b}\in G$, and any bijective group homomorphism $\phi:G\rightarrow G$ (also called a group \textit{automorphism}) applied to each element of $A$ and $B$.

In the morphing protocol, the edges of the connectivity graph are labeled by the group elements $G_{E}=\{e_{(i,j)}=a_{i}^{-1}b_{j}^{\vphantom{-1}}\mid i\text{ or }j=1\}$. Take for example, the morphing protocol in \cref{fig:surface_code_example}, where the mid-cycle code $C$ is the $\lambda\times\mu$ unrotated toric code for some $\lambda,\mu\in\mathbb{Z}$ and the end-cycle code $\tilde{C}_{i}$ is the rotated toric code. Then, the elements in $G_{E}$ can be written $\{1,x,y\}$. Therefore, we say that the morphing protocol for a 2BGA code has a (hex-grid rotated) toric$^+$ layout if it can be rewritten using \cref{eq:relabelling} into the form $G=\mathbb{Z}_{\lambda}\times\mathbb{Z}_{\mu}$ with $\{1,x,y\}\subseteq G_{E}$ for some $\lambda,\mu\in \mathbb{Z}$.

We claim that such a toric$^+$ layout exists if there exists a pair of elements $g_{1}=e_{(i,j)}^{-1}e_{(i',j')}^{\vphantom{-1}}$ and $g_{2}=e_{(i,j)}^{-1}e_{(i'',j'')}^{\vphantom{-1}}$ satisfying the condition that:
\begin{enumerate}
    \item The set $\{g_{1},g_{2}\}$ generates $G$, and
    \item $\lambda\mu=\ell m$, where $\lambda$ and $\mu$ are the smallest positive integers such that $g_{1}^{\lambda}=g_{2}^{\mu}=1$.
\end{enumerate}
This condition is very similar to that given in Lemma 4 of Ref.~\cite{Bravyi24} and is satisfied by every code and every homomorphism listed in Table I of the main text~\cite{Shaw24Lowering}.

To show this, we begin by setting $r_{a}=1$, $r_{b}=e_{(i,j)}^{-1}$. Then, define the homomorphism $\tilde{\phi}:\mathbb{Z}_{\lambda}\times\mathbb{Z}_{\mu}\rightarrow\mathbb{Z}_{\ell}\times\mathbb{Z}_{m}$ with $\tilde\phi(x)=g_{1}$, $\tilde\phi(y)=g_{2}$. From our restrictions on $g_{1}$ and $g_{2}$, every element in $G$ can be written uniquely as a product $g_{1}^{q}g_{2}^{r}$ for integers $0\leq q<\lambda$ and $0\leq r < \mu$, and therefore $\tilde{\phi}$ is bijective. Moreover, $\tilde{\phi}$ being bijective implies that the groups $\mathbb{Z}_{\lambda}\times\mathbb{Z}_{\mu}$ and $\mathbb{Z}_{\ell}\times\mathbb{Z}_{m}$ are isomorphic, and hence $\tilde{\phi}$ is an automorphism from $G\rightarrow G$. We define the relabelling automorphism as $\phi=\tilde{\phi}^{-1}$. Now, note that under these relabellings, $e_{(i,j)}=a_{i}^{-1}b_{j}^{\vphantom{-1}}\mapsto \phi(r_{a}^{-1}a_{i}^{-1}r_{b}^{\vphantom{-1}}b_{j}^{\vphantom{-1}})=\phi(r_{b}e_{(i,j)})$. Applying these definitions to each element in $G_{E}$ gives
\begin{subequations}
    \begin{align}
        G_{E}&=\{e_{(i,j)},e_{(i',j')},e_{(i'',j'')}\dots\}\\
        &\mapsto \{\phi(r_{b}e_{(i,j)}),\phi(r_{b}e_{(i',j')}),\phi(r_{b}e_{(i'',j'')}),\dots\}\\
        &=\{\phi(1),\phi(e_{(i,j)}^{-1}e_{(i',j')}^{\vphantom{-1}}),\phi(e_{(i,j)}^{-1}e_{(i'',j'')}^{\vphantom{-1}}),\dots\}\\
        &=\{1,x,y,\dots\},
    \end{align}
\end{subequations}
as required. As a consequence, every code and every homomorphism in Table I of the main text~\cite{Shaw24Lowering} admits a hex-grid rotated toric$^+$ layout.

Note that the biplanarity of $G_{\rm con}$ does \textit{not} depend on the geometric layout: the existence of one geometric layout in which $G_{\rm con}$ is biplanar is enough to guarantee that $G_{\rm con}$ is biplanar in any other layout. However, it may be the case that if the qubits are placed in the toric$^+$ layout, the edges cannot travel in a geometrically straight line while still being biplanar. It is therefore not clear whether the toric$^+$ layout is the optimal layout for qubits in a biplanar system.

\section{Logical Operations}\label{sec:logical_operations}

In this section, we show how to input and output (I/O) an arbitrary quantum state between a surface code patch and any of the $k$ logical qubits encoded in the BB code under the morphing protocol. Such capabilities are enough to guarantee universal quantum computation with BB codes, since arbitrary logical qubits can be stored in memory in the BB code and then teleported into a surface code to perform logical gates.

In Ref.~\cite{Bravyi24} the authors show how to perform such I/O in a BB code using three operations. The first is what we refer to as a \textit{shift automorphism}~\footnote{Note here that the term ``automorphism'' refers to an automorphism of the \textit{stabilizer group} of the code, not the Abelian group $G$ that defines the BB code.}, whereby the label of every qubit and stabilizer generator is shifted via multiplication by a group element $g$: $X(P,Q)\mapsto X(gP,gQ)$ and $Z(P,Q)\mapsto Z(gP,gQ)$. The second is a $ZX$-duality that for BB codes corresponds to implementing the transformation $X(P,Q)\mapsto Z(Q^{-1},P^{-1})$, $Z(P,Q)\mapsto X(Q^{-1},P^{-1})$ (up to a shift automorphism). And third, Bravyi \textit{et al.} show how to teleport the logical state of \textit{one} of the qubits in the BB code onto the surface code using the scheme of Ref.~\cite{Cohen22}, albeit with a relatively large qubit overhead.

These three operations allow I/O of an arbitrary logical state due to the structure of the logical operators in ``primed'' and ``unprimed'' blocks. To construct the primed and unprimed blocks of logical qubits, we begin by fixing some subsets $P_{1},P_{2},Q_{2}\subseteq G$. Then, the primed logical operators are given by 
\begin{equation}
 \overline{X}=X(P_{1}g,0), \overline{Z}=Z(P_{2}g,Q_{2}g),\; \forall g\in G,
 \label{eq:struc-p}
\end{equation} while the unprimed logical operators are given by 
\begin{equation}
 \overline{X}=X(Q_{2}^{-1}g,P_{2}^{-1}g), \overline{Z}=Z(0,P_{1}^{-1}g),\; \forall g\in G.
 \label{eq:struc-up}
\end{equation} 
Bravyi \textit{et al.} showed that for the codes in Table I of the main text~\cite{Shaw24Lowering}, there exists a choice of $P_{1},P_{2},Q_{2}$ such that the primed and unprimed logical operators generate the entire logical group. Given the ability to perform I/O of one of the primed logical qubits, one can then access the remaining primed logical qubits via shift automorphisms, and the unprimed logical qubits via the $ZX$-duality.

One advantage of the approach taken by Bravyi \textit{et al.}\ is that these three operations can be performed in a biplanar layout. Specifically, the shift automorphisms and $ZX$-duality do not require any additional connectivity beyond what is required already for implementing the standard parity-check schedule, while the I/O of a single primed logical qubit to a surface code can be implemented in a biplanar layout. However, the $ZX$-duality proposed is currently extremely impractical, requiring on the order of 100 rounds of CNOT gates for the $[[144,12,12]]$ code, during which no error-correction can be performed~\cite{Bravyi24}.

In this section, we show how to perform a shift automorphism and the I/O of two logical qubits encoded using our morphing protocol, while leaving the possibility of $ZX$-dualities to future work. In particular, we show that shift automorphisms can be performed using the existing connectivity required to perform error correction, while I/O can be implemented using separate ancillary systems for an arbitrary number of logical qubits all within a biplanar layout. One caveat is that we require each logical operator to satisfy some technical restrictions on their support, and we find numerically that such logical operators exist for the three smaller BB codes from Table I of the main text~\cite{Shaw24Lowering}. We envisage building two separate I/O ancillary systems that interface with one primed and one unprimed logical qubit each, removing the need to implement a potentially extremely noisy $ZX$-duality. Moreover, the I/O scheme that we design implements a morphing protocol to perform error correction on all of the ancillary systems, reducing the connectivity requirements compared to Ref.~\cite{Bravyi24}. Note that we have no particular reason to believe that implementing a $ZX$-duality would be more difficult using a morphing protocol than in the standard protocol, and we do not analyze it here due to the need for more work optimizing its design.

\subsection{Shift Automorphisms}\label{subsec:logical_shifts}

We implement a shift automorphism using a similar approach as in Ref.~\cite{Bravyi24}. In particular, we will describe how to perform shifts of qubit labels corresponding to the group elements $g=a_{i}^{\vphantom{-1}}a_{j}^{-1}$ (an $A$-type shift) and $g=b_{i}^{\vphantom{-1}}b_{j}^{-1}$ (a $B$-type shift) for $i,j=1,2,3$. In Lemma 3 of Ref.~\cite{Bravyi24}, the authors show that these group elements generate $G$ whenever the Tanner graph of the BB code is connected. We perform the shift in between QEC cycles, while the left qubits are encoded in the end-cycle code $\tilde{C}_{i}$ and the right qubits are free to be used as ancilla qubits. Moreover, we are only able to use CNOTs between qubits that are adjacent in the connectivity graph of the morphing protocol; that is, between qubits $q(L,g)$ and $q(R,a_{i}^{-1}b_{j}^{\vphantom{-1}}g)$ for $(i,j)=(1,1),(1,2),(1,3),(2,1),(3,1)$.

The primitive operation we will use is a fault-tolerant SWAP between a data and an ancilla qubit, given by the circuit
\begin{equation}\label{eq:swap}
\includegraphics{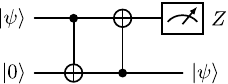}
\end{equation}
To perform the $A$-type shift $a_{i}^{\vphantom{-1}}a_{j}^{-1}$, we begin with each right qubit reset to $\ket{0}$. We use \cref{eq:swap} to swap the qubits $q(L,g)\leftrightarrow q(R,a_{j}^{-1}b_{1}^{\vphantom{-1}}g)$ for all $g\in G$. Then, we reset the left qubits to $\ket{0}$, and perform the swaps $q(R,g)\leftrightarrow q(L,a_{i}^{\vphantom{-1}}b_{1}^{-1}g)$. This completes the permutation $q(L,g)\rightarrow q(L,a_{i}^{\vphantom{-1}}a_{j}^{-1}g)$. To perform the $B$-type permutation $b_{i}^{\vphantom{-1}}b_{j}^{-1}$, we perform the same procedure except that the first set of swaps is between $q(L,g)\leftrightarrow q(R,a_{1}^{-1}b_{i}^{\vphantom{-1}}g)$ and the second set is between $q(R,g)\leftrightarrow q(L,a_{1}^{\vphantom{-1}}b_{j}^{-1}g)$.

At this point, we must discuss some subtleties regarding the implementation of these shift automorphisms. Let us write $U(g)$ to represent the unitary operator that shifts each qubit by a group element $g\in G$. We need to ensure that a shift automorphism performed on the end-cycle code $\tilde{C}_{i}$ also implements a shift automorphism of the mid-cycle code $C$, thereby cycling through the primed and unprimed blocks of logical operators. First, note that the contraction circuits $F_{i}$ commute with $U(g)$ only when $g\in K$; for $g\in K^{c}$ we have $U(g)^{\dag}\circ F_{i}\circ U(g)=F_{i+1}$ (with addition taken mod 2). The same is true for the measurements $M_{i}$ and resets $R_{i}$. So, when $g\in K$, we can apply $U(g)$ on the end-cycle code $\tilde{C}_{i}$, and the composite circuit $F_{i}^{\dag}\circ R_{i}^{\vphantom{\dag}}\circ U(g)\circ M_{i}^{\vphantom{\dag}}\circ F_{i}^{\vphantom{\dag}}$ will apply the shift $U(g)$ to the mid-cycle code $C$. However, if $g\in K^{c}$, $U(g)$ will map the end-cycle code $\tilde{C}_{i}\mapsto \tilde{C}_{i+1}$. Therefore, the circuit that applies the shift $U(g)$ to $C$ is instead $F_{i+1}^{\dag}\circ R_{i+1}^{\vphantom{\dag}}\circ U(g)\circ M_{i}^{\vphantom{\dag}}\circ F_{i}^{\vphantom{\dag}}$. After this, one would continue with QEC as usual starting again from the contraction circuit $F_{i}$.

Finally, if the $A$- or $B$-shift being implemented has $i$ or $j$ equal to 1, we can remove the need to perform one of the swaps by modifying the contraction circuit $F_{i}$ either immediately before or after the shift. In particular, we show in \cref{subsec:round_3_reversal} that with a simple modification, one can prepare the end-cycle code $\tilde{C}_{i}^{(R)}$ instead of $\tilde{C}_{i}$. $\tilde{C}_{i}^{(R)}$ is identical to $\tilde{C}_{i}$ except that the qubits $q(L,g)$ and $q(R,a_{1}^{-1}b_{1}^{\vphantom{-1}}g)$ are swapped. If this swap is used as the first step of an $A$- or $B$-shift, we can implement it instead using this modification to $F_{i}$. If this swap is the second step of the $A$- or $B$-shift, we can implement it by modifying $F_{i}^{\dag}$.

\subsection{Logical Input/Output}\label{subsec:logical_I/O}

In this subsection, we describe how to implement the I/O of an arbitrary logical qubit from the BB code to an ancillary surface code using a linking code described by Ref.~\cite{Cohen22}. We use the following circuits~\cite{Breuckmann_2017,Xu24}, obtained from one-bit teleportation circuits with CNOTs replaced by $XX$ and $ZZ$ measurements using a linking code, to perform input
\begin{subequations}\label{eq:I/O_circuits}
\begin{equation}
\includegraphics{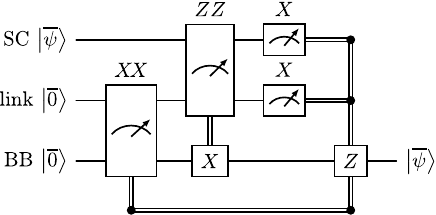}
\hspace{-0.5 cm}
\end{equation}
and output
\begin{equation}
\includegraphics{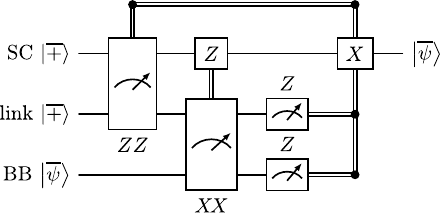}
\hspace{-0.5 cm}
\end{equation}
\end{subequations}
where SC refers to the surface code.

To use the circuits in \cref{eq:I/O_circuits}, we need to be able to implement a logical $\overline{Z}_{\text{BB}}$ measurement of the BB code and a joint $\overline{X}_{\text{link}}\overline{X}_{\text{BB}}$ measurement between the BB code and the linking code. Both of these can be done by performing lattice surgery with a designed linking code, described in Ref.~\cite{Cohen22}. The linking code depends on the structure of the logical operator being measured, and we therefore require a separate ancillary linking code system for each logical measurement. We will focus on how to perform a $\overline{X}_{\text{link}}\overline{X}_{\text{BB}}$ measurement, since a $\overline{Z}_{\text{BB}}$ measurement can be performed in an analogous way by measuring $\overline{Z}_{\text{link}}\overline{Z}_{\text{BB}}$ with the linking code initialized in $\ket{\overline{0}}$.

We will begin by discussing in general terms how to perform lattice surgery using a morphing protocol, with the hex-grid rotated surface code as a simple example. Then, we will apply this to BB codes and explicitly construct the required linking code and morphing circuit. Finally, we show that the connectivity graph is biplanar --- even for an arbitrary number of independent I/O apparatuses --- and that the degree of the connectivity graph is one fewer than the corresponding graph in Ref.~\cite{Bravyi24}.

\subsubsection{Lattice Surgery with the Morphing Protocol}\label{subsubsec:lattice_surgery}

We consider the following picture of measuring $\overline{X}\overline{X}$ using lattice surgery~\cite{Horsman12,Vuillot_2019}. We begin with two CSS codes that we together call the \textit{split} codes that undergo QEC separately. Then, we consider a \textit{merged} code that contains all the physical qubits of the split codes plus some number of interface data qubits that are ``in between'' the two split codes. The lattice surgery protocol proceeds by initializing all interface qubits in $\ket{0}$, then measuring the stabilizers of the merged code for $O(d)$ number of rounds, and then measuring the interface qubits in the $Z$-basis. Following this, the split codes are again operated independently. 

For such a protocol to work, we need the following.
\begin{appendixcriterion}\label{crit:lattice_surgery}\,
\end{appendixcriterion}
\begin{enumerate}[label=(\alph*)]
    \item \textit{The logical operator $\overline{X}\overline{X}$ is contained in the stabilizer group of the merged code, $\mathcal{S}_{\text{merged}}$, so that it can be inferred from the merged stabilizer values;}
    \item \textit{each $X$-stabilizer of the split codes is contained in $\mathcal{S}_{\text{merged}}$, so that $X$-detectors can be constructed in the merged code; and}
    \item \textit{every $Z$-stabilizer of $\mathcal{S}_{\text{merged}}$ can be written as some product of $Z$-stabilizers of the split codes and $Z$-operators on the interface qubits, so that $Z$-detectors can be constructed in the merged code.}
\end{enumerate}
It may be possible to consider other general descriptions of lattice surgery, but we find this description sufficient for our present discussion. Note also that Crit.~\ref{crit:lattice_surgery}(c) follows from Crit.~\ref{crit:lattice_surgery}(b): since every $Z$-stabilizer of $\mathcal{S}_{\text{merged}}$ must commute with every $X$-stabilizer of $\mathcal{S}_{\text{merged}}$ and hence every $X$-stabilizer of $\mathcal{S}_{\text{split}}$, we have that every $Z$-stabilizer of $\mathcal{S}_{\text{merged}}$ is contained in $\mathcal{S}_{\text{split}}$ up to multiplication by $Z$-operators on the interface qubits (of course assuming that $\mathcal{S}_{\text{merged}}$ does not contain any $Z$ logical operators of the split code).

Now, suppose that we have designed morphing protocols with two contraction circuits ($I=2$) for both split codes \textit{and} the merged code. That is, the split codes $C_{\text{split}}$ and the merged code $C_{\text{merged}}$ are mid-cycle codes, where we have combined the two disconnected split codes into the single notation $C_{\text{split}}$ for convenience. We now have four distinct split end-cycle codes $\tilde{C}_{i,\text{split}}$ and two merged end-cycle codes $\tilde{C}_{i,\text{merged}}$ for $i=1,2$. How do we use these morphing circuits to implement the lattice surgery protocol?

We propose to simply transition from the split to the merged codes during one of the end-cycle codes, in an entirely analogous way to a standard protocol, see \cref{fig:surface_code_lattice_surgery} for the surface code example. Explicitly, we begin by running QEC on the separate split codes, for some number of rounds. From the mid-cycle codes $C_{\text{split}}$ we run the contraction circuits $M_{i,\text{split}}\circ F_{i,\text{split}}$ and end up in the end-cycle code $\tilde{C}_{i,\text{split}}$. From here, we initialize any interface data qubits in $\ket{0}$ (we label this as $R_{\text{int}}$), and begin running morphing parity-check circuits for the merged code via $M_{i+1,\text{merged}}^{\vphantom{\dag}}\circ F_{i+1,\text{merged}}^{\vphantom{\dag}}\circ F_{i,\text{merged}}^{\dag}\circ R_{i,\text{merged}}^{\vphantom{\dag}}$. We continue for any number of rounds until we are in another merged end-cycle code $\tilde{C}_{i',\text{merged}}$, at which point we measure the interface qubits in the $Z$-basis (called $M_{\text{int}}$) and continue performing QEC on the split codes. For this to work, we need to ensure that Crit.~\ref{crit:lattice_surgery} applies to the \textit{end-cycle} split and merged codes instead of the \textit{mid-cycle} codes --- something that is not guaranteed and must be checked separately.


\begin{figure*}
    \includegraphics[width=\linewidth]{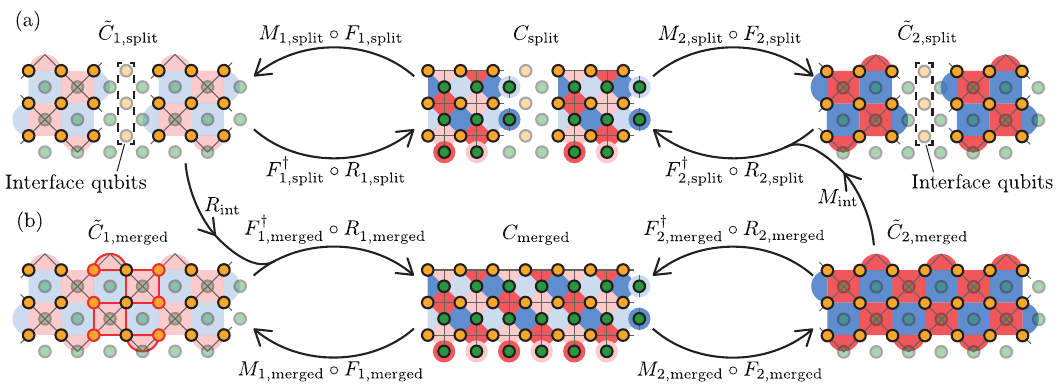}
    \caption{Surface code lattice surgery using morphing circuits. The $X$- and $Z$-stabilizers are represented by red and blue faces, respectively. stabilizers are opaque if they are contracted in $F_{1}$, and partly transparent if they are contracted in $F_{2}$. Data and ancilla qubits are colored yellow and green, respectively. Qubits are partly transparent if they are not in use at a given time step. The grey grid denotes the toric code lattice representation of the code, with $X$-stabilizers on vertices, $Z$-stabilizers on faces, and qubits on edges. The morphing protocol corresponds to that of \cref{fig:surface_code_example} but with boundaries; in this case, the circuits coincide exactly with the hex-grid rotated surface code of Ref.~\cite{McEwen23}. (a) The split morphing circuits of two surface codes. (b) The merged morphing circuits of one larger surface code, used for lattice surgery. We can initiate the logical measurement from the end-cycle codes by resetting the interface data qubits in the $\ket{0}$ state ($R_{\text{int}}$) before proceeding with the morphing parity-check schedule of the merged code. We represent this above for just the $\tilde{C}_{1,\text{split}}$ end-cycle code, but of course this could be done for the $\tilde{C}_{2,\text{split}}$ code as well. The logical operator $\tilde{X}_{1}\tilde{X}_{1}$ of the split end-cycle code $\tilde{C}_{1,\text{split}}$ can be written as a product of interface $X$-stabilizers that are contracting in $F_{2}$, marked with a red outline. All of these stabilizers are measured after a single parity-check cycle in the merged code. After repeating the merged parity-check schedule for $O(d)$ rounds to ensure fault-tolerance, we return to the split codes by measuring the interface data qubits in the $Z$-basis ($M_{\text{int}}$). }\label{fig:surface_code_lattice_surgery}
\end{figure*}

We now describe how the above morphing protocol allows us to infer the logical measurement outcome of the \textit{end-cycle} codes. Consider one parity-check cycle of the merged morphing circuit from $\tilde{C}_{i,\text{merged}}\rightarrow C_{\text{merged}}\rightarrow\tilde{C}_{i+1,\text{merged}}$ (with $i+1$ taken mod 2). We will write the mid-cycle support of the logical measurement as $\overline{XX}$ and the end-cycle support as $\tilde{X}_{i}\tilde{X}_{i}$. In the mid-cycle code $C_{\text{merged}}$, $\overline{XX}$ is equal to some product of $X$-stabilizers that we call \textit{interface} $X$-stabilizers since they have support ``in between'' the two split codes. Some of the interface $X$-stabilizers are contracting in $F_{i,\text{merged}}$ (i.e.~they are in the stabilizer generator subset $S_{i,\text{merged}}$) and some are expanding (i.e.~they are in $S_{i+1,\text{merged}}$). 
The unitarity of the contraction circuit $F_{i,\text{merged}}$ ensures that in the ``post-reset'' code, i.e.~the code obtained after the resets $R_{i,\text{merged}}$, the product of the $X$ logical operators is still equal to the product of all the interface $X$-stabilizers. However, the interface $X$-stabilizers in $S_{i,\text{merged}}$ only have support on the recently-reset qubits in $R_{i,\text{merged}}$ and therefore have a deterministic $+1$ eigenvalue. As such, in the end-cycle code $\tilde{C}_{i,\text{merged}}$, i.e.~before the resets $R_{i,\text{merged}}$, the product of logical operators $\tilde{X}_{i}\tilde{X}_{i}$ is only a product of the expanded interface $X$-stabilizers that are in $S_{i+1,\text{merged}}$. These are precisely the same interface stabilizers that are contracted by $F_{i+1,\text{merged}}$ and are included in the following round of measurements $M_{i+1,\text{merged}}$. Therefore we can infer the value of the end-cycle logical operator $\tilde{X}_{i}\tilde{X}_{i}$ by taking the product of the contracted interface $X$-stabilizer measurement outcomes in the measurement round $M_{i+1,\text{merged}}$.

In \cref{fig:surface_code_lattice_surgery} we show how to perform lattice surgery on the surface code using morphing circuits. In this case, the mid-cycle codes are unrotated surface codes and the end-cycle codes are rotated surface codes, simplifying the analysis. Note that in the mid-cycle merged code, the $\overline{XX}$ logical operator is the product of both expanding and contracting interface $X$-stabilizers. However, in the end-cycle codes, the logical operator $\tilde{X}_{i}\tilde{X}_{i}$ is the product only of interface $X$-stabilizers in $S_{i+1,\text{merged}}$ and are therefore all measured in the next round $M_{i+1,\text{merged}}$.

\subsubsection{Mid-cycle Linking Code and Mid-cycle Merged Code}

Having given a general description of lattice surgery with morphing circuits, we now turn our attention to the specific linking code that will be used to perform lattice surgery with the BB code. In particular, we will begin by constructing the mid-cycle linking code and merged code that measures $\overline{X}_{\text{link}}\overline{X}_{\text{BB}}$. In \cref{sec:morphing-p} we will then define the morphing protocols and check that the end-cycle split and merged codes satisfy Crit.~\ref{crit:lattice_surgery}.

Before continuing, we will make a few important assumptions about the mid-cycle logical $\overline{X}=X(P,Q)$ operator being measured. These are summarized in the following criterion.
\begin{appendixcriterion}\label{crit:logical_criterion}
\,
\end{appendixcriterion}
    \begin{enumerate}[label=(\alph*)]
        \item \textit{Every $Z$-stabilizer $s(Z,g)$ that intersects $\overline{X}$ does so on at most two qubits; that is, $|(B^{-1}g)\cap P|+|(A^{-1}g)\cap Q|\leq 2$ for all $g\in G$; and}
        \item \textit{The two overlapping qubits are not $q(L,b_{1}^{-1}g)$ and $q(R,a_{1}^{-1}g)$; that is, $P\cap(a_{1}^{\vphantom{-1}}b_{1}^{-1}Q)=0$.}
    \end{enumerate}
These properties are guaranteed for any logical operator with support only on the left qubits (or only on the right qubits), as is the case for the $\overline{X}$ logical operators in the primed logical block, see \cref{eq:struc-p}. We have checked that all unprimed logical $\overline{X}$ operators can be generated by logical operators with these properties for the $[[72,12,6]]$, $[[108,8,10]]$ and $[[144,12,12]]$ mid-cycle BB codes listed in Table I of the main text~\cite{Shaw24Lowering} by writing the restrictions in Crit.~\ref{crit:logical_criterion} as linear constraints and using the Gurobi numerical optimization package~\cite{Gurobi}. The required optimization is quite numerically intensive for the $[[288,12,18]]$ code, but we have no reason to believe that such logical operators would not exist.

\begin{figure}
\centering
    \includegraphics[width = \linewidth]{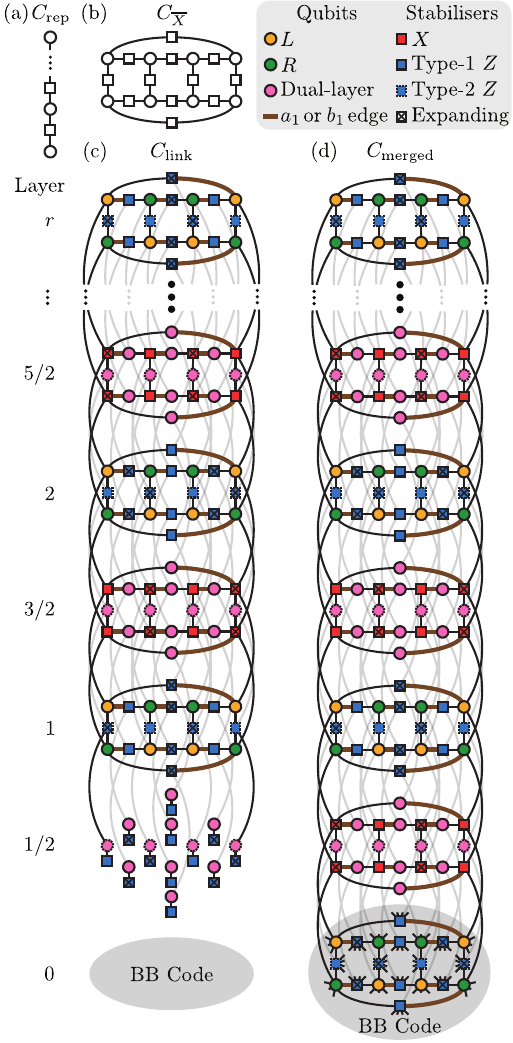}
    \caption{Tanner graphs of the linking and merged mid-cycle codes and their construction as a hypergraph product (HGP) code.
    (a) The Tanner graph of the $r$-bit classical repetition code $C_{\text{rep}}$ with bits as circles and checks as squares. (b) The BB code restricted to qubits in the support of a single logical $\overline{X}$ operator and adjacent $Z$-stabilizers $C_{\overline{X}}$. (c) The linking code. Layers $\rho=1,3/2,\dots,r$ correspond exactly to the HGP code between $C_{\text{rep}}$ and $C_{\overline{X}}$, while layer $\rho=1/2$ corresponds to $Z$-ancillas. Layer 0 consists of the separate BB code. (d) The merged code. Layers $\rho=0$ and $\rho=1,3/2,\dots,r$ are the same as in (c), but layer $\rho=1/2$ now contains the interface $X$-stabilizers that allow the measurement of $\overline{X}_{\text{link}}\overline{X}_{\text{BB}}$. Also displayed are the $a_{1}$ and $b_{1}$ edges (brown, bold), expanding stabilizers in the first contraction circuit (cross), and the Type-2 $Z$-stabilizers and dual-layer qubits (dashed). Some edges are shaded for clarity.}\label{fig:linking_and_merged_codes}
\end{figure}

The linking mid-cycle code itself is most succinctly described as a hypergraph product (HGP) code between two classical codes, as shown in \cref{fig:linking_and_merged_codes}(a)-(c). The first classical code is the $r$-bit repetition code $C_{\text{rep}}$. The positive integer $r$ here is arbitrary, but in general must be $\geq d$ to guarantee fault-tolerance~\cite{Cohen22}. The second classical code $C_{\overline{X}}$ has bits corresponding to the qubits in the support of the $\overline{X}=X(P,Q)$ operator, and checks corresponding to the $Z$-checks that act on any of the qubits in $X(P,Q)$. From Crit.~\ref{crit:logical_criterion}(b), it is guaranteed that each check of $C_{\overline{X}}$ has weight two, and each bit is contained in three checks. The HGP code between $C_{\text{rep}}$ and $C_{\overline{X}}$ is the linking code $C_{\text{link}}$.

Although such a description as an HGP is sufficient to construct the stabilizers of the code, we will need a more detailed description to construct the morphing circuits in \cref{sec:morphing-p}. To this end, we split the linking code up into $2r-1$ \textit{layers} from $\rho=1,3/2,\dots,r$. Layers labeled by an integer $\rho=1,2,\dots,r$ are called ``primal'' layers and correspond in the HGP construction to bits of the repetition code, while layers labeled by a half-integer $\rho=3/2,5/2,\dots,r-1/2$ are called ``dual'' layers and correspond to checks of the repetition code. In each primal layer, we have a qubit for each qubit in the support of $X(P,Q)$, and a $Z$-stabilizer for each $Z$-stabilizer in the BB code that intersects $X(P,Q)$. In each dual layer, these are swapped: we have qubits corresponding to $Z$-stabilizers of the BB code, and $X$-stabilizers corresponding to each qubit in the support of $X(P,Q)$. Finally, we will also need to use an additional ancillary layer $\rho=1/2$ consisting of dual-layer qubits that are each stabilized by a single-qubit $Z$-stabilizer.

Let $H\subseteq G$ label the subset of $Z$-stabilizers $s(Z,h)$ that have overlapping support with $X(P,Q)$. Explicitly, we have
\begin{equation}\label{eq:defH}
    H=\bigcup_{i=1}^{3}(b_{i}P\cup a_{i}Q).
\end{equation}
Moreover, let $\tilde{H}=H\setminus\big(b_{1}P\cup a_{1}Q\big)=\{h\in H\mid b_{1}^{-1}h\notin P,a_{1}^{-1}h\notin Q\}$.
Then, in summary, the linking code has the following sets of qubits: labeled labeled labeled
\begin{itemize}
    \item \textit{Primal-layer left qubits} $q(\rho,L,p)$, for $p\in P$ and $\rho=1,\dots,r$;
    \item \textit{Primal-layer right qubits} $q(\rho,R,q)$, for $q\in Q$ and $\rho=1,\dots,r$; and
    \item \textit{Dual-layer qubits} $q(\rho,h)$, for $h\in H$ and $\rho=1/2,\dots,r-1/2$.
\end{itemize}
Moreover, the linking code has the following sets of stabilizers:
\begin{itemize}
    \item \textit{Left $X$-stabilizers} $s(X,\rho,L,p)$, for $p\in P$ and $\rho=3/2,\dots,r-1/2$;
    \item \textit{Right $X$-stabilizers} $s(X,\rho,R,q)$, for $q\in Q$ and $\rho=3/2,\dots,r-1/2$;
    \item \textit{Type-1 $Z$-stabilizers} $s(Z,\rho,h)$ for $h\in H\setminus\tilde{H}$ and $\rho=1,\dots,r$;
    \item \textit{Type-2 $Z$-stabilizers} $s(Z,\rho,\tilde{h})$ for $\tilde{h}\in\tilde{H}$ and $\rho=1,\dots,r$; and
    \item \textit{Ancillary $Z$-stabilizers} $s(Z,0,h)$ for $h\in H$ (specified below).
\end{itemize}

The support of each of these stabilizers is given graphically in both \cref{fig:linking_and_merged_codes,fig:linking_contracting_stabilizers}. \cref{fig:linking_and_merged_codes}(c) shows the Tanner graph of $C_{\text{link}}$ (which is \textit{not} the same as the connectivity graph) that corresponds to the logical $\overline{X}$-operator $X(\{x^{2},y^{3},y^{5},xy^{5}\},\{1,xy,y^{5},x^{4}y^{5}\})$ of the [[72,12,6]] BB code (with morphing protocol defined by the $f_{xy}$ homomorphism). Explicitly, each $X$-stabilizer has weight five: each left $X$-stabilizer $s(X,\rho,L,p)$ has support on the primal qubits $q(\rho\pm1/2,L,p)$ and dual qubits $q(\rho,b_{i}p)$ for $i=1,2,3$, while each right $X$-stabilizer $s(X,\rho,R,q)$ has support on $q(\rho\pm1/2,R,q)$ and $q(\rho,a_{i}q)$. The ancillary $Z$-stabilizers $s(Z,0,h)$ have support only on the qubit $q(1/2,h)$. Each non-ancillary $Z$-stabilizer $s(Z,\rho,h)$ is guaranteed to have weight four (three on the $\rho=r$ boundary) by Crit.~\ref{crit:logical_criterion}(b), with support on dual-layer qubits $q(\rho\pm1/2,h)$ and whichever two primal-layer qubits can be written $q(\rho,L,b_{j}^{-1}h)$ for $b_{j}^{-1}h\in P$ or $q(\rho,R,a_{j}^{-1}h)$ for $a_{j}^{-1}h\in Q$. Moreover, we split the $Z$-stabilizers into two subsets that will be important for the morphing protocol: type-1 contains $Z$-stabilizers that have support on either $q(\rho,L,b_{1}^{-1}h)$ or $q(\rho,R,a_{1}^{-1}h)$, while type-2 contains $Z$-stabilizers that have support on neither of those qubits. Note that Crit.~\ref{crit:logical_criterion}(a) guarantees that no $Z$-stabilizers will have support on both $q(\rho,L,b_{1}^{-1}h)$ and $q(\rho,R,a_{1}^{-1}h)$. We also define type-1 and type-2 dual-layer qubits $q(\rho,h)$ depending on whether $h\in H\setminus\tilde{H}$ or $h\in \tilde{H}$, respectively.

The linking code has a logical $\overline{X}_{\text{link}}$ representative with support on qubits $q(1,L,p)$ and $q(1,R,q)$ for all $p\in P$ and $q\in Q$. This $\overline{X}_{\text{link}}$ representative is ``facing'' the BB code so that it can be used in lattice surgery in the merged code. The linking code also has a $\overline{Z}_{\text{link}}$ representative with support on qubits $q(\rho,L,p)$ for all $\rho=1,\dots,r$ and for a fixed $p\in P$. This representative is ``vertical'' in \cref{fig:linking_and_merged_codes} and can be ``facing'' a surface code so that it can be used for the $\overline{Z}_{\text{SC}}\overline{Z}_{\text{link}}$ lattice surgery measurement. There are also similar logical $\overline{Z}_{\text{link}}$ representatives for each $p\in P$ and $q\in Q$.

Next, we define the merged mid-cycle code $C_{\text{merged}}$, shown in \cref{fig:linking_and_merged_codes}(d), that is used to perform the $\overline{X}_{\text{link}}\overline{X}_{\text{BB}}$ measurement. The merged mid-cycle code is very similar to the split linking and BB codes but with the following modifications. First, we add a new set of \textit{interface} $X$-stabilizers on the dual layer $\rho=1/2$ labeled by $s(X,1/2,L,p)$ and $s(X,1/2,R,q)$. Their support is the same as the other dual-layer $X$-stabilizers of the linking code, but the $\rho=0$ layer of qubits is identified with qubits in the BB code itself. The product of all these interface $X$-stabilizers is $\overline{X}_{\text{link}}\overline{X}_{\text{BB}}$, as required. However, they commute with neither the ancillary $Z$-stabilizers $s(Z,0,h)$ of the linking code nor the $Z$-stabilizers $s(Z,h)$ of the BB code. Therefore, in the merged mid-cycle code we merge pairs of $Z$-stabilizers with the same label $h \in H$ forming a single, weight-7 stabilizer 
\begin{align}\label{eq:merged-Z-stab}
   s(Z,0,h)_{\text{merged}}\equiv s(Z,h)_{\text{merged}}=s(Z,0,h)_{\text{link}}s(Z,h)_{\text{BB}}. 
\end{align} We refer to both the new $X$-stabilizers and the merged $Z$-stabilizers as \textit{interface} stabilizers.

\subsubsection{Morphing Protocols}
\label{sec:morphing-p}

\begin{figure*}
\centering
    \includegraphics[width = \linewidth]{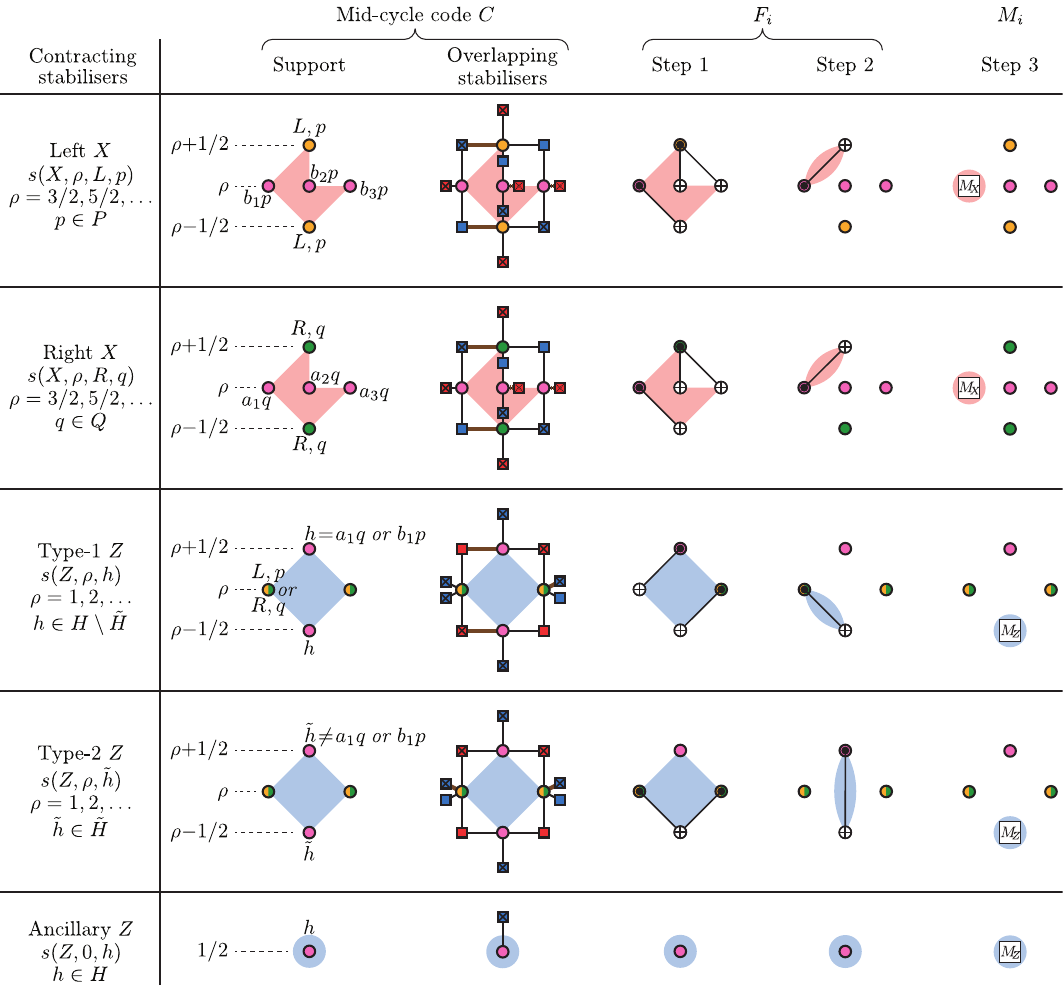}
    \caption{Visual representation of the contraction circuits for the linking code in \cref{tab:linking_code_contraction}, from the perspective of the five types of contracting stabilizers. Qubits are represented by circles, colored yellow (primal layer left), green (primal layer right), or pink (dual layer). The weight of each stabilizer is guaranteed to be as shown due to Crit.~\ref{crit:logical_criterion}, but the precise identity of some of the qubits is not always the same. In the ``Overlapping stabilizers'' column, the $X$- and $Z$-stabilizers that overlap with the contracting stabilizer are represented by red and blue boxes (respectively) and connected to the overlapping qubits. The overlapping stabilizers that are eXpanding are marked with a cross, those that are contracting are without a cross, and those that could be either with a dashed cross. Furthermore, edges between qubits and stabilizers that are related by the multiplication of $a_{1}$ or $b_{1}$ are formatted brown and bold. The three steps of the circuit $M_{i}\circ F_{i}$ are shown, with the shaded region representing the support of the contracting stabilizer before the operations in the step are implemented. Note that the contraction circuits for $Z$-stabilizers $s(Z,r,h)$ are the same as those with $\rho<r$ but with the top dual-layer qubit, and the related CNOTs, removed.} \label{fig:linking_contracting_stabilizers}
\end{figure*}

We can now describe our morphing protocol for the linking code $C_{\text{link}}$. Our protocol has two contraction circuits $F_{i}$ and two end-cycle codes $\tilde{C}_{i,\text{link}}$, and assumes of course that the BB code satisfies the homomorphism Crit.~\ref{crit:homomorphism_appendix}. Which stabilizers --- each labeled by a group element --- are contracting in a circuit is determined by which of the two cosets $K=\ker f$ or $K^{c}=G\setminus \ker f$ the group element is contained in. Specifically, writing $K_{\rho}=K$ if $\rho$ is even and $K_{\rho}=K^{c}$ if $\rho$ is odd, the following stabilizers are contracting in $F_{1}$:
\begin{itemize}
    \item the left $X$-stabilizers $s(X,\rho,L,p)$ for $p\in K_{\rho+1/2}\cap P$,
    \item the right $X$-stabilizers $s(X,\rho,R,q)$ for $q\in (a_{1}b_{1}K_{\rho+1/2})\cap Q$, and
    \item the $Z$-stabilizers $s(Z,\rho,h)$ for $h\in (b_{1}^{\vphantom{c}}K_{\rho}^{c})\cap H$.
\end{itemize}
For $F_{2}$, the contracting stabilizers are given by replacing $K_{\rho}^{\vphantom{c}}\leftrightarrow K_{\rho}^{c}$. Note that, for example, the coset $a_{1}b_{1}K$ is equal either to $K$ or $K^{c}$, depending on the value of $f(a_{1}b_{1})$.

\begin{table}
    \caption{Definition of the contracting circuits $F_{i,\text{link}}$ and measurements $M_{i,\text{link}}$ for the linking code, again assuming Crit.~\ref{crit:homomorphism_appendix} is satisfied. The circuit is defined in terms of three steps, the first of which contains a set of commuting but not-necessarily-simultaneously-executable CNOTs (step 2 is simultaneously executable). The gates are applied for all $\rho=1,2,\dots,r-1$ and $\tau=1,2,\dots,r$. The circuit $M_{1,\text{link}}\circ F_{1,\text{link}}$ is defined by applying the gates below for each $p^{(\rho)}\in K_{\rho}^{c}\cap P$, $q^{(\rho)}\in (a_{1}^{\vphantom{c}}b_{1}^{\vphantom{c}}K_{\rho}^{c})\cap Q$, and $\tilde{h}^{(\rho)}\in (b_{1}^{\vphantom{c}}K_{\rho}^{c})\cap \tilde{H}$; while for $M_{2,\text{link}}\circ F_{2,\text{link}}$ one chooses $p^{(\rho)}\in K_{\rho}\cap P$, $q^{(\rho)}\in (a_{1}b_{1}K_{\rho})\cap Q$, and $\tilde{h}^{(\rho)}\in (b_{1}^{\vphantom{c}}K_{\rho})\cap \tilde{H}$. In both circuits, every qubit in the first dual layer $q(1/2,h)$ is measured for all $h\in H$.}\label{tab:linking_code_contraction}
    \renewcommand{\arraystretch}{1.35}
    \setlength{\tabcolsep}{6pt}
    \begin{tabular}{|c|c|}
    \hline
        \multirow{6}{*}{Step 1}& CNOT$\big(q(\rho+1/2,b_{1}p^{(\rho)}),q(\rho,L,p^{(\rho)})\big)$,\\
        & CNOT$\big(q(\tau,L,p^{(\tau{+}1)}),q(\tau-1/2,b_{2}p^{(\tau{+}1)})\big)$,\\
        & CNOT$\big(q(\tau,L,p^{(\tau{+}1)}),q(\tau-1/2,b_{3}p^{(\tau{+}1)})\big)$,\\
        & CNOT$\big(q(\rho+1/2,a_{1}q^{(\rho)}),q(\rho,R,q^{(\rho)})\big)$,\\
        & CNOT$\big(q(\tau,R,q^{(\tau{+}1)}),q(\tau-1/2,a_{2}q^{(\tau{+}1)})\big)$,\\
        & CNOT$\big(q(\tau,R,q^{(\tau{+}1)}),q(\tau-1/2,a_{3}q^{(\tau{+}1)})\big)$\\\hline
        \multirow{5}{*}{Step 2}& CNOT$\big(q(\rho+1/2,b_{1}p^{(\rho)}),q(\rho+1,L,p^{(\rho)})\big)$,\\
        & CNOT$\big(q(\rho+1/2,a_{1}q^{(\rho)}),q(\rho+1,R,q^{(\rho)})\big)$,\\
        & CNOT$\big(q(\tau,L,p^{(\tau)}),q(\tau-1/2,b_{1}p^{(\tau)})\big)$,\\
        & CNOT$\big(q(\tau,R,q^{(\tau)}),q(\tau-1/2,a_{1}q^{(\tau)})\big)$,\\
        & CNOT$\big(q(\rho+1/2,\tilde{h}^{(\tau)}),q(\rho-1/2,\tilde{h}^{(\tau)})\big)$\\
        \hline
        \multirow{6}{*}{Step 3}& $M_{X}\big(q(\rho+1/2,b_{1}p^{(\rho)})\big)$,\\
        &$M_{X}\big(q(\rho+1/2,a_{1}q^{(\rho)})\big)$,\\
        &$M_{Z}\big(q(\rho-1/2,b_{1}p^{(\rho)})\big)$,\\
        &$M_{Z}\big(q(\rho-1/2,a_{1}q^{(\rho)})\big)$,\\
        &$M_{Z}\big(q(\rho-1/2,\tilde{h}^{(\rho)})\big)$,\\
        &$M_{Z}\big(q(1/2,h)\big)$\\
        \hline
    \end{tabular}
\end{table}

We list all the gates involved in \cref{tab:linking_code_contraction}, but a more intuitive picture is given by \cref{fig:linking_contracting_stabilizers}. There, we show each of the five types of stabilizers in the linking code and the local contraction circuits we use to contract each of them. We sort the gates and measurements into ``steps'' instead of ``rounds'' to indicate that each step contains a set of CNOTs that commute but are not necessarily simultaneously executable.

Note the special role played by the elements $a_1$ and $b_1$ in the construction. For example, each contracting left $X$-stabilizer $s(X,\rho,L,p)$ has support on qubits labeled $q(\rho,b_{1}p)$, $q(\rho,b_{2}p)$ and $q(\rho,b_{3}p)$. In the first step, the support of $s(X,\rho,L,p)$ on these qubits is removed by CNOTs with the primal-layer qubits $q(\rho\pm1/2,L,p)$; specifically, $q(\rho,b_{1}p)$ participates in a CNOT with $q(\rho-1/2,L,p)$, while $q(\rho,b_{2}p)$ and $q(\rho,b_{3}p)$ participate in a CNOT with $q(\rho+1/2,L,p)$. This asymmetry between the qubit $q(\rho,b_{1}p)$ and the qubits $q(\rho,b_{2}p)$ and $q(\rho,b_{3}p)$ is why we need to split the $Z$-stabilizers into the type-1 and type-2 subsets, since the pattern of CNOT gates that are shared with adjacent $X$-stabilizers is different. In the Tanner graph in \cref{fig:linking_and_merged_codes} we explicitly show which of the edges between a stabilizer and a qubit represent multiplication by either $a_{1}$ or $b_{1}$. With this, it is easy to see which $Z$-stabilizers are type-2 since they are not connected by any $a_{1}$ or $b_{1}$ edges in the Tanner graph.

One can verify that the contraction circuit is indeed valid by considering a contracting stabilizer $s$ and each of the other adjacent stabilizers that overlap with it, as shown in the ``Overlapping stabilizers'' column of \cref{fig:linking_contracting_stabilizers}. By considering the contraction circuits in \cref{tab:linking_code_contraction,fig:linking_contracting_stabilizers}, one can convince oneself that whenever there is an adjacent contracting stabilizer $s'$, the contraction circuit for $s$ does not interfere with the contraction circuit for $s'$ by inadvertently expanding it. With this property satisfied, \cref{tab:linking_code_contraction} defines a valid contraction circuit for the linking code.

In the end-cycle linking code $\tilde{C}_{i,\text{link}}$, all the primal-layer qubits correspond to data qubits and \textit{most} of the dual-layer qubits correspond to ancilla qubits. For each layer $\rho=3/2,5/2,\dots,r-1/2$, all the type-1 dual-layer qubits $q(\rho,h)$ with $h\in H\setminus\tilde{H}$ are used to measure either an $X$-stabilizer or a type-1 $Z$-stabilizer, depending on which coset $K_{\rho+1/2}$ or $K_{\rho+1/2}^{c}$ $h$ falls in. Meanwhile, the type-2 dual qubits $q(\rho,\tilde{h})$ with $\tilde{h}\in\tilde{H}$ are only measured when the $Z$-stabilizer $s(Z,\rho+1/2,\tilde{h})$ is contracting. Therefore, there are some type-2 dual-layer qubits that remain data qubits in $\tilde{C}_{i,\text{link}}$; these have labels $\tilde{h}\in (b_{1}K_{\rho+1/2})\cap\tilde{H}$. Note that in the linking code, all the ancillary dual-layer qubits $q(1/2,h)$ are measured every round.

\begin{figure*}
\centering
    \includegraphics[width = \linewidth]{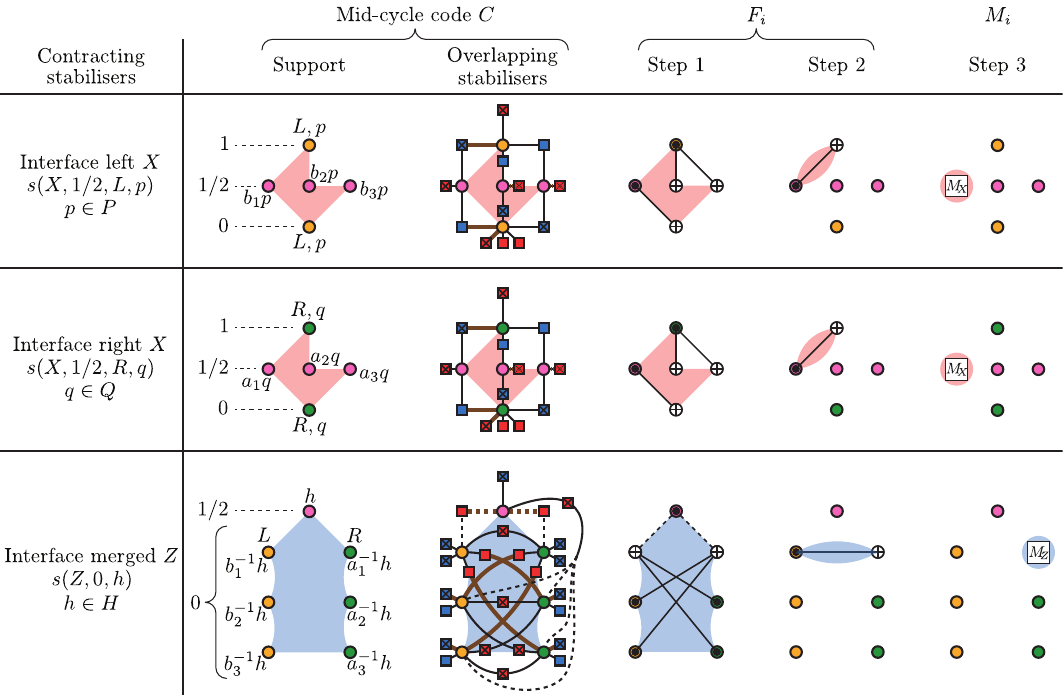}
    \caption{Visual representation of the contraction circuits for the interface stabilizers of the merged code in \cref{tab:merged_code_contraction}, from the perspective of contracting stabilizers. Our notation is the same as in \cref{fig:linking_contracting_stabilizers}. For the interface merged $Z$-stabilizer, the details of the overlapping $X$-stabilizers depend on the value of $h$. The merged $Z$-stabilizer will intersect nine stabilizers in the BB code as per usual, as well as two interface $X$-stabilizers that intersect the $Z$-stabilizer on two qubits, one of which will always be the dual qubit $q(1/2,h)$. In step 1, only one of the dashed CNOT gates is implemented for each $Z$-stabilizer, depending on which of the qubits the $X$-stabilizers overlap on.}\label{fig:merged_contracting_stabilizers}
\end{figure*}

The contraction circuits for the merged code, $C_{\text{merged}} \rightarrow \tilde{C}_{i,\text{merged}}$, are only slightly modified compared to the morphing protocols of the separate (split) linking and BB codes, as listed in \cref{tab:merged_code_contraction} and shown visually in \cref{fig:merged_contracting_stabilizers}. To obtain the gates for the merged code, we simply extend the contraction pattern for the linked code down to the new interface $X$-stabilizers. For the type-1 merged $Z$-stabilizers $s(Z,0,h)$ with $h\in H\setminus\tilde{H}$, defined in \cref{eq:merged-Z-stab}, these gates ensure that the merged $Z$-stabilizer is contracted correctly. However, we also need to add corresponding gates for the type-2 merged $Z$-stabilizers $s(Z,0,\tilde{h})$ with $\tilde{h}\in \tilde{H}$. One can check again that with our definitions of contracting stabilizers for the linking and the BB code, these circuits form a valid morphing circuit for the merged code.

\begin{table}
    \caption{Definition of the contracting circuits $F_{i}$ and measurements $M_{i}$ for the merged BB and linking code, assuming Crit.~\ref{crit:homomorphism_appendix} is satisfied. The circuit is defined with respect to the contraction circuits for the BB code in \cref{tab:contraction_circuits_appendix} and the linking code in \cref{tab:linking_code_contraction}. The additional gates for $F_{1}$ are defined for all $p\in K^{c}\cap P$, $q\in (a_{1}b_{1}K^{c})\cap Q$, $h\in (b_{1}^{\vphantom{c}}K^{c})\cap(H\setminus\tilde{H})$, $h'\in (b_{1}K)\cap(H\setminus\tilde{H})$, $\tilde{h}\in (b_{1}^{\vphantom{c}}K^{c})\cap \tilde{H}$ and $\tilde{h}'\in (b_{1}K)\cap\tilde{H}$. The sets for $F_{2}$ are found by swapping $K\leftrightarrow K^{c}$.}\label{tab:merged_code_contraction}
    \renewcommand{\arraystretch}{1.35}
    \setlength{\tabcolsep}{6pt}
    \begin{tabular}{|c|c|}
    \hline
        \multirow{5}{*}{Step 1}& Rounds 1 and 2 from \cref{tab:contraction_circuits_appendix},\\
        &Step 1 from \cref{tab:linking_code_contraction},\\
        & CNOT$\big(q(1/2,b_{1}p),q(L,p)\big)$,\\
        & CNOT$\big(q(1/2,a_{1}q),q(R,q)\big)$,\\
        & CNOT$\big(q(1/2,\tilde{h}),q(L,a_{1}\tilde{h})\big)$\\\hline
        \multirow{4}{*}{Step 2}& Round 3 from \cref{tab:contraction_circuits_appendix},\\
        & Step 2 from \cref{tab:linking_code_contraction},\\
        & CNOT$\big(q(1/2,b_{1}p),q(1,L,p)\big)$,\\
        & CNOT$\big(q(1/2,a_{1}q),q(1,R,q)\big)$\\
        \hline
        \multirow{7}{*}{Step 3}& Round 4 from \cref{tab:contraction_circuits_appendix},\\
        &Step 3 from \cref{tab:linking_code_contraction},\\
        &\textit{but with the $\rho=1/2$}\\[-4 pt]
        &\textit{measurements replaced by:}\\
        &$M_{X}\big(q(1/2,h)\big)$,\\
        &$M_{Z}\big(q(1/2,h')\big)$,\\
        &$M_{Z}\big(q(1/2,\tilde{h}')\big)$\\
        \hline
    \end{tabular}
\end{table}

It is important to note that due to the modifications of the measurements on the $\rho=1/2$ layer, the end-cycle merged code $\tilde{C}_{i,\text{merged}}$ may have more physical qubits than the end-cycle split codes $\tilde{C}_{i,\text{split}}=\tilde{C}_{i,\text{link}}\cup\tilde{C}_{i,\text{BB}}$. Specifically, in the linking code, all the $\rho=1/2$ layer qubits represented $Z$-ancillas that were measured every round. In the merged code, they now obey the same pattern of measurement as the other dual-layer qubits with $\rho\geq3/2$; that is, the type-2 dual-layer qubits with $\tilde{h}\in (b_{1}K^{c})\cap \tilde{H}$ are data qubits in $\tilde{C}_{i,\text{merged}}$. These correspond to the \textit{interface data qubits} discussed in \cref{subsubsec:lattice_surgery} and will be the only qubits involved in the interface $\ket{0}$-resets $R_{i,\text{int}}$ and $Z$-measuments $M_{i,\text{int}}$.

In order for the morphing circuits for the split and merged codes to perform lattice surgery correctly, we need to ensure that the \textit{end-cycle} codes $\tilde{C}_{i,\text{split}}$ and $\tilde{C}_{i,\text{merged}}$ satisfy Crit.~\ref{crit:lattice_surgery}. By construction, the mid-cycle codes $C_{i,\text{split}}$ and $C_{i,\text{merged}}$ satisfy Crit.~\ref{crit:lattice_surgery}, and we will now argue that the structure of the contraction circuits guarantees that the end-cycle codes do too. It is straightforward to show that Crit.~\ref{crit:lattice_surgery}(a) is satisfied using arguments from \cref{subsubsec:lattice_surgery}. For Crit.~\ref{crit:lattice_surgery}(b), one can check that all the additional gates in steps 1 and 2 of \cref{tab:merged_code_contraction} \textit{commute} with all the $X$-stabilizers of the split linking and BB codes (not including, of course, the interface $X$-stabilizers that are in the merged code but not the split codes). This means that the end-cycle support of any $X$-stabilizer of $\tilde{C}_{i,\text{split}}$ corresponds exactly to an $X$-stabilizer in $\tilde{C}_{i,\text{merged}}$. Finally, as discussed in \cref{subsubsec:lattice_surgery}, Crit.~\ref{crit:logical_criterion}(c) follows from Crit.~\ref{crit:logical_criterion}(b), and therefore lattice surgery works for the end-cycle split and merged codes.

We leave a distance analysis of the end-cycle linking code $\tilde{C}_{i,\text{link}}$ and the end-cycle merged code $\tilde{C}_{i,\text{merged}}$ to future work. However, we do note that by \cref{prop:distance_lower_bound}, the distance of the end-cycle linking code is lower-bounded by $\tilde{d}_{\text{link}}\geq d_{\text{link}}/3$, where $d_{\text{link}}$ is the distance of the corresponding mid-cycle code linking code; while the distance of the end-cycle merged code is lower bounded by $\tilde{d}_{\text{merged}}\geq d_{\text{merged}}/4$.

\subsubsection{Biplanarity}

\begin{figure}
\centering
    \includegraphics[width = \linewidth]{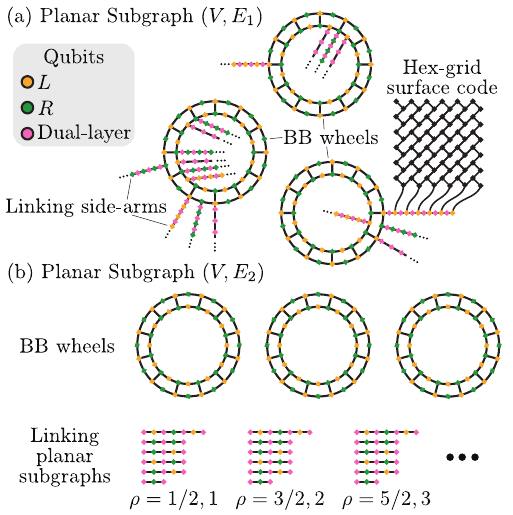}
    \caption{The biplanar layout for the $\overline{X}_{\text{link}}\overline{X}_{\text{BB}}$ and $\overline{Z}_{\text{SC}}\overline{Z}_{\text{link}}$ logical measurements of \cref{eq:I/O_circuits}. Each planar subgraph $(V,E_{1})$ and $(V,E_{2})$ contains the ``wheel graphs'' of the BB code from \cref{fig:biplanarity}, along with additional vertices and edges for the linking code. In $(V,E_{1})$, the linking code edges constitute a set of ``side-arms'' extending from one of the vertices of the wheel graph; while in $(V,E_{2})$, the linking code edges constitute a disjoint set of planar subgraphs for each pair of layers $\rho=k-1/2,k$ for $k=1,\dots,r$. The connectivity graph of the hex-grid surface code, from Ref.~\cite{McEwen23} (see also \cref{fig:surface_code_example,fig:surface_code_lattice_surgery}), used for the $\overline{Z}_{\text{SC}}\overline{Z}_{\text{link}}$ measurement, can be connected either in the subgraph $(V,E_{1})$ or $(V,E_{2})$, here we show it in $(V,E_{1})$. Any number of additional linking codes and/or surface codes connecting to different, potentially overlapping logical operators within the BB code can be accommodated simply by adding more side-arms to $(V,E_{1})$ and more disjoint planar subgraphs to $(V,E_{2})$.}\label{fig:full_biplanarity}
\end{figure}

We now show that the lattice surgery morphing protocol for measuring $\overline{X}_{\text{link}}\overline{X}_{\text{BB}}$ has a biplanar connectivity graph. First, recall from \cref{subsec:biplanarity} that the BB code itself is biplanar, where each planar subgraph consists of a set of disjoint ``wheel graphs''. Now, we must add the qubits of the linking code, and the following edges to the biplanar BB graph. We can read off the new edges by considering the CNOT gates in \cref{tab:linking_code_contraction,tab:merged_code_contraction} in both contraction circuits $F_{1}$ and $F_{2}$ and split the edges into two edge subsets, $E_{1}$ and $E_{2}$, to be added to the existing planar subgraphs of the BB code. For all $p\in P$, $q\in Q$ and $\tilde{h}\in \tilde{H}$, these edges are
\begin{enumerate}
    \item $\big(q(\rho,L,p),q(\rho+1/2,b_{1}p)\big)$ for $\rho=0,1,\dots,r-1$, which we place in $E_{1}$;
    \item $\big(q(\rho,b_{j}p),q(\rho+1/2,L,p)\big)$ for $\rho=1/2,3/2,\dots,r-1/2$, which we place in $E_{1}$ for $j=1$ and in $E_{2}$ for $j=2,3$;
    \item $\big(q(\rho,R,q),q(\rho+1/2,a_{1}q)\big)$ for $\rho=0,1,\dots,r-1$, which we place in $E_{1}$;
    \item $\big(q(\rho,a_{j}q),q(\rho+1/2,R,q)\big)$ for $\rho=1/2,3/2,\dots,r-1/2$, which we place in $E_{1}$ for $j=1$ and in $E_{2}$ for $j=2,3$;
    \item $\big(q(L,a_{1}^{-1}\tilde{h}),q(1/2,\tilde{h})\big)$, which we place in $E_{1}$; and
    \item $\big(q(\rho,\tilde{h}),q(\rho+1,\tilde{h})\big)$ for $\rho=1/2,3/2,\dots,r-1/2$, which we place in $E_{1}$.
\end{enumerate}

The biplanarity of the resulting graph is shown in \cref{fig:full_biplanarity}; this shows an example linking code for the measurement of an $\overline{X}_{\text{link}}\overline{X}_{\text{BB}}$ and a $\overline{Z}_{\text{SC}}\overline{Z}_{\text{link}}$ for the [[144,12,12]] mid-cycle BB code. In the $(V,E_{1})$ graph, we have chosen to take edges that span the ``vertical'' dimension of the linking code. When $\tilde{h}\in\tilde{H}$, it is straightforward to see that the edges in items 5 and 6 above form a ``side-arm'' extending out from the BB qubit $q(L,a_{1}^{-1}\tilde{h})$ and passing through the dual-layer qubits $q(\rho,\tilde{h})$. The story is similar for when $h\in H\setminus\tilde{H}$, except now the legs pass through both dual-layer qubits $q(\rho,h)$ and primal-layer qubits $q(\rho,L,b_{1}^{-1}h)$ or $q(\rho,R,a_{1}^{-1}h)$ (but not both due to Crit.~\ref{crit:logical_criterion}). Meanwhile, the leftover edges form the graph $(V,E_{2})$ that now consists of a set of disjoint subgraphs each containing qubits in layers $\rho-1/2$ and $\rho$, for $\rho=1,\dots,r$ (alongside the original wheel graphs from the BB code). Moreover, due to Crit.~\ref{crit:logical_criterion}, each vertex has degree at most two in $(V,E_{2})$, and therefore the $(V,E_{2})$ graph is planar.

Moreover, it is possible to incorporate any number of separate linking codes used for the lattice surgery of different logical operators into this biplanar structure. Each separate linking code will add more side-arms to the $(V,E_{1})$ graph, and more disjoint planar graphs to the $(V,E_{2})$ graph. Neither of these additions compromises the biplanarity of the overall graph, even if the logical operators of the BB code overlap on multiple qubits, although in this case, the maximum degree of the connectivity graph would increase.

It is straightforward to design a morphing protocol to measure $\overline{Z}_{\text{SC}}\overline{Z}_{\text{link}}$ between the linking code and the surface code in an analogous way to above. The mid-cycle code of the surface code ancilla is the \textit{unrotated} surface code, while the end-cycle code is the \textit{rotated} surface code, as shown in \cref{fig:surface_code_lattice_surgery}. The morphing protocol corresponds to that of \cref{fig:surface_code_example}, but with the addition of boundaries, which makes it exactly equal to the hex-grid rotated surface code of Ref.~\cite{McEwen23}. The required connectivity straightforwardly fits into the biplanar layout of the linking and BB codes; specifically, we must connect each qubit in the support of $\bar{Z}_{\text{link}}$ to a qubit in the hex-grid surface code. This can be done in either of the planar subgraphs, in \cref{fig:full_biplanarity} we show it in $(V,E_{1})$ for simplicity. Using the logic of the previous paragraph one can incorporate any number of separate I/O apparatuses for different logical qubits into a biplanar layout. Thus if one constructs the required ancillary systems to read out one primed and one unprimed logical qubit, combined with the shift automorphisms in \cref{subsec:logical_shifts}, this allows the I/O of all logical qubits in a biplanar layout.

We finish by pointing out two (possible) improvements that this morphing lattice surgery scheme has compared to the standard protocol as in Ref.~\cite{Bravyi24} (although note that \cite{Bravyi24} did not provide parity-check circuits for the full linking code in the standard protocol). The first is that the degree of each qubit in the BB, linking, and surface codes is reduced by one compared to the equivalent construction in Ref.~\cite{Bravyi24}. Indeed, for a single logical $\overline{X}_{\text{link}}\overline{X}_{\text{BB}}$ measurement, the Tanner graph of the linking code in \cref{fig:linking_and_merged_codes}(c) has vertices of degree 4 and 5, while the connectivity graph of morphing scheme has vertices of degree 3 and 4. Moreover, the Tanner graph of the merged code in \cref{fig:linking_and_merged_codes}(d) has a maximum degree of seven (these correspond to qubits and $Z$-stabilizers in the BB code that connect to the linking code), while the maximum degree of the connectivity graph is six. The second, possible improvement is that the qubit overhead of each I/O apparatus could be reduced. This is because of the use of the rotated surface code as the end-cycle code instead of the unrotated surface code. Therefore, it might be possible for $r$ to be chosen to be $\tilde{d}$ -- the distance of the end-cycle BB code --- if the distance of the end-cycle linking code does not go below $\tilde{d}$. 

\section{Modifications to the BB Morphing Procedure}\label{sec:modifications}
In this section we discuss two modifications to the BB morphing procedure described in \cref{tab:contraction_circuits} (see \cref{tab:contraction_circuits_appendix} for the schedule without assuming $a_{1}=b_{1}=1$). The first and most minor modification (\cref{subsec:round_3_reversal}) involves reversing the direction of the CNOTs in the final round of each contraction circuit $F_{i}$, such that the ancilla and data qubits in the end-cycle code are swapped. We will show that the end-cycle code is unaffected by such a modification, and therefore one does not change the error-correction properties of the circuit. This modification may be useful to reduce leakage in qubits by swapping data and ancilla qubits between each round, or to reduce the number of steps required to perform a shift of the logical operators as described in~\cref{subsec:logical_shifts}.

The second modification, discussed in \cref{subsec:round_2_reversal}, involves reversing the direction of the CNOTs in the second round of each contraction circuit $F_{i}$, and adjusting the CNOTs in the last round accordingly. This modification is more significant and can change both the distance of the end-cycle code and the circuit-level distance of the protocol. Applying this modification can sometimes result in an end-cycle code with both a larger distance and a larger circuit-level distance upper-bound than shown in~Table I of the main text~\cite{Shaw24Lowering}. Surprisingly however, numerical simulations reveal that its performance under uniform circuit-level noise is worse than the code in Table I of the main text~\cite{Shaw24Lowering}, at least for the noise rates probed by our simulations.

\subsection{Reversing the CNOTs in Round 3}\label{subsec:round_3_reversal}

\begin{table}
    \caption{Definition of the contracting circuits $F_{i}$ and measurements $M_{i}$ assuming Crit.~\ref{crit:homomorphism_appendix} is satisfied, where the CNOTs in round 3 have been reversed compared to \cref{tab:contraction_circuits_appendix,fig:general_contracting_stabilizers}. As a result, the measurements $M_{i}$ are also modified compared to \cref{tab:contraction_circuits_appendix}. The circuit $F_{1}$ is defined by applying the gates below for each $g\in K$ and $h\in K^{c}$, while for $F_{2}$ one chooses $g\in K^{c}$ and $h\in K$.}\label{tab:round_3_reversal}
    \renewcommand{\arraystretch}{1.35}
    \setlength{\tabcolsep}{6pt}
    \begin{tabular}{|c|c|}
    \hline
        \multirow{2}{*}{Round 1}& CNOT$\big(q(L,g),q(R,a_{1}^{-1}b_{3}^{\vphantom{-1}}g)\big),$\\
        & CNOT$\big(q(R,a_{3}^{-1}b_{1}^{\vphantom{-1}}h),q(L,h)\big)$\\\hline
        \multirow{2}{*}{Round 2}& CNOT$\big(q(L,g),q(R,a_{1}^{-1}b_{2}^{\vphantom{-1}}g)\big),$\\
        & CNOT$\big(q(R,a_{2}^{-1}b_{1}^{\vphantom{-1}}h),q(L,h)\big)$\\
        \hline
        \multirow{2}{*}{Round 3}& CNOT$\big(q(L,g),q(R,a_{1}^{-1}b_{1}^{\vphantom{-1}}g)\big),$\\
        &CNOT$\big(q(R,a_{1}^{-1}b_{1}^{\vphantom{-1}}h),q(L,h)\big)$\\\hline
        Round 4& $M_{X}\big(q(L,g)\big)$, $M_{Z}\big(q(L,h)\big)$\\\hline
    \end{tabular}
\end{table}

We present the modification in \cref{tab:round_3_reversal}. To understand how this modification works, consider the code after the first two rounds of the contraction circuit $F_{i}$, which are the same in both \cref{tab:contraction_circuits_appendix,tab:round_3_reversal}. Recall that the contracting stabilizers are labeled by elements in the contracting subsets
\begin{subequations}
\begin{align}
    G_{X,1}&=a_{1}K,&G_{Z,1}&=b_{1}K^{c},\\
    G_{X,2}&=a_{1}K^{c},&G_{Z,2}&=b_{1}K,
\end{align}
\end{subequations}
(these are the same as in \cref{eq:contracting_subsets}). After the first two rounds of CNOTs in $F_{i}$, all the contracting $X$- and $Z$-stabilizers have support on just two qubits, with:
\begin{subequations}
    \begin{align}
        s(X,g)_{\rm{Round 2}}&=X(a_{1}g,b_{1}g),\\
        s(Z,g')_{\rm{Round 2}}&=Z(b_{1}^{-1}g',a_{1}^{-1}g'),
    \end{align}
\end{subequations}
for $g\in G_{X,i}$ and $g'\in G_{Z,i}$. Moreover, these stabilizers have disjoint supports, i.e.~whenever $g\in G_{X,i}$ and $g'\in G_{Z,i}$, we have $a_{1}^{\vphantom{-1}}g\neq b_{1}^{-1}g'$ and $b_{1}^{\vphantom{-1}}g\neq a_{1}^{-1}g'$. For example, for $i=1$, if $g\in G_{X,1}=a_{1}K$, then we have $a_{1}g\in K$ and $b_{1}g\in a_{1}b_{1}K$, while if $g'\in G_{Z,1}=b_{1}K^{c}$, then we have $b_{1}^{-1}g'\in K^{c}$ and $a_{1}^{-1}g'\in a_{1}b_{1}K^{c}$. Consider for a moment just the contracting $X$-stabilizer $s(X,g)_{\rm{Round 2}}$. There are two circuits we wish to consider, corresponding to the two different directions of the Round 3 CNOT in \cref{tab:contraction_circuits_appendix,tab:round_3_reversal}, given by
\begin{subequations}
\begin{equation}\label{eq:round_3_no_reverse}
\includegraphics{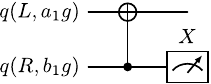}
\end{equation}
in \cref{tab:contraction_circuits_appendix}, and by
\begin{equation}\label{eq:round_3_with_reverse}
\includegraphics{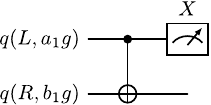}
\end{equation}
\end{subequations}
in \cref{tab:round_3_reversal}. After \cref{eq:round_3_no_reverse}, the end-cycle code $\tilde{C}_{i}^{(L)}\equiv \tilde{C}_i$ is encoded in the left qubits, while after \cref{eq:round_3_with_reverse} the end-cycle code $\tilde{C}_{i}^{(R)}$ is encoded in the right qubits.

We claim that these two end-cycle codes $\tilde{C}_{i}^{(L)}$ and $\tilde{C}_{i}^{(R)}$ are identical up to swapping the qubits $q(L,a_{1}g)$ and $q(R,b_{1}g)$. Indeed, simple circuit identities show that
\begin{equation}\label{eq:round_3_reverse_circuit_identities}
\includegraphics{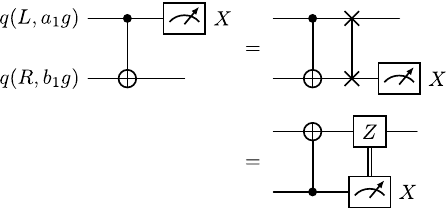}
\end{equation}
Similar circuits hold for any contracting $Z$-stabilizer.
Since the expanding stabilizers and logical operators of $C$ commute with the contracting $X$-stabilizer, they cannot flip the measurement outcome in \cref{eq:round_3_reverse_circuit_identities} and therefore the conditional $Z$ operator is never triggered. Therefore, the stabilizers and logical operators of $\tilde{C}_{i}^{(L)}$ and $\tilde{C}_{i}^{(R)}$ are identical, up to swapping the qubits $q(L,a_{1}g)$ and $q(R,b_{1}g)$. In principle, one could even make an independent choice of which direction to implement the CNOT in Round 3 for each contracting stabilizer, although we do not see a reason why this would be preferable.

The main application of this reversal is in leakage reduction, similar to the proposal in Ref.~\cite{McEwen23} for the surface code. In particular, swapping data and ancilla qubits between each QEC round means that every qubit will be reset every two rounds, thereby removing any leakage. We can achieve this swapping by defining $F_{1}$ using \cref{tab:contraction_circuits_appendix} and $F_{2}$ using \cref{tab:round_3_reversal}. With this, the end-cycle code $\tilde{C}_{1}$ is defined over the left qubits and the measurements $M_{1}$ over the right qubits, while the code $\tilde{C}_{2}$ is defined over the right qubits and the measurements $M_{2}$ over the left.

A secondary application is to reduce the number of swaps required to perform a shift automorphism, as discussed in \cref{subsec:logical_shifts}. In particular, performing the reversal as in \cref{tab:round_3_reversal} results in the end-cycle code $\tilde{C}_{i}$ but with the qubits $q(L,a_{1}g)$ and $q(R,b_{1}g)$ swapped. This swap is the first step in many of the shift automorphisms in \cref{subsec:logical_shifts} and therefore removes the need to explicitly perform the swap.

\subsection{Reversing the CNOTs in Round 2}\label{subsec:round_2_reversal}

\begin{table}
    \caption{Definition of the contracting circuits $F_{i}$ and measurements $M_{i}$ assuming Crit.~\ref{crit:homomorphism_appendix} is satisfied, where the CNOTs in round 2 have been reversed compared to \cref{tab:contraction_circuits_appendix,fig:general_contracting_stabilizers}. As a result, the round 3 CNOTs and measurements $M_{i}$ are also modified compared to \cref{tab:contraction_circuits_appendix}. The circuit $F_{1}$ is defined by applying the gates below for each $g\in K$ and $h\in K^{c}$, while for $F_{2}$ one chooses $g\in K^{c}$ and $h\in K$.}\label{tab:round_2_reversal}
    \renewcommand{\arraystretch}{1.35}
    \setlength{\tabcolsep}{6pt}
    \begin{tabular}{|c|c|}
    \hline
        \multirow{2}{*}{Round 1}& CNOT$\big(q(L,g),q(R,a_{1}^{-1}b_{3}^{\vphantom{-1}}g)\big),$\\
        & CNOT$\big(q(R,a_{3}^{-1}b_{1}^{\vphantom{-1}}h),q(L,h)\big)$\\\hline
        \multirow{2}{*}{Round 2}& CNOT$\big(q(R,a_{1}^{-1}b_{2}^{\vphantom{-1}}g),q(L,g)\big),$\\
        & CNOT$\big(q(L,h),q(R,a_{2}^{-1}b_{1}^{\vphantom{-1}}h)\big)$\\
        \hline
        \multirow{2}{*}{Round 3}& CNOT$\big(q(R,a_{2}^{-1}b_{2}^{\vphantom{-1}}h),q(L,h)\big)$\\
        &CNOT$\big(q(L,g),q(R,a_{2}^{-1}b_{2}^{\vphantom{-1}}g)\big),$\\\hline
        \multirow{2}{*}{Round 4}& $M_{X}\big(q(R,a_{2}^{-1}b_{2}^{\vphantom{-1}}h)\big)$,\\
        &$M_{Z}\big(q(R,a_{2}^{-1}b_{2}^{\vphantom{-1}}g)\big)$\\\hline
    \end{tabular}
\end{table}

Unlike the modification in \cref{subsec:round_3_reversal}, reversing the direction of the CNOT in round 2 has significant consequences for the performance of the code. We explicitly present the modification in \cref{tab:round_2_reversal}.

First, we must convince ourselves that \cref{tab:round_2_reversal} is a valid contraction circuit, i.e.~the contracting stabilizers are indeed measured by $M_{i}$. Similar to \cref{prop:mid_to_end_propagation}, setting $i=1$ for simplicity, one can show that a (possibly trivial) logical $X$ operator $X(P,Q)$ satisfying $BP+AQ=0$ propagates under $F_{1}$ to
\begin{subequations}\label{eq:mid_to_end_reverse_2}
\begin{multline}
    X(P,Q)\mapsto X\Big(P+\big(b_{2}^{-1}b_{3}^{\vphantom{-1}}P+a_{1}^{\vphantom{-1}}b_{2}^{-1}Q\big)\cap K\\
    +\big(a_{1}^{-1}a_{2}^{\vphantom{-1}}b_{2}^{-1}b_{3}^{\vphantom{-1}}P+a_{3}^{\vphantom{-1}}b_{1}^{-1}Q+a_{2}^{\vphantom{-1}}b_{2}^{-1}Q\big)\cap K^{c},\\
    \big(Q+a_{1}^{-1}b_{3}^{\vphantom{-1}}P\big)\cap \big(a_{1}b_{1}K^{c}\big)\Big),
\end{multline}
while a logical $Z$ operator $Z(P,Q)$ satisfying $A^{-1}P+B^{-1}Q=0$ propagates to
\begin{multline}
    Z(P,Q)\mapsto X\Big(P+\big(a_{2}^{\vphantom{-1}}a_{3}^{-1}P+a_{2}^{\vphantom{-1}}b_{1}^{-1}Q\big)\cap K^{c}\\
    +\big(a_{2}^{\vphantom{-1}}a_{3}^{-1}b_{1}^{\vphantom{-1}}b_{2}^{-1}P+a_{1}^{\vphantom{-1}}b_{3}^{-1}Q+a_{2}^{\vphantom{-1}}b_{2}^{-1}Q\big)\cap K,\\
    \big(Q+a_{3}^{-1}b_{1}^{\vphantom{-1}}P\big)\cap \big(a_{1}b_{1}K\big)\Big).
\end{multline}
\end{subequations}
With \cref{eq:mid_to_end_reverse_2}, it is straightforward to show that a contracting $X$-stabilizer $X(Ag,Bg)$ with $g\in a_{1}K$ propagates to $X(0,b_{2}g)$ and is measured by an $X$-measurement, while a contracting $Z$-stabilizer $Z(B^{-1}g,A^{-1}g$ with $g\in b_{1}K^{c}$ propagates to $Z(0,a_{2}^{-1}g)$. \cref{tab:round_2_reversal} therefore defines a valid morphing protocol.

Next, we consider the stabilizer generators of the new end-cycle codes. Substituting an expanding stabilizer into \cref{eq:mid_to_end_reverse_2} shows that the stabilizer generators of the end-cycle codes have weight at most 11, given by
\begin{subequations}\label{eq:reverse_2_end-cycle_stabilizers}
    \begin{align}
        &X\Big(\big(a_{3}^{\vphantom{-1}}b_{1}^{-1}B+b_{2}^{-1}(a_{1}^{\vphantom{-1}}{+}a_{2}^{\vphantom{-1}})(a_{1}^{-1}(a_{2}^{\vphantom{-1}}{+}a_{3}^{\vphantom{-1}})+b_{1}^{\vphantom{-1}}+b_{2}^{\vphantom{-1}})\big)g,0\Big),\\
        &Z\Big(\big(a_{1}^{\vphantom{-1}}b_{3}^{-1}A^{-1}\nonumber\\
        &\qquad+a_{2}^{\vphantom{-1}}(b_{1}^{-1}{+}b_{2}^{-1})(a_{1}^{-1}{+}a_{2}^{-1}{+}a_{3}^{-1}b_{1}^{\vphantom{-1}}(b_{2}^{-1}{+}b_{3}^{-1}))\big)g,0\Big).
    \end{align}
\end{subequations}

\begin{table*}
    \caption{A summary of the code parameters of the codes from Table I of the main text~\cite{Shaw24Lowering} that have improved performance when the round 2 CNOTs are reversed as in \cref{tab:round_2_reversal}. We only find improvements for two of BB mid-cycle codes from Ref.~\cite{Bravyi24}. For these codes, we list both the code distance and circuit-level distance upper-bound under the morphing protocol from \cref{tab:contraction_circuits_appendix}, as well as under the round 2 CNOT reversed contraction circuits from \cref{tab:round_2_reversal}. Note that the ordering of the elements in $A$ and $B$ are not important for the schedules defined by \cref{tab:contraction_circuits_appendix}, since the choice of $a_{1}$ and $b_{1}$ is determined by the choice of homomorphism, and the end-cycle codes are invariant under the swaps $a_{2}\leftrightarrow a_{3}$ or $b_{2}\leftrightarrow b_{3}$ (see \cref{cor:end-cycle_stabilizers}). Moreover, the distances obtained are the same regardless of the choice of homomorphism. This is \textit{not} the case for the circuits defined by \cref{tab:round_2_reversal}, and we therefore list all the relevant permutations of the $A$ and $B$ that provide an improvement.}
    \label{tab:round_2_reversal_code_parameters}
    \renewcommand{\arraystretch}{1.35}
    \setlength{\tabcolsep}{6pt}
    \centering
\begin{tabular}{|c|c|c|c|c|c|c|c|}
        \hline
        $\ell$, $m$ & $F_{i}$ & $\{a_{1},a_{2},a_{3}\}$ & $\{b_{1},b_{2},b_{3}\}$ & $f$ & $[[\tilde{n},k,\tilde{d}]]$ & $\tilde{d}_{\text{circ}}$ \\\hline
        \multirow{7}{*}{6, 6} & \cref{tab:contraction_circuits_appendix} & $\{x^3,y,y^2\}$ & $\{y^3,x,x^2\}$ &  $f_x,f_y,f_{xy}$ & $[[36,12,3]]$ & $\leq 3$ \\ \cline{2-7}
        & \multirow{6}{*}{\cref{tab:round_2_reversal}} & $\{x^3,y,y^2\}$ & $\{x,x^2,y^3\}$ &  \multirow{2}{*}{$f_x$} & \multirow{6}{*}{$[[36,12,4]]$} & \multirow{6}{*}{$\leq 4$} \\ \cline{3-4}
        && $\{x^3,y^2,y\}$ & $\{x,y^3,x^2\}$ & & & \\ \cline{3-5}
        && $\{y,y^2,x^3\}$ & $\{y^3,x,x^2\}$ &  \multirow{2}{*}{$f_y$} & & \\ \cline{3-4}
        && $\{y,x^3,y^2\}$ & $\{y^3,x^2,x\}$ & & & \\ \cline{3-5}
        && $\{y^2,y,x^3\}$ & $\{x^2,x,y^3\}$ & \multirow{2}{*}{$f_{xy}$} & & \\ \cline{3-4}
        && $\{y^2,x^3,y\}$ & $\{x^2,y^3,x\}$ & & & \\ \hline
        \multirow{5}{*}{12, 6} & \cref{tab:contraction_circuits_appendix} & $\{x^3,y,y^2\}$ & $\{y^3,x,x^2\}$ &  $f_x,f_y,f_{xy}$ & $[[72,12,6]]$ & $\leq 6$ \\ \cline{2-7}
        & \multirow{4}{*}{\cref{tab:round_2_reversal}} & $\{x^3,y^2,y\}$ & $\{x,y^3,x^2\}$ & \multirow{2}{*}{$f_{x}$} & \multirow{2}{*}{$[[72,12,8]]$} & \multirow{2}{*}{$\leq6$}\\ \cline{3-4}
        && $\{x^3,y,y^2\}$ & $\{x,x^2,y^3\}$ &&&\\ \cline{3-7}
        && $\{y,y^2,x^3\}$ & $\{y^3,x,x^2\}$ & \multirow{2}{*}{$f_{y}$} & \multirow{2}{*}{$[[72,12,7]]$} & \multirow{2}{*}{$\leq7$}\\ \cline{3-4} 
        && $\{y,x^3,y^2\}$ & $\{y^3,x^2,x\}$ &&&\\ \hline
    \end{tabular}
\end{table*}

\begin{figure}
    \centering
    \includegraphics[width = \linewidth]{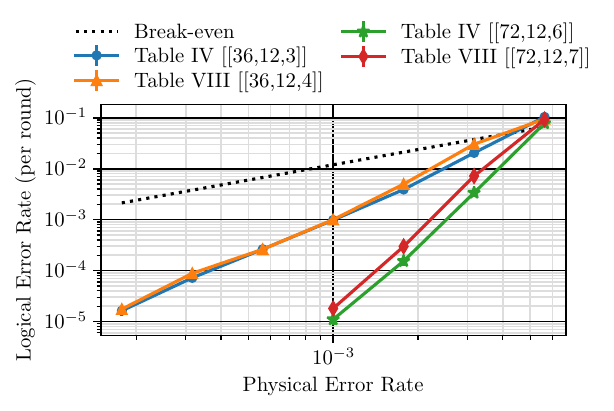}
    \caption{Numerical logical performance of some end-cycle codes using the contraction circuits from \cref{tab:round_2_reversal}. As listed in \cref{tab:round_2_reversal_code_parameters}, these codes have a larger circuit-level distance upper-bound than the corresponding code in Table I of the main text~\cite{Shaw24Lowering} (whose numerical performance is also shown). The numerical performance is evaluated with respect to a uniform circuit-level depolarizing noise model and decoded using BP-OSD.} 
   \label{fig:reverse_2_numerics}
\end{figure}

Clearly, the end-cycle codes defined by the stabilizer generators in~\cref{eq:reverse_2_end-cycle_stabilizers} are not equivalent to the end-cycle codes of \cref{tab:contraction_circuits_appendix} or \ref{tab:round_3_reversal}, and therefore may have different logical properties. We therefore checked the end-cycle code distance and the BP-OSD circuit-level distance upper-bound for all the possible modified end-cycle codes for the morphing protocols listed in Table I of the main text~\cite{Shaw24Lowering}. Codes that had an improved end-cycle distance \textit{or} circuit-level distance upper-bound are listed in \cref{tab:round_2_reversal_code_parameters}. We find that two sets of end-cycle codes have an improved circuit-level distance upper-bound: one with $[[36,12,4]]$ and $d_{\text{circ}}\leq 4$ (improved from 3), and one with $[[72,12,7]]$ and $d_{\text{circ}}\leq 7$ (improved from 6). We numerically evaluate the circuit-level performance of one code representative with these improved parameters in \cref{fig:reverse_2_numerics}. Perhaps surprisingly, we do not see any improvement in the logical performance compared to the original morphing protocols defined in \cref{tab:contraction_circuits_appendix}. We leave it to future work to explore these results in more detail, as well as other potential modifications to the morphing protocol.

\bibliography{my_bib}

\begin{thebibliography}{28}%
\makeatletter
\providecommand \@ifxundefined [1]{%
 \@ifx{#1\undefined}
}%
\providecommand \@ifnum [1]{%
 \ifnum #1\expandafter \@firstoftwo
 \else \expandafter \@secondoftwo
 \fi
}%
\providecommand \@ifx [1]{%
 \ifx #1\expandafter \@firstoftwo
 \else \expandafter \@secondoftwo
 \fi
}%
\providecommand \natexlab [1]{#1}%
\providecommand \enquote  [1]{``#1''}%
\providecommand \bibnamefont  [1]{#1}%
\providecommand \bibfnamefont [1]{#1}%
\providecommand \citenamefont [1]{#1}%
\providecommand \href@noop [0]{\@secondoftwo}%
\providecommand \href [0]{\begingroup \@sanitize@url \@href}%
\providecommand \@href[1]{\@@startlink{#1}\@@href}%
\providecommand \@@href[1]{\endgroup#1\@@endlink}%
\providecommand \@sanitize@url [0]{\catcode `\\12\catcode `\$12\catcode `\&12\catcode `\#12\catcode `\^12\catcode `\_12\catcode `\%12\relax}%
\providecommand \@@startlink[1]{}%
\providecommand \@@endlink[0]{}%
\providecommand \url  [0]{\begingroup\@sanitize@url \@url }%
\providecommand \@url [1]{\endgroup\@href {#1}{\urlprefix }}%
\providecommand \urlprefix  [0]{URL }%
\providecommand \Eprint [0]{\href }%
\providecommand \doibase [0]{https://doi.org/}%
\providecommand \selectlanguage [0]{\@gobble}%
\providecommand \bibinfo  [0]{\@secondoftwo}%
\providecommand \bibfield  [0]{\@secondoftwo}%
\providecommand \translation [1]{[#1]}%
\providecommand \BibitemOpen [0]{}%
\providecommand \bibitemStop [0]{}%
\providecommand \bibitemNoStop [0]{.\EOS\space}%
\providecommand \EOS [0]{\spacefactor3000\relax}%
\providecommand \BibitemShut  [1]{\csname bibitem#1\endcsname}%
\let\auto@bib@innerbib\@empty
\bibitem [{\citenamefont {Bravyi}\ \emph {et~al.}(2024)\citenamefont {Bravyi}, \citenamefont {Cross}, \citenamefont {Gambetta}, \citenamefont {Maslov}, \citenamefont {Rall},\ and\ \citenamefont {Yoder}}]{Bravyi24}%
  \BibitemOpen
  \bibfield  {author} {\bibinfo {author} {\bibfnamefont {S.}~\bibnamefont {Bravyi}}, \bibinfo {author} {\bibfnamefont {A.~W.}\ \bibnamefont {Cross}}, \bibinfo {author} {\bibfnamefont {J.~M.}\ \bibnamefont {Gambetta}}, \bibinfo {author} {\bibfnamefont {D.}~\bibnamefont {Maslov}}, \bibinfo {author} {\bibfnamefont {P.}~\bibnamefont {Rall}},\ and\ \bibinfo {author} {\bibfnamefont {T.~J.}\ \bibnamefont {Yoder}},\ }\href@noop {} {\bibfield  {journal} {\bibinfo  {journal} {Nature}\ }\textbf {\bibinfo {volume} {627}},\ \bibinfo {pages} {778} (\bibinfo {year} {2024})}\BibitemShut {NoStop}%
\bibitem [{\citenamefont {Kitaev}(2003)}]{Kitaev03}%
  \BibitemOpen
  \bibfield  {author} {\bibinfo {author} {\bibfnamefont {A.~Y.}\ \bibnamefont {Kitaev}},\ }\href@noop {} {\bibfield  {journal} {\bibinfo  {journal} {Ann. Phys. (N. Y.)}\ }\textbf {\bibinfo {volume} {303}} (\bibinfo {year} {2003})}\BibitemShut {NoStop}%
\bibitem [{\citenamefont {Bravyi}\ and\ \citenamefont {Kitaev}(1998)}]{bravyi1998}%
  \BibitemOpen
  \bibfield  {author} {\bibinfo {author} {\bibfnamefont {S.~B.}\ \bibnamefont {Bravyi}}\ and\ \bibinfo {author} {\bibfnamefont {A.~Y.}\ \bibnamefont {Kitaev}},\ }\href@noop {} {\bibinfo {title} {Quantum codes on a lattice with boundary}} (\bibinfo {year} {1998}),\ \Eprint {https://arxiv.org/abs/quant-ph/9811052} {arXiv:quant-ph/9811052 [quant-ph]} \BibitemShut {NoStop}%
\bibitem [{\citenamefont {Cleland}(2022)}]{Cleland22}%
  \BibitemOpen
  \bibfield  {author} {\bibinfo {author} {\bibfnamefont {A.}~\bibnamefont {Cleland}},\ }\bibfield  {journal} {\bibinfo  {journal} {SciPost Physics Lecture Notes}\ }\href {https://doi.org/10.21468/scipostphyslectnotes.49} {10.21468/scipostphyslectnotes.49} (\bibinfo {year} {2022})\BibitemShut {NoStop}%
\bibitem [{\citenamefont {Kovalev}\ and\ \citenamefont {Pryadko}(2013)}]{Kovalev13}%
  \BibitemOpen
  \bibfield  {author} {\bibinfo {author} {\bibfnamefont {A.~A.}\ \bibnamefont {Kovalev}}\ and\ \bibinfo {author} {\bibfnamefont {L.~P.}\ \bibnamefont {Pryadko}},\ }\href {https://doi.org/10.1103/physreva.88.012311} {\bibfield  {journal} {\bibinfo  {journal} {Physical Review A}\ }\textbf {\bibinfo {volume} {88}},\ \bibinfo {pages} {012311} (\bibinfo {year} {2013})}\BibitemShut {NoStop}%
\bibitem [{\citenamefont {Lin}\ and\ \citenamefont {Pryadko}(2023)}]{Lin23}%
  \BibitemOpen
  \bibfield  {author} {\bibinfo {author} {\bibfnamefont {H.-K.}\ \bibnamefont {Lin}}\ and\ \bibinfo {author} {\bibfnamefont {L.~P.}\ \bibnamefont {Pryadko}},\ }\href@noop {} {\bibinfo {title} {Quantum two-block group algebra codes}} (\bibinfo {year} {2023}),\ \Eprint {https://arxiv.org/abs/2306.16400} {arXiv:2306.16400 [quant-ph]} \BibitemShut {NoStop}%
\bibitem [{\citenamefont {McEwen}\ \emph {et~al.}(2023)\citenamefont {McEwen}, \citenamefont {Bacon},\ and\ \citenamefont {Gidney}}]{McEwen23}%
  \BibitemOpen
  \bibfield  {author} {\bibinfo {author} {\bibfnamefont {M.}~\bibnamefont {McEwen}}, \bibinfo {author} {\bibfnamefont {D.}~\bibnamefont {Bacon}},\ and\ \bibinfo {author} {\bibfnamefont {C.}~\bibnamefont {Gidney}},\ }\href {https://doi.org/10.22331/q-2023-11-07-1172} {\bibfield  {journal} {\bibinfo  {journal} {Quantum}\ }\textbf {\bibinfo {volume} {7}},\ \bibinfo {pages} {1172} (\bibinfo {year} {2023})}\BibitemShut {NoStop}%
\bibitem [{\citenamefont {Gidney}\ and\ \citenamefont {Jones}(2023)}]{Gidney23}%
  \BibitemOpen
  \bibfield  {author} {\bibinfo {author} {\bibfnamefont {C.}~\bibnamefont {Gidney}}\ and\ \bibinfo {author} {\bibfnamefont {C.}~\bibnamefont {Jones}},\ }\href@noop {} {\bibinfo {title} {New circuits and an open source decoder for the color code}} (\bibinfo {year} {2023}),\ \Eprint {https://arxiv.org/abs/2312.08813} {arXiv:2312.08813 [quant-ph]} \BibitemShut {NoStop}%
\bibitem [{\citenamefont {Vasmer}\ and\ \citenamefont {Kubica}(2022)}]{Vasmer22}%
  \BibitemOpen
  \bibfield  {author} {\bibinfo {author} {\bibfnamefont {M.}~\bibnamefont {Vasmer}}\ and\ \bibinfo {author} {\bibfnamefont {A.}~\bibnamefont {Kubica}},\ }\href {https://doi.org/10.1103/prxquantum.3.030319} {\bibfield  {journal} {\bibinfo  {journal} {PRX Quantum}\ }\textbf {\bibinfo {volume} {3}},\ \bibinfo {pages} {030319} (\bibinfo {year} {2022})}\BibitemShut {NoStop}%
\bibitem [{\citenamefont {Panteleev}\ and\ \citenamefont {Kalachev}(2021)}]{Panteleev21}%
  \BibitemOpen
  \bibfield  {author} {\bibinfo {author} {\bibfnamefont {P.}~\bibnamefont {Panteleev}}\ and\ \bibinfo {author} {\bibfnamefont {G.}~\bibnamefont {Kalachev}},\ }\href {https://doi.org/10.22331/q-2021-11-22-585} {\bibfield  {journal} {\bibinfo  {journal} {Quantum}\ }\textbf {\bibinfo {volume} {5}},\ \bibinfo {pages} {585} (\bibinfo {year} {2021})}\BibitemShut {NoStop}%
\bibitem [{\citenamefont {Roffe}\ \emph {et~al.}(2020)\citenamefont {Roffe}, \citenamefont {White}, \citenamefont {Burton},\ and\ \citenamefont {Campbell}}]{Roffe20}%
  \BibitemOpen
  \bibfield  {author} {\bibinfo {author} {\bibfnamefont {J.}~\bibnamefont {Roffe}}, \bibinfo {author} {\bibfnamefont {D.~R.}\ \bibnamefont {White}}, \bibinfo {author} {\bibfnamefont {S.}~\bibnamefont {Burton}},\ and\ \bibinfo {author} {\bibfnamefont {E.}~\bibnamefont {Campbell}},\ }\href@noop {} {\bibfield  {journal} {\bibinfo  {journal} {Physical Review Research}\ }\textbf {\bibinfo {volume} {2}},\ \bibinfo {pages} {043423} (\bibinfo {year} {2020})}\BibitemShut {NoStop}%
\bibitem [{Note1()}]{Note1}%
  \BibitemOpen
  \bibinfo {note} {Strictly speaking of course the ``mid-way'' point of a parity check circuit is only well-defined if the circuit has an even depth --- which is satisfied for all the new circuits constructed in this Letter.}\BibitemShut {Stop}%
\bibitem [{\citenamefont {Aaronson}\ and\ \citenamefont {Gottesman}(2004)}]{Aaronson04}%
  \BibitemOpen
  \bibfield  {author} {\bibinfo {author} {\bibfnamefont {S.}~\bibnamefont {Aaronson}}\ and\ \bibinfo {author} {\bibfnamefont {D.}~\bibnamefont {Gottesman}},\ }\href@noop {} {\bibfield  {journal} {\bibinfo  {journal} {Physical Review A}\ }\textbf {\bibinfo {volume} {70}},\ \bibinfo {pages} {052328} (\bibinfo {year} {2004})}\BibitemShut {NoStop}%
\bibitem [{sup()}]{supp}%
  \BibitemOpen
  \href@noop {} {}\bibinfo {note} {See Supplemental Material, which includes Refs.~\cite{Morphing_GitHub,Gidney21,Breuckmann_2017,Xu24,Horsman12,Vuillot_2019}.}\BibitemShut {Stop}%
\bibitem [{Note2()}]{Note2}%
  \BibitemOpen
  \bibinfo {note} {Note that different homomorphisms applied to the same mid-cycle BB code can indeed lead to distinct end-cycle codes that are unrelated via the mappings of Ref.~\cite {Lin23}.}\BibitemShut {Stop}%
\bibitem [{\citenamefont {{Gurobi Optimization, LLC}}(2024)}]{Gurobi}%
  \BibitemOpen
  \bibfield  {author} {\bibinfo {author} {\bibnamefont {{Gurobi Optimization, LLC}}},\ }\href {https://www.gurobi.com} {\bibinfo {title} {{Gurobi Optimizer Reference Manual}}} (\bibinfo {year} {2024})\BibitemShut {NoStop}%
\bibitem [{\citenamefont {Landahl}\ \emph {et~al.}(2011)\citenamefont {Landahl}, \citenamefont {Anderson},\ and\ \citenamefont {Rice}}]{Landahl11}%
  \BibitemOpen
  \bibfield  {author} {\bibinfo {author} {\bibfnamefont {A.~J.}\ \bibnamefont {Landahl}}, \bibinfo {author} {\bibfnamefont {J.~T.}\ \bibnamefont {Anderson}},\ and\ \bibinfo {author} {\bibfnamefont {P.~R.}\ \bibnamefont {Rice}},\ }\href@noop {} {\bibinfo {title} {Fault-tolerant quantum computing with color codes}} (\bibinfo {year} {2011}),\ \Eprint {https://arxiv.org/abs/1108.5738} {arXiv:1108.5738 [quant-ph]} \BibitemShut {NoStop}%
\bibitem [{\citenamefont {Cohen}\ \emph {et~al.}(2022)\citenamefont {Cohen}, \citenamefont {Kim}, \citenamefont {Bartlett},\ and\ \citenamefont {Brown}}]{Cohen22}%
  \BibitemOpen
  \bibfield  {author} {\bibinfo {author} {\bibfnamefont {L.~Z.}\ \bibnamefont {Cohen}}, \bibinfo {author} {\bibfnamefont {I.~H.}\ \bibnamefont {Kim}}, \bibinfo {author} {\bibfnamefont {S.~D.}\ \bibnamefont {Bartlett}},\ and\ \bibinfo {author} {\bibfnamefont {B.~J.}\ \bibnamefont {Brown}},\ }\bibfield  {journal} {\bibinfo  {journal} {Science Advances}\ }\textbf {\bibinfo {volume} {8}},\ \href {https://doi.org/10.1126/sciadv.abn1717} {10.1126/sciadv.abn1717} (\bibinfo {year} {2022})\BibitemShut {NoStop}%
\bibitem [{Note3()}]{Note3}%
  \BibitemOpen
  \bibinfo {note} {We leave investigations into the $ZX$-duality to future work, since even for the codes of Ref.~\cite {Bravyi24}, more work must be done to optimize the gate sequence before it can be implemented in practice.}\BibitemShut {Stop}%
\bibitem [{Mor()}]{Morphing_GitHub}%
  \BibitemOpen
  \href@noop {} {}\bibinfo {note} {\url{https://github.com/Mac-Shaw/morphing_qec_circuits}}\BibitemShut {NoStop}%
\bibitem [{\citenamefont {Gidney}(2021)}]{Gidney21}%
  \BibitemOpen
  \bibfield  {author} {\bibinfo {author} {\bibfnamefont {C.}~\bibnamefont {Gidney}},\ }\href {https://doi.org/10.22331/q-2021-07-06-497} {\bibfield  {journal} {\bibinfo  {journal} {{Quantum}}\ }\textbf {\bibinfo {volume} {5}},\ \bibinfo {pages} {497} (\bibinfo {year} {2021})}\BibitemShut {NoStop}%
\bibitem [{Sha()}]{Shaw24Lowering}%
  \BibitemOpen
  \href@noop {} {}\bibinfo {note} {See Main text.}\BibitemShut {Stop}%
\bibitem [{Note4()}]{Note4}%
  \BibitemOpen
  \bibinfo {note} {Of course, it would be fairly simple to reverse the direction of the CNOT in Step 2 of, say, $F(Z,g)$ to avoid this problem. However in our constructions, this modification doesn't provide any benefits, so for simplicity we do not consider it here.}\BibitemShut {Stop}%
\bibitem [{Note5()}]{Note5}%
  \BibitemOpen
  \bibinfo {note} {Note here that the term ``automorphism'' refers to an automorphism of the \protect \textit {stabilizer group} of the code, not the Abelian group $G$ that defines the BB code.}\BibitemShut {Stop}%
\bibitem [{\citenamefont {Breuckmann}\ \emph {et~al.}(2017)\citenamefont {Breuckmann}, \citenamefont {Vuillot}, \citenamefont {Campbell}, \citenamefont {Krishna},\ and\ \citenamefont {Terhal}}]{Breuckmann_2017}%
  \BibitemOpen
  \bibfield  {author} {\bibinfo {author} {\bibfnamefont {N.~P.}\ \bibnamefont {Breuckmann}}, \bibinfo {author} {\bibfnamefont {C.}~\bibnamefont {Vuillot}}, \bibinfo {author} {\bibfnamefont {E.}~\bibnamefont {Campbell}}, \bibinfo {author} {\bibfnamefont {A.}~\bibnamefont {Krishna}},\ and\ \bibinfo {author} {\bibfnamefont {B.~M.}\ \bibnamefont {Terhal}},\ }\href {https://doi.org/10.1088/2058-9565/aa7d3b} {\bibfield  {journal} {\bibinfo  {journal} {Quantum Science and Technology}\ }\textbf {\bibinfo {volume} {2}},\ \bibinfo {pages} {035007} (\bibinfo {year} {2017})}\BibitemShut {NoStop}%
\bibitem [{\citenamefont {Xu}\ \emph {et~al.}(2024)\citenamefont {Xu}, \citenamefont {Bonilla~Ataides}, \citenamefont {Pattison}, \citenamefont {Raveendran}, \citenamefont {Bluvstein}, \citenamefont {Wurtz}, \citenamefont {Vasi{\'c}}, \citenamefont {Lukin}, \citenamefont {Jiang},\ and\ \citenamefont {Zhou}}]{Xu24}%
  \BibitemOpen
  \bibfield  {author} {\bibinfo {author} {\bibfnamefont {Q.}~\bibnamefont {Xu}}, \bibinfo {author} {\bibfnamefont {J.~P.}\ \bibnamefont {Bonilla~Ataides}}, \bibinfo {author} {\bibfnamefont {C.~A.}\ \bibnamefont {Pattison}}, \bibinfo {author} {\bibfnamefont {N.}~\bibnamefont {Raveendran}}, \bibinfo {author} {\bibfnamefont {D.}~\bibnamefont {Bluvstein}}, \bibinfo {author} {\bibfnamefont {J.}~\bibnamefont {Wurtz}}, \bibinfo {author} {\bibfnamefont {B.}~\bibnamefont {Vasi{\'c}}}, \bibinfo {author} {\bibfnamefont {M.~D.}\ \bibnamefont {Lukin}}, \bibinfo {author} {\bibfnamefont {L.}~\bibnamefont {Jiang}},\ and\ \bibinfo {author} {\bibfnamefont {H.}~\bibnamefont {Zhou}},\ }\href@noop {} {\bibfield  {journal} {\bibinfo  {journal} {Nature Physics}\ ,\ \bibinfo {pages} {1}} (\bibinfo {year} {2024})}\BibitemShut {NoStop}%
\bibitem [{\citenamefont {Horsman}\ \emph {et~al.}(2012)\citenamefont {Horsman}, \citenamefont {Fowler}, \citenamefont {Devitt},\ and\ \citenamefont {Van~Meter}}]{Horsman12}%
  \BibitemOpen
  \bibfield  {author} {\bibinfo {author} {\bibfnamefont {D.}~\bibnamefont {Horsman}}, \bibinfo {author} {\bibfnamefont {A.~G.}\ \bibnamefont {Fowler}}, \bibinfo {author} {\bibfnamefont {S.}~\bibnamefont {Devitt}},\ and\ \bibinfo {author} {\bibfnamefont {R.}~\bibnamefont {Van~Meter}},\ }\href@noop {} {\bibfield  {journal} {\bibinfo  {journal} {New Journal of Physics}\ }\textbf {\bibinfo {volume} {14}},\ \bibinfo {pages} {123011} (\bibinfo {year} {2012})}\BibitemShut {NoStop}%
\bibitem [{\citenamefont {Vuillot}\ \emph {et~al.}(2019)\citenamefont {Vuillot}, \citenamefont {Lao}, \citenamefont {Criger}, \citenamefont {García~Almudéver}, \citenamefont {Bertels},\ and\ \citenamefont {Terhal}}]{Vuillot_2019}%
  \BibitemOpen
  \bibfield  {author} {\bibinfo {author} {\bibfnamefont {C.}~\bibnamefont {Vuillot}}, \bibinfo {author} {\bibfnamefont {L.}~\bibnamefont {Lao}}, \bibinfo {author} {\bibfnamefont {B.}~\bibnamefont {Criger}}, \bibinfo {author} {\bibfnamefont {C.}~\bibnamefont {García~Almudéver}}, \bibinfo {author} {\bibfnamefont {K.}~\bibnamefont {Bertels}},\ and\ \bibinfo {author} {\bibfnamefont {B.~M.}\ \bibnamefont {Terhal}},\ }\href {https://doi.org/10.1088/1367-2630/ab0199} {\bibfield  {journal} {\bibinfo  {journal} {New Journal of Physics}\ }\textbf {\bibinfo {volume} {21}},\ \bibinfo {pages} {033028} (\bibinfo {year} {2019})}\BibitemShut {NoStop}%
\end{thebibliography}%
\end{document}